\documentclass[aps,pra,twocolumn,superscriptaddress,longbibliography,nofootinbib]{revtex4-1}
\usepackage[normalem]{ulem}
\usepackage{float}
\usepackage{graphicx}  
\usepackage{dcolumn}          
\usepackage{amssymb}
\usepackage{appendix}
\usepackage{physics}   
\usepackage{mathtools}
\usepackage{esvect}
\usepackage{wrapfig}
\usepackage{amsthm}
\usepackage{verbatim}
\usepackage{bbm}
\usepackage[colorlinks=true,linkcolor=blue,citecolor=red,plainpages=false,pdfpagelabels]{hyperref}
\usepackage{subcaption}
\usepackage{xcolor} 
\usepackage{verbatim} 
\usepackage[mathscr]{euscript}
\usepackage{makecell}

\usepackage{algorithm}
\usepackage[noend]{algpseudocode}
\usepackage{setspace}
\usepackage{enumitem}

\algblock[Class]{Class}{End}

\def\Tr{\operatorname{Tr}}

\def\({\left(}
\def\){\right)}
\def\[{\left[}
\def\]{\right]}

\let\emptyset\varnothing

\DeclarePairedDelimiter\floor{\lfloor}{\rfloor}

\newtheorem{theorem}{Theorem}
\newtheorem{note}{Note}

\newtheorem{assumption}{Assumption}

\newtheorem{definition}{Definition}
\newtheorem{example}{Example}

\newtheorem{observation}{Observation}

\newtheorem{proposition}{Proposition}
\newtheorem{remark}{Remark}

\def\>{\rangle}
\def\<{\langle}

\algnewcommand{\Inputs}[1]{
  \State \textbf{Inputs:}
  \Statex \hspace*{\algorithmicindent}\parbox[t]{.8\linewidth}{\raggedright #1}
}
\algnewcommand{\Initialize}[1]{
  \State \textbf{Initialize:}
  \Statex \hspace*{\algorithmicindent}\parbox[t]{.8\linewidth}{\raggedright #1}
}

\begin{document}

\widetext

\title{Practical limitations on robustness and scalability of quantum Internet}

\author{Abhishek Sadhu}
\email{a.sadhu@bham.ac.uk}
\affiliation{Raman Research Institute, Bengaluru, Karnataka 560080, India}
\affiliation{Center for Security, Theory and Algorithmic Research, International Institute of Information Technology, Hyderabad, Gachibowli, Telangana 500032, India}
\affiliation{School of Computer Science, University of Birmingham, Edgbaston, Birmingham B15 2TT, United Kingdom}

\author{Meghana Ayyala Somayajula}
\affiliation{Department of Computer Science Engineering, G Narayanamma Institute of Technology and Science, Hyderabad, Telangana 500104, India}

\author{Karol Horodecki}\email{karol.horodecki@ug.edu.pl}
\affiliation{Institute of Informatics, National Quantum Information Centre, Faculty of Mathematics, Physics and Informatics, University of Gda\'nsk, Wita Stwosza 57, 80-308 Gda\'nsk, Poland}
\affiliation{International Centre for Theory of Quantum Technologies, University of Gda\'nsk, Wita Stwosza 63, 80-308 Gda\'nsk, Poland}

\author{Siddhartha Das}\email{das.seed@iiit.ac.in}
\affiliation{Center for Security, Theory and Algorithmic Research, International Institute of Information Technology, Hyderabad, Gachibowli, Telangana 500032, India}

\date{\today}
\begin{abstract}
The realization of a desirably functional quantum Internet is hindered by fundamental and practical challenges such as high loss during transmission of quantum systems, decoherence due to interaction with the environment, fragility of quantum states under measurement, etc. We study the implications of these constraints by analyzing the limitations on the scaling and robustness of the quantum Internet. Considering quantum networks, we present practical bottlenecks for secure communication, delegated computing, and resource distribution among end nodes. In particular, we show that the device-independent quantum key distribution (DI-QKD) network of isotropic states with parameter $p$ has a finite diameter depending on $p$. Motivated by the power of abstraction in graph theory (in association with quantum information theory), that allows us to consider the topology of a network of users spatially separated over long distances around the globe or separated as nodes within the small space over a processor or circuit within a same framework, we consider graph-theoretic quantifiers to assess network robustness and provide critical values of communication lines for viable communication over quantum Internet. We consider some quantum networks of practical interest, ranging from satellite-based networks connecting far-off spatial locations to currently available quantum processor architectures within computers, and analyze their robustness to perform quantum information processing tasks. Some of these tasks form primitives for delegated quantum computing, e.g., entanglement distribution and quantum teleportation, that are desirable for all scale of quantum networks, be it quantum processors or quantum Internet.
\end{abstract}

\maketitle

\tableofcontents
\section{Introduction}
The quantum Internet~\cite{dowling2003quantum,kimble2008quantum,wehner2018quantum} is a network of users who are enabled to perform desired quantum information processing and computing tasks among them. Some of the important communication and computing tasks that we envision to be possible in full-fledged quantum Internet are cryptographic tasks against quantum adversaries of varying degree, quantum computation, distributed quantum computing, delegated quantum computing, quantum sensing, quantum clock synchronization, quantum communication, superdense coding, quantum teleportation, randomness generation and distribution, and others~\cite{arnon2018practical,DW19,XMZ+20,D18thesis,LHZ+23,DGP24,BGD25}. Some of the crucial primitives for many of the quantum information and computing tasks are the distribution of resources like entanglement~\cite{schrodinger1935discussion,horodecki2009quantum,das2018robust,HWS+25}, nonlocal quantum correlations~\cite{bell1964einstein,HSD15,Luo22,BBN+23}, or quantum coherence~\cite{SSD+15}, etc.

Recent technological advancements have made it possible to realise quantum networks with fewer users at small distances~\cite{valivarthi2020teleportation,wang2014field,stucki2011long}, however, the short lifetime of physical qubits (in some coherent states or quantum correlations), the high loss experienced during the distribution of quantum systems across different channels~\cite{azuma2015all,das2021universal}, and the use of imperfect measurement devices~\cite{len2022quantum,sadhu2023testing} limits the realization of quantum networks for practical purposes. It is important to study the limitations of the quantum Internet, be it network connecting nodes in a quantum processor (or quantum computer) or network connecting users in different spatial locations to allow quantum information processing tasks among them. In this work, we use tools and techniques from graph (network) theory and quantum information theory to shed light on the limitations on the robustness and scalability of the quantum Internet.

Unlike classical signals, the losses associated with the transmission of quantum states between two parties via a lossy channel cannot be reduced by using amplifiers since the measurement will disturb the system~\cite{fuchs1996quantum} and unknown quantum states cannot be cloned~\cite{wootters1982single}. An alternative approach is to use quantum repeaters~\cite{briegel1998quantum,zhan2023performance} to enable entanglement distribution~\cite{das2018robust,wehner2018quantum,khatri2019practical,SvL22} and perform cryptographic tasks over lossy channels. The quantum repeater-based protocols use entanglement swapping~\cite{zukowski1993event,SSB+05} and optionally entanglement purification~\cite{bennett1996purification,BVK99,HHS+21} at intermediate points along the quantum channel. Advances in satellite-based quantum communication networks~\cite{ecker2022advances,wang2021exploiting}, improving system designs~\cite{sliwa2003conditional,you2022quantum}, and developing post-processing techniques~\cite{BFB+92, SR08,YLL+20} are considerable efforts towards overcoming constraints on realizing the quantum Internet. With the assumption that some repeater-based schemes may improve the transmissivity of quantum channels connecting nodes in a network, we highlight instances to illustrate the limitations of current technology in implementing scalable quantum communication and \textcolor{black}{device-independent quantum key distribution (DI-
QKD)} networks.

If we focus on the topology of the quantum Internet, certain important questions arise: What is the effective structure of the underlying network to enable different parties to perform information processing tasks? Once a network structure is established, a significant concern is regarding its susceptibility to node and edge failures. A resilient network should possess the ability to carry out different information-processing tasks even if a substantial portion of its nodes becomes inactive. Consequently, it is desirable that the realization of the quantum Internet withstands the failure of a fraction of its nodes and links. This naturally leads to the question of how to evaluate the resilience of the quantum Internet and identify the nodes that are most crucial to its overall functioning. Having designed robust networks to enable different parties to perform desired information processing tasks, it is important to ask what strategies can be employed to optimize the flow of resources at the buffer nodes\footnote{A buffer node in a network is a node that temporarily stores resources from different input channels and distributes them to different output channels upon request.}.

\subsection{Motivation and our contribution}
The power of abstraction in graph (network) theory~\cite{west2000graph,tutte2001graph} that provides a general formalism to analyze networks without getting into specific details of implementation~\cite{barnes1983graph}, motivates the use of graph-theoretic tools in analyzing the scalability and robustness of a global scale quantum Internet. Focusing on the robustness of networks, we present figures of merit for comparing different network topologies. We also apply ideas from percolation theory~\cite{BH57,SA18,BS05} to discuss the robustness of networks formed when performing a class of information processing tasks over any lattice network (sufficiently large graph).

One might find the level of abstraction quite high. However it is justified by the fact that we consider a system be it network or processor as the one {\it after} any technological improvement involving repeater techniques or fault tolerant quantum computing. No matter how they are defined, resulting output state will never be ideal but can only at best be close to ideal. For example, in some cases of practical interests, these final states can be modelled (in the case of links) as an isotropic state: ideal maximally entangled state mixed with the maximally mixed state (a useless noisy state) with some non-zero probability.

The structure of the quantum Internet will depend on the information-processing tasks that the network needs to perform. The building blocks for the structure of any network are the elementary links connecting two nodes of the network. Hence, it is important to assess the limitations at the elementary link scale. To do this, we present the critical success probability of the elementary links for performing some desired information processing tasks. Extending to a more general repeater-based network with memory, we present a trade-off between the channel length and time spent at the nodes such that the state shared by the end nodes remains useful for various tasks. 
\begin{figure}
    \includegraphics[scale=0.25]{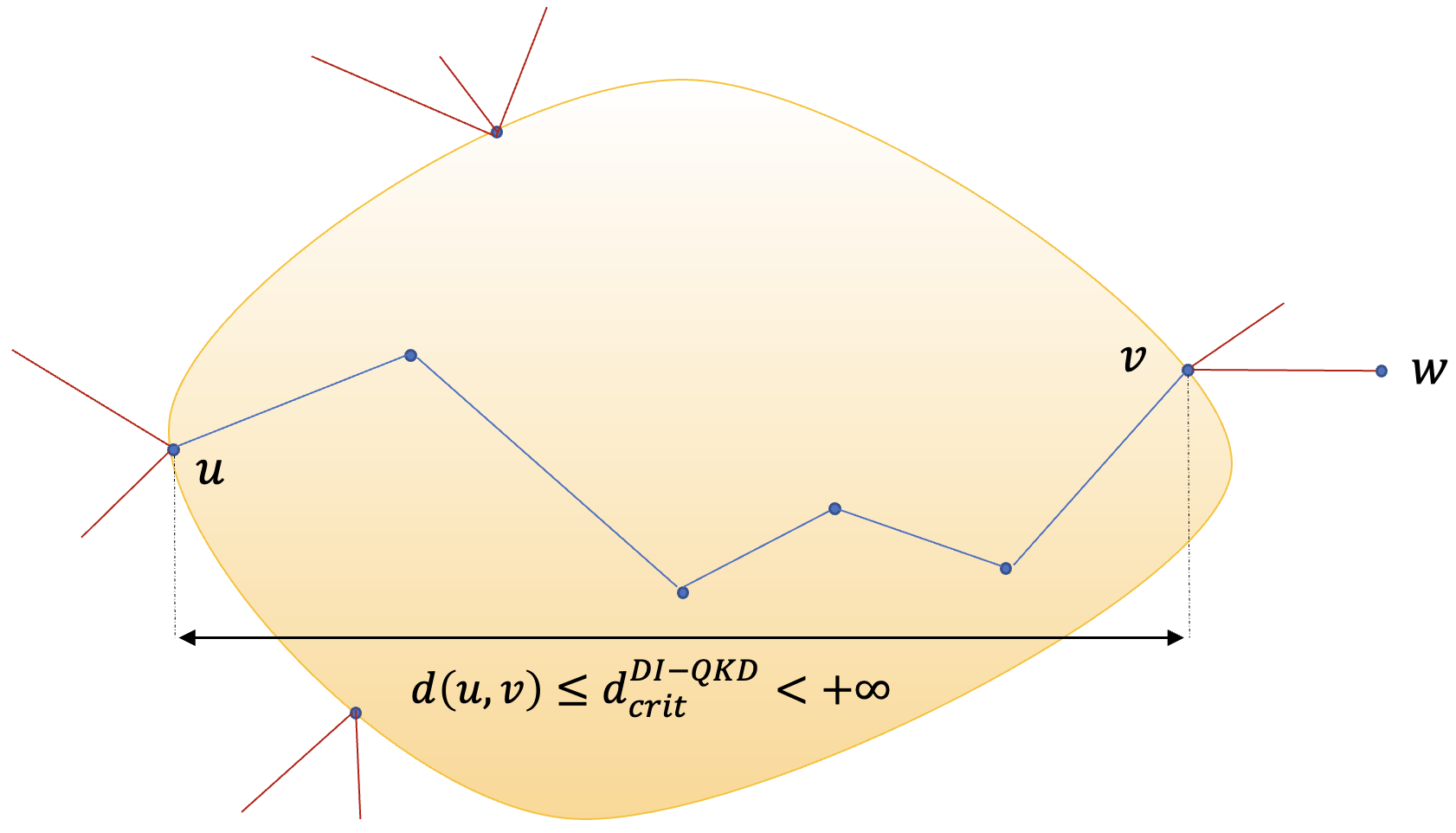}
    \caption{A network of point-to-point connections made of isotropic states is depicted. Any two nodes can be made directly connected via entanglement swapping to generate DI-QKD key along connecting path if and only if they are in a finite distance less than a critical one $\leq d_{crit}^{DI-QKD}$.
    E.g. nodes $u$ and $w$ cannot be connected as they are too far. Only in the yellow area any pair of vertices can be made directly connected in device-independent way.
    In graph-theoretical language, the diameter of any fully connected subgraph is finite, bounded by $d_{crit}^{DI-QKD}$. The same fact limits distance of connecting nodes with pure entanglement needed for faithful teleportation which enables e.g. delegated quantum computing. (Color online)}
    \label{fig:finite_diemeter}
\end{figure}

There is a strong motivation to enable remote places to securely access quantum processors via the quantum Internet and perform delegated quantum computing~\cite{childs01,GLM08,BFK09} over the quantum Internet as it is less resource extensive on the individual user. To illustrate this, we may assume an instance where a user with limited computing resources in Bangalore (in India) requests to securely access the IBM Quantum Hub at Pozna\'n (in Poland)~\cite{IBMPol} via the quantum Internet. Enabling secure access requires the sharing of entangled states among distant nodes~\cite{das2018robust,halder2022circulating,inesta2023optimal,lee2022quantum,GI22} along with performing secure cryptographic tasks against adversaries. We present limitations involved in implementing these tasks. 

In particular, we present limitations on using isotropic state~\cite{PhysRevA.59.4206} for distilling secret keys via DI-QKD protocols~\cite{PhysRevLett.67.661,acin2007device} over the quantum Internet. We provide upper bounds on the number of elementary links between the end nodes for performing secure communication and information-processing tasks (see Fig. \ref{fig:finite_diemeter}). We show that a region of a (possibly infinite) network has a bounded size (graph-theoretically bounded diameter). Any two nodes can be connected by entanglement swapping to produce a device-independent key by so-called standard protocols only if the nodes are closer to each other than the critical distance $d_{crit}^{DI-QKD}$, which is always finite. It is a consequence of the recently discovered fact that 
even states exhibiting quantum non-locality may have zero device-independent key secure against quantum adversary \cite{farkas2021bell}. \textit{Our result can be phrased as no-percolation for DI-QKD networks in graph-theoretic language.}

Let us recall here that we consider only the standard Bell measurements (see Eq.~\eqref{eq:stdBellMeasurement}) in the nodes since any technique improving the quality of the links, such as entanglement distillation or quantum error correction, is assumed to be performed beforehand. 

Assuming there exists some scheme that can mitigate losses and improve transmittance over a quantum channel (see Assumption~\ref{assumption:transmittance}), we present limitations on the scalability of networks for performing quantum communication and implementing DI-QKD protocols. In particular, we present limitations on performing DI-QKD between two end nodes at a continental scale of distances and connected by repeater-based elementary links of metropolitan scale (see Example~\ref{example:DIQKD}). Considering performing quantum communication over a lattice with optical fiber-based elementary links, we present limitations on its scalability (see Observation~\ref{obs:graphDiameter}). These illustrate how far is the current technology from designing quantum networks for information processing tasks. 

For long-distance communication over the continental scale of (order of $1000$ km) distances, the losses in transmission of optical signals over free space are significantly lower as compared to an optical fiber (see Fig.~\ref{fig:transmissionFiberSpace}) making it better suited for distributing entangled states between far-off places. In this work, we present practical bottlenecks in realising such networks (see Fig.~\ref{fig:sampleGraph99}~and~\ref{fig:sampleGraph15}). As a potential application of the satellite-based network, we present the entanglement yield (see Fig.~\ref{fig:sampleGraph115}) and figures of merit of a global mesh network having $3462$ major airports across the globe as nodes and $25482$ airplane routes as edges connecting the major airports in the world. For such an airport network, we present the entanglement yield considering currently available and desirable technology available in the future. Looking at short-distance communication over the regional scale (order of $100$ km) we present an instance of secure communication between a central agency and end parties. It may be desirable here for the central agency to prevent direct communication between the end parties. We consider a fiber-based star network for this task and present bottlenecks in its implementation (see Fig.~\ref{fig:upperBoundAlphaRepeaterRelay}).

\subsubsection{Summary of the main results} \label{sec:summary}
The main theme of this paper is to analyse how quantum information processing tasks over a network (Internet) get limited by intrinsic noise due to the interaction of quantum information carriers with the environment. We provide several figures of merits from network theory to compare different quantum network architectures and topologies. 

\begin{itemize}[leftmargin=*]
  \item The first main result is that a quantum network with a regular lattice structure will not form percolation in practice for quantum information processing tasks including DI-QKD, sharing entanglement, violating Bell-CHSH inequality, and performing teleportation protocol. We first consider a network with $n$ repeater stations and neighbouring nodes sharing qubit isotropic states (see Eq.~\ref{eq:iso}) given by
\begin{eqnarray}
    \rho^I_{AB}\bigg(p(\lambda),2\bigg) \coloneqq \lambda \Psi^+_{AB} + (1 - \lambda) \frac{\mathbbm{1}_{AB}}{4}.
\end{eqnarray}
In such a repeater chain of $n$ nodes with the standard CHSH-based protocol for DI-QKD, we show that there exists a range of isotropic state visibility $\lambda \in \bigg(\gamma_{\text{crit}}^{\theta},(\gamma_{\text{crit}}^{\theta}/q^n)^{1/(n+1)}\bigg)$ such that the device-independent secret key distillation rate is zero among the end nodes even though neighbouring nodes can distill device-independent secret keys at non-zero rate (see Observation~\ref{obv:isotropicQKD}). We further show that the key rates from tilted CHSH inequality~\cite{acin2012randomness} based and modified standard CHSH inequality~\cite{schwonnek2021device} based DI-QKD protocols are non-zero for $n < \floor{\frac{\log (\lambda/\gamma^{\theta}_{\text{crit}})}{\log (1/(q\lambda))}}$, where $\gamma^{\theta}_{\text{crit}}$ is the critical visibility for CHSH-based DI-QKD protocols without repeaters (see Eq.~\eqref{eq:boundKey}). We thus observe that the critical success probability of performing information processing tasks over a network limits its diameter (see Fig.~\ref{fig:finite_diemeter}). 

We further define a network $\mathscr{N}_{\text{crit}}(G(\mathbb{V},\mathbb{E}))$ as \textit{critically large} for  $\mathrm{Task}_\ast$ if for all $e_{ij} \in \mathbb{E}$ the success probability of transmitting resource $\chi$ denoted by $p_{ij}$ such that $p_{ij} \leq c $ where $c \in (0,1)$ and there are at least two nodes $x_0,y_0 \in \mathbb{V}$ which are at a distance  
\begin{equation}
    dist(x_0,y_0)\geq \left\lceil \frac{\log p_*}{\log c} \right\rceil +1,
\end{equation}
where $p_{\ast}$ is the critical probability for successful transmission of resource $\chi$. Assuming that any two nodes $v_i,v_j \in \mathbb{V}$ can share $\chi$ over network path ${\cal P}(v_i,v_j)$ having success probability $s({\cal P}(v_i,v_j)) \geq p_\ast >0$, we show in Proposition~\ref{theorem:upperLimitSharing} that there exists a threshold distance $n_0$ such that two vertices $v_i$ and $v_j$ within that distance satisfies bound:
    \begin{eqnarray} \label{results:finiteSize}
    \exists_{n_0 \in \mathbb{N}}~\forall_{\{v_i,v_j\} \in \mathbb{V}} \,&&\bigg[dist(v_i,v_j)\geq n_0 \nonumber \\
    && \Rightarrow \forall_{{\cal P}(v_i,v_j)}\,\,s({\cal P}(v_i,v_j))< p_\ast\bigg].
\end{eqnarray}
In other words, there will be at least two nodes in $\mathscr{N}_{\text{crit}}(G(\mathbb{V},\mathbb{E}))$ that cannot perform $\mathrm{Task}_\ast$ by sharing $\chi$. It then follows from Eq.~\eqref{results:finiteSize} that there will be at least two finite subgraphs in a lattice graph whose distance is greater than critical threshold for inter-subgraph nodes to remain connected to perform tasks such as quantum communications and DI-QKD (see Theorem~\ref{lem:percolation}). 
\item The second main result is that in a fiber-based repeater network with quantum memories, there is a trade-off between the critical fiber length and critical storage time in quantum memories such that the state shared by the end nodes remain useful for information processing tasks. In a network scenario, it may be of interest to share a maximally entangled state $\Psi^+$ between the end nodes as a primitive for some information processing tasks such as secure communication or DI-QKD. These tasks require sharing $\Psi^+$ with a success probability greater than a critical threshold $p_\ast$. For such tasks, we first consider an entanglement swapping-based repeater network with neighbouring nodes connected by optical fibers (having channel loss parameter $\alpha$) and each network node storing their share of qubits in a lossy quantum memory (having loss parameter $\beta$). For such a network, we present in Sec.~\ref{sec:criticalParametersForNetworks}, a trade-off between the critical fiber length $(l_{cr})$ and critical storage time $(t_{cr})$ in quantum memories as 
\begin{equation}\label{eq:critLengthTimeMultipleRepeaterRelay1}
    \alpha~l_{cr} + \beta~t_{cr} < \frac{1}{2r} \ln \bigg( \frac{q^r \eta_s^{r+1}}{p_\ast} \bigg).
\end{equation}
In Eq.~\eqref{eq:critLengthTimeMultipleRepeaterRelay1}, $r$ denotes the number of repeater stations between the end nodes, $q$ denotes the success probability of standard Bell measurement at intermediate stations, $\eta_s$ denotes the efficiency of the intermediate quantum sources and $p_\ast$ denotes the critical success probability with which the end nodes require to share $\Psi^+$ for performing desired information processing task. 

Assuming that there exists repeater-based quantum communication or quantum key distribution schemes that can mitigate transmission loss over qubit erasure channel and improve the rate of communication from $\eta$ to $\eta^{1/f}$ with $f \geq 1$ (see Assumption~\ref{assumption:transmittance}), we show in Sec.~\ref{sec:networkTopology} an $f-$fold advantage in the critical channel length $l_{cr}$ and critical storage time $t_{cr}$ given by
\begin{equation}
    \alpha~l_{cr} + \beta~t_{cr} < -f \ln p_\ast,
\end{equation}
when required to send ebits or private bits~\cite{HHHO05} over the qubit erasure channel at rate greater than $p_\ast$. The case of $f = 1$ has been shown in ~\cite{das2021universal} for measurement-device-independent quantum key distribution (see \cite{pirandola2017fundamental,WTB16} for quantum key distribution over point-to-point channel) and the case of $f = 2$ has been shown in \cite{lucamarini2018overcoming,wang2018twin,ma2018phase,zeng2020symmetry,XWL+23} for twin-field quantum key distribution and asynchronous measurement-device-independent quantum key distribution~\cite{xie2022breaking,ZLX+23}.
\item For the third main result, we present the entanglement yield of a two-layered satellite-based network distributing bipartite entangled pairs between far-off cities. Such a network is composed of global and local scale networks (see Sec.~\ref{sec:entanglementDistributionCities}).   
The global-scale satellite network distributes entangled pairs between pairs of ground stations across different cities, while the local-scale optical networks transmit the entangled pairs from the ground stations to different localities via optical fibers (loss parameter $\alpha$). The satellite-satellite links are modelled as qubit erasure channels (having channel parameter $\eta_e$), while the satellite-ground station links are modelled as qubit thermal channels (having channel parameters $\eta_g$ and $\kappa_g$). We assume that the ground stations have access to imperfect quantum memories, and in the memory, a stored quantum state evolves via a depolarising channel (having channel parameter $p$). The entanglement yield of the network is given by 
\begin{eqnarray} 
\xi_{\text{avg}} &&=  \mathrm{e}^{-\alpha(l_B + l_M)} \left(\eta_e^2\right)^{n-1} \eta_s^{n - 1} q^{n-1} \nonumber \\
&&\bigg[(1-p)^{2 s}-\frac{1}{4} (p-2) p \left((s-1) (1-p)^{2 (s-1)}+1\right)\bigg] \nonumber \\
&&\qquad \qquad \bigg[\kappa_g(\kappa_g-1)(\eta_g-1)^2+\frac{1}{2}(1+\eta_g^2)\bigg] 
\end{eqnarray}
where $n$ denotes the number of satellite-satellite links, $\eta_s$ denotes the efficiency of the quantum sources in the satellite network, $s$ denotes the number of time steps that the shared quantum state is stored in the quantum memory and the distances of the end stations from the nearest ground station is denoted by $l_B$ and $l_M$ (see Eq.~\eqref{eq:avgYieldSatelliteNetwork}). 

\item For the fourth main result, we consider graph-theoretic parameters of networks to estimate the robustness of networks. Motivated by degree centrality~\cite{freeman1979centrality} and betweeness centrality~\cite{freeman1977set} of graphs, we introduce \textit{link sparsity} (see Definition~\ref{def:linkSparsity}) and \textit{node connection strength} (see Definition~\ref{def:connectionStrength}) as parameters to compare robustness of different network topologies. Ideally, a robust network will have low values of link sparsity and high values of node connection strength. Furthermore, we present \textit{critical parameter} of a node (see Definition~\ref{def:criticalParameter}) as a measure to identify the critical nodes of a given network. The critical nodes have high value of the critical parameter and are essential for the proper functioning of the network. If one of these nodes are removed, it will reduce the success probability of sharing resources between pairs of end nodes for performing different information processing tasks. 

Considering the currently available quantum processors by Google~\cite{Google}, IBM~\cite{IBM,BGC23} and Rigetti~\cite{Rigetti} as real-world instances of graphical networks (see Sec.~\ref{sec:processorNetworks} ), we observe that the 54-qubit square topology of the Sycamore processor developed by Google~\cite{Google} has the highest node connection strength (see Fig.~\ref{fig:totalConnectionStrengthProcessor}) and the lowest link sparsity (see Eq.~\eqref{eq:sparsityTopology}). With the possibility of having a 1024-qubit quantum processor in the future, we extend the 54-qubit Sycamore processor layout to include 1024 qubits in a $32 \times 32$ layout. 
\begin{figure}
    \centering
    \includegraphics[scale=0.9]{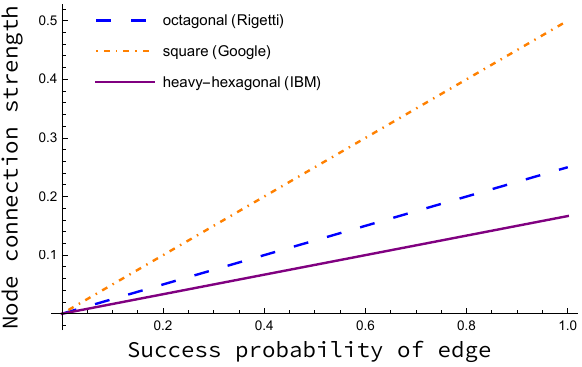}
    \caption{In this figure, we consider three processor designs by (a) Rigetti~\cite{Rigetti} (octagonal lattice), (b) Google~\cite{Google} (square lattice), and (c) IBM~\cite{IBM,BGC23} (heavy-hexagonal lattice) as quantum networks and plot the connection strength of the nodes (see Eq.~\eqref{eq:nodeConnectionStrength}) as a function of success probability of edge for non-cooperative strategy. It is assumed that the network have uniform distribution of edge success probability. (Color online)}
    \label{fig:totalConnectionStrengthProcessor}
\end{figure}
The link sparsity of the 1024-node square network is $0.9962$. Assuming that the edges of the square network have uniform success probability, we present the connection strength of different network nodes. For a $4\times 4$ slice of the 1024-node square network (see Fig.~\ref{SquareLattice}), we observe that the inner nodes have the highest value of the critical parameter. 
\end{itemize}

\subsubsection{Structure of the paper}
The structure of the paper is as follows. We introduce
notations and review basic concepts and standard definitions in Sec.~\ref{sec:preliminaries}. In Sec.~\ref{sec:limitationNetwork}, we present limitations on using isotropic states in networks for device-independent secret key distillation. We present an upper bound on the number of elementary links between the end nodes such that the state shared by them remains entangled and is useful for different information-processing tasks. In Sec.~\ref{sec:GraphTheoryAnalysis}, we analyze the robustness and scalability of the quantum Internet using graph (network) theoretic tools. Specifically, we present a graph theoretic framework for networks performing various tasks in Sec.~\ref{sec:GraphTheoryFramework} and present conditions for no percolation in a lattice network. In Sec.~\ref{sec:criticalParametersForNetworks}, we present the critical success probability of elementary links and critical length scales for various tasks over a repeater-based network. In Sec.~\ref{sec:networkTopology}, we present limitations on the scalability of networks for quantum communication assuming some hypothetical scheme can improve the transmission of channels connecting the nodes of the network. In Sec.~\ref{sec:robustnessMeasure} we present figures of merit to compare robustness of network topologies and observe them for different networks in Sec.~\ref{sec:robustnessComparison}. In Sec.~\ref{sec:critNodesDef} we present measures to identify the critical nodes in a given network. In Sec.~\ref{sec:realWorldInstances}, we present limitations on the potential uses of the quantum Internet for real-world applications. In particular, we present practical bottlenecks in the distribution of entangled states between far-off cities using a satellite-based network in Sec.~\ref{sec:entanglementDistributionCities}. Considering different quantum processor architectures as networks, we observe their figures of merit in Sec.~\ref{sec:processorNetworks}. In Sec.~\ref{sec:networkPropositions}, we first consider the task of connecting major airports across the globe via the satellite-based network and present the yield of the network. Then we present practical bottlenecks in connecting a central agency to the end parties via a star-based network. In Sec.~\ref{sec:algorithms}, we
present algorithms that are primitives for implementing
different information-processing tasks using the quantum
Internet. In particular, we consider the tasks of finding
the shortest network path between a pair of nodes, re-
source allocation at a node with multiple input and out-
put channels and constructing the network structure for
connecting different end nodes. We provide concluding remarks in Sec.~\ref{sec:discussion}.

\section{Preliminaries} \label{sec:preliminaries}
In this section, we introduce notations and review basic concepts and standard definitions that are used frequently in the paper. We consider quantum systems associated with the separable Hilbert spaces. The Hilbert space of a quantum system $A$ and a composite system $AB$ are denoted as $\mathscr{H}_A$ and $\mathscr{H}_{AB}\coloneqq \mathscr{H}_A\otimes\mathscr{H}_B$ respectively. Let the dimension of the Hilbert space $\mathscr{H}_A$ be denoted as $|A|\coloneqq \dim(\mathscr{H}_A)$. A quantum state of $A$ is represented by the density operator defined on $\mathscr{H}_{A}$. The density operator $\rho$ defined on $\mathscr{H}$ satisfies three necessary conditions: (i) $\rho\geq 0$, (ii) $\rho=\rho^\dag$, (iii) $\Tr[\rho]=1$. A pure state is a rank-one density operator given by $\psi_A\coloneqq \op{\psi}_A$ where $\ket{\psi}_A \in\mathscr{H}_A$. The set of density operators of $A$ is denoted by $\mathscr{D}(\mathscr{H}_A)$. We denote the density operator of a composite system $AB$ as $\rho_{AB} \in \mathscr{D}(\mathscr{H}_{AB})$; $\Tr_B[\rho_{AB}]=\rho_{A}\in\mathscr{D}(\mathscr{H}_A)$ is the reduced state of $A$. Separable states are those that can be expressed as a convex combination of product states
\begin{equation} \label{eq:productStates}
    \rho_{AB} = \sum_{x} p_x \rho_A^x \otimes \rho_B^x, 
\end{equation}
where $p_x \in [0,1]$ and $\sum_x p_x = 1$. States that cannot be expressed in the form of Eq.~\eqref{eq:productStates} are said to be entangled. If $\rho_{AB}$ is an entangled state and one user has system $A$ and the user has system $B$, then we say that these two users share entangled pair.

A maximally entangled state of bipartite system $AB$ is defined as $\Psi_{AB}^+\coloneqq \op{\Psi^+}_{AB}$
where \begin{equation}
\ket{\Psi^+}_{AB} = \frac{1}{\sqrt{d}}\sum_{i = 0}^{d-1} \ket{ii}_{AB},
\end{equation}
$d = \min\{|A|,|B|\}$ is the Schmidt-rank of the state $\Psi_{AB}^+$,  and $\{\ket{i}\}_{i=0}^{d-1}$ forms an orthonormal set of vectors (kets). We next discuss families of states called isotropic states and Werner states.
\begin{definition}({Isotropic state}~\cite{PhysRevA.59.4206})
An isotropic state $\rho_{AB}^{I}(p,d)$ is $U \otimes U^\ast$ invariant for an arbitrary unitary $U$. For $p\in [0,1]$, such a state can be written as
\begin{equation} \label{eq:iso}
    \rho_{AB}^{I}(p,d) \coloneqq p \Psi^+_{AB} + (1 - p) \frac{\mathbbm{1}_{AB} - \Psi^+_{AB}}{d^2 - 1}
\end{equation}
where $\Psi^+_{AB}$ is a maximally entangled state of Schmidt rank $d$. An Isotropic state $\rho_{AB}^{I}(p,d)$ written as in Eq.~\eqref{eq:iso} is separable iff $p \in [0,1/d]$. 
\end{definition}
\begin{remark}
We can also express isotropic states (Eq.~\eqref{eq:iso}) as
\begin{equation}
    \rho^I_{AB}\bigg(p(\lambda),d\bigg) = \lambda \Psi^+_{AB} + (1 - \lambda) \frac{\mathbbm{1}_{AB}}{d^2}
\end{equation}
for $p(\lambda)= [\lambda (d^2 - 1) + 1]/d^2$ and $\lambda \in [-1/(d^2-1),1]$. We note that $\lambda^{n}\geq 0$ for all even $n\in\mathbb{N}$. For our purposes in this work, we will be restricting $\rho^I_{AB}\bigg(p(\lambda),d\bigg)$ to the case $\lambda\in[0,1]$ without loss of generality. We call $\lambda$ as the visibility of the state $\rho^I_{AB}\bigg(p(\lambda),d\bigg)$.
\end{remark}

\begin{definition}({Werner state}~\cite{werner1989quantum}) A Werner state $\rho_{AB}^{W}(p,d)$ is $U \otimes U$ invariant for an arbitrary unitary $U$. For $p\in [0,1]$, such a state can be written as
\begin{equation} \label{def:wernerState}
    \rho_{AB}^W (p,d) := p \frac{2}{d(d+1)}\Pi_{AB}^+ + (1 - p) \frac{2}{d(d-1)}\Pi_{AB}^-  
\end{equation}
where $\Pi_{AB}^\pm := (\mathbb{I} \pm F_{AB})/2$ are the projections onto the symmetric and anti-symmetric sub-spaces of $\mathscr{H}_A$ and $\mathscr{H}_B$. $F_{AB} = \sum_{ij} \ket{i}\bra{j}_A \otimes \ket{j}\bra{i}_B$ is the SWAP operator on $A$ and $B$. A Werner state $\rho_{AB}^W (p,d)$ written as in Eq.~\eqref{def:wernerState} is separable iff $p \in [1/2,1]$. 
\end{definition}

A quantum channel $\mathcal{M}_{A\to{B}}:\mathscr{D}(\mathscr{H}_A)\to \mathscr{D}(\mathscr{H}_B)$ is a completely positive, trace-preserving map. A measurement channel $\mathcal{M}_{A'\to{AX}}$ is a quantum instrument whose action is defined as
\begin{equation}
    \mathcal{M}_{A'\to{AX}}(.) := \sum_x \mathcal{E}_{A'\to{A}}^x (.) \otimes \op{x}_X,
\end{equation}
where each $\mathcal{E}^x_{A'\to{A}}$ is a completely positive, trace nonincreasing map such that $\mathcal{M}_{A'\to{AX}}$ is a quantum channel and $X$ is a classical register that stores the measurement outcomes. A classical register $X$ is represented by a set of orthogonal quantum states $\{\op{x}_X\}$ defined on the Hilbert space $\mathscr{H}_X$. We define qubit standard Bell measurement with success probability $q$ as 
\begin{eqnarray} \label{eq:stdBellMeasurement}
\mathcal{M}_{A_1A_2 \to{X}}(.) := && q \sum_{j=1}^{4} \Tr[\Psi^{(j)}(.)\Psi^{(j)}] \op{j}_X \nonumber \\
&&+ (1 - q) \Tr[.] \otimes \op{\perp}_X, 
\end{eqnarray}
where $\{\Psi^{(j)}_{A_1A_2}\}_{j=1}^{4}$ denotes projective measurements on the set of maximally entangled states $\{\Psi^+_{A_1A_2},\Psi^-_{A_1A_2}, \Phi^+_{A_1A_2},\Phi^-_{A_1A_2}\}$ (see Appendix~\ref{app:twoQubitStates} for details) and $\ket{\perp} \perp \ket{j}$. 

\begin{definition}(Max-heap data structure \cite{leiserson1994introduction}) \label{def:maxHeap}
A max heap is a tree-based data structure that satisfies the following heap property: for any given node Y, if X is a parent of Y, then the key (the value) of X is greater than or equal to the key of Y. 
\end{definition} 
In Fig.~\ref{fig:maxHeapExample}, we present the tree and array representations of a max heap data structure with $9$ nodes. 
\begin{figure}
    \centering
    \includegraphics[scale=0.5]{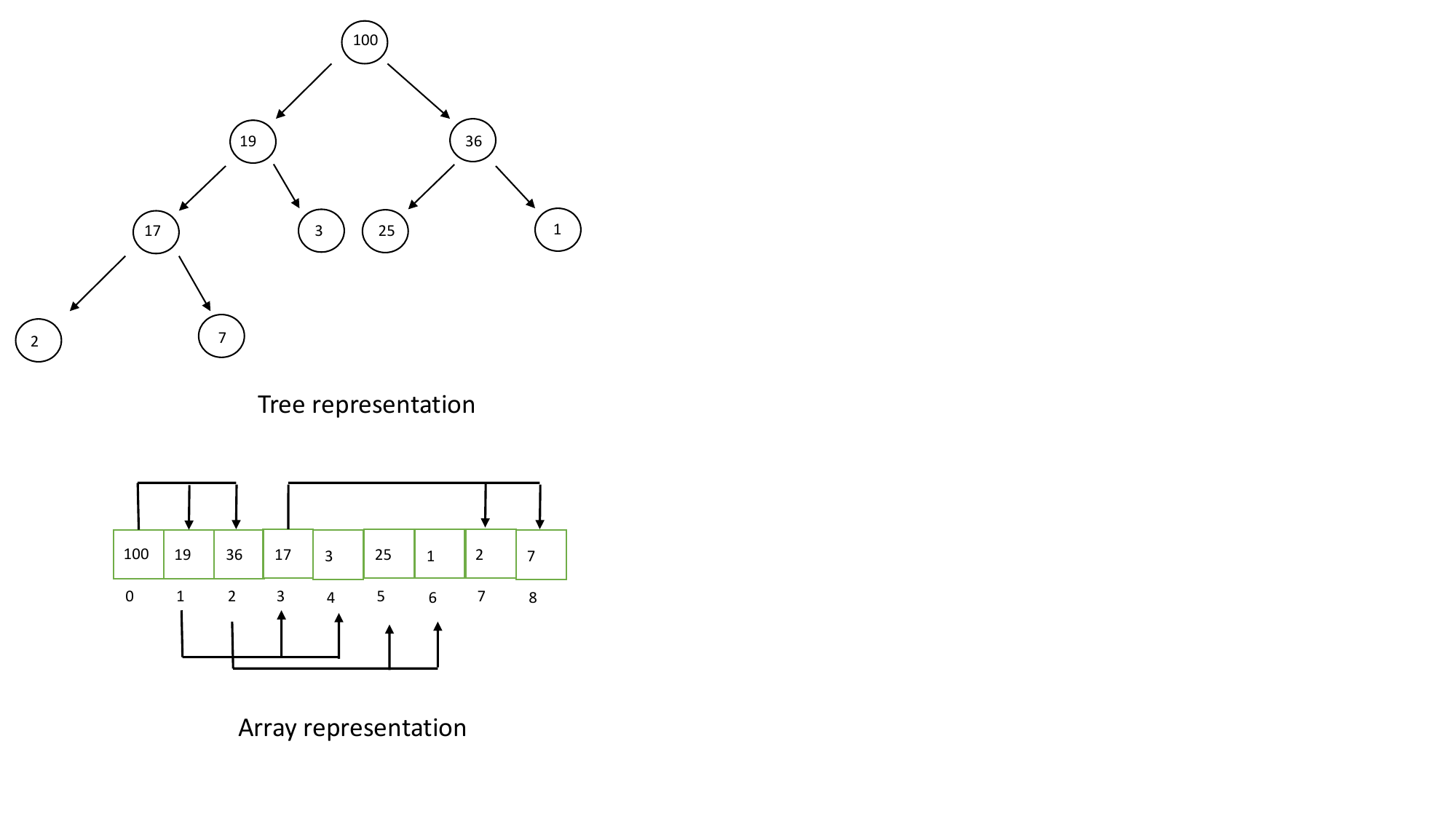}
    \caption{In this figure, we present the tree and array representations of a max-heap data structure with 9 nodes. In the tree representation, the nodes are shown in circles with the node values written inside them. In the array representation, the nodes are stored in continuous memory allocations (shown in green boxes numbered 0 to 9). The relation between the nodes is shown using arrows. We see that it satisfies the heap property: the value (or key) of a parent is always greater than its children. (Color online)}
    \label{fig:maxHeapExample}
\end{figure}

In the bipartite Bell scenario where there are two observables $\{x_1,x_2\}$ for Alice and $\{y_1,y_2\}$ for Bob, the family of tilted CHSH operators introduced in \cite{acin2012randomness} is given by
\begin{equation}
    I_\alpha^\beta \coloneqq \beta~x_1 + \alpha~x_1 y_1 + \alpha~x_1 y_2 + x_2 y_1 - x_2 y_2,
\end{equation}
where $\alpha \geq 1$ and $\beta \geq 0$. The choice of $\beta = 0$ and $\alpha = 1$ corresponds to the CHSH operator. The local upper bound of the tilted CHSH operator is given by $\beta + 2 \alpha$. The quantum upper bound of the tilted CHSH operator is given by $2 \sqrt{(1 + \alpha^2)(1 + \frac{\beta^2}{4})}$. 

\begin{definition}(Graph isomorphism~\cite{west2000graph}) \label{def:graphIsomorphism}
The graphs $G(\mathbb{V},\mathbb{E},\mathbb{L})$ and $G'(\mathbb{V'},\mathbb{E'},\mathbb{L}')$ are isomorphic iff there exists a bijective function $f:\mathbb{V}\to{\mathbb{V'}}$ such that: 
\begin{itemize}
    \item[1.] $\forall u \in \mathbb{V}, l \in \mathbb{L}, l' \in \mathbb{L}',~l(u) = l'(f(u))$
    \item[2.] $\forall u,v \in \mathbb{V},~(u,v) \in \mathbb{E} \leftrightarrow (f(u),f(v)) \in \mathbb{E'}$
    \item[3.] $\forall (u,v) \in \mathbb{E}, l \in \mathbb{L}, l' \in \mathbb{L}',~l(u,v) = l'(f(u),f(v))$
\end{itemize}
\end{definition}

\begin{note}
In general, the logarithmic function can have base as the natural exponent $e$ or some natural number greater than or equal to $2$. In this paper, we consider the logarithmic function $\log$ to have base $2$ and natural logarithm $\ln$ to have base $e$ unless stated otherwise.
\end{note}

\section{Limitations on the network architecture with repeaters} \label{sec:limitationNetwork}
A method for perfectly secure communication between a receiver and a sender requires sharing cryptographic keys between the parties~\cite{bennett1984proceedings}. The secret keys can be shared between the receiver and the sender using quantum key distribution (QKD) protocols. For these protocols, the transmission of quantum states from one party to another is an important step. However, the transmission of quantum states from a sender to a receiver via a lossy channel inevitably degrades the state being transmitted. The overlap of the shared state with the intended state typically decreases monotonically with the length of the channel. Unlike a classical signal, for quantum states this loss cannot be reduced using amplifiers since the measurement will disturb the system~\cite{fuchs1996quantum} and also quantum states cannot be cloned~\cite{wootters1982single}. The degradation of the quantum states when transmitted over a quantum channel places limitations on the distance over which there can be secure communication~\cite{das2021universal}. This limitation may be overcome by using entanglement-based QKD protocols~\cite{PhysRevLett.67.661,PhysRevLett.68.557} along with quantum repeaters~\cite{briegel1998quantum,BCH+15,sangouard2011quantum}. 

In classical computing, there is a strong motivation to use delegated computation in the form of cloud computing as it is less resource extensive on the individual user. Now, given that there is no full clarity regarding the path along which quantum computing will develop, delegated quantum computing~\cite{childs01,GLM08,BFK09} is a vision ahead~\cite{fitzsimons2017private}. This vision has been supported by efforts to provide access to quantum processors over the Internet~\cite{steffen2016progress}. The recent developments in the field of QKD and the current high-speed global communication networks only increase the scope for early adoption of delegated quantum computing. 

\begin{figure}
        \centering
        \begin{subfigure}{0.95\linewidth}
            \centering
            \includegraphics[scale=0.5]{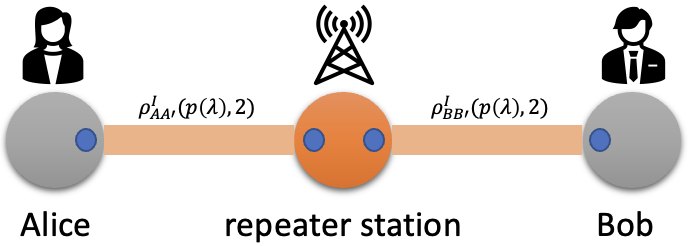}
            \caption{Alice and Bob each send halves of isotropic state of visibility $\lambda$ to a repeater station which performs standard Bell measurement on the received qubits with a success probability $q$. After the Bell measurement, Alice and Bob share an isotropic state of visibility $q\lambda^2$.}
            \label{fig:qkdSingleRepeater}
        \end{subfigure}
        \hfill
        \begin{subfigure}{0.95\linewidth}
            \centering
            \includegraphics[scale=0.4
            ]{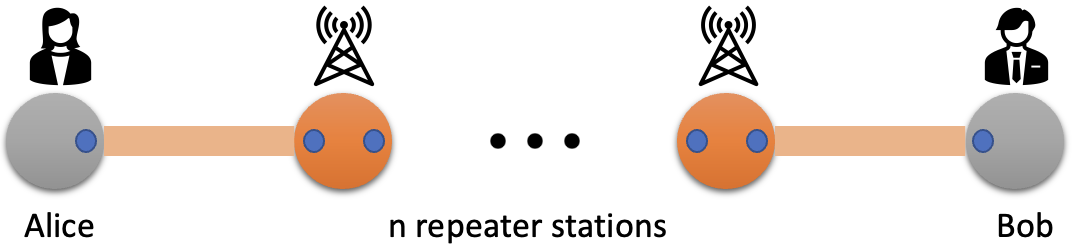}
            \caption{Alice and Bob with n repeater stations in between them. After the repeater stations have performed standard Bell measurement, Alice and Bob share isotropic state of visibility $q^n\lambda^{n+1}$.}
            \label{fig:qkdMultipleRepeater}
        \end{subfigure}
    \caption{In this figure, we present a repeater-based network to share isotropic states between Alice and Bob. The shared state is then used to perform DI-QKD protocols. The blue circles in the figure depict qubits. We assume all the repeater stations are equidistant and identical. (Color online)}
    \label{fig:qkdNetwork}
    \end{figure}

\textit{Standard linear DI key repeater chain}--- Let us consider the task of sharing secret key between two distant parties over a repeater-based network. The parties say Alice and Bob each have identical two-qubit isotropic states (see Eq.~\ref{eq:iso}) $\rho_{AA'}^I(p(\lambda),2)$ and $\rho_{BB'}^I(p(\lambda),2)$ given by
\begin{eqnarray}
    \rho_{AA'}^I(p(\lambda),2) =  \lambda \Psi^+_{AB} + (1 - \lambda) \frac{\mathbbm{1}_{AB}}{4}.
\end{eqnarray}
with $\lambda\in[0,1],$ and $\rho_{BB'}^I(p(\lambda),2) = \rho_{AA'}^I(p(\lambda),2)$. Alice and Bob send halves of their isotropic states to a repeater station. The repeater station performs a standard Bell measurement on the halves of the isotropic states with success probability $q$ (see Fig.~\ref{fig:qkdSingleRepeater}). The action of the noisy standard Bell measurement is described by Eq.~\eqref{eq:WernerBellMeasurement} (see Appendix~\ref{app:repeaterNode} for details). With the error correction possible post-Bell measurement, from a single use of the repeater Alice and Bob share the two-qubit isotropic state  
\begin{equation}
    \rho_{AB}^{I}(p(q\lambda^2),2) = q\lambda^2~\Psi^+_{AB} + \frac{1}{4} (1 - q\lambda^2) \mathbbm{1}_{AB}
\end{equation}
of visibility $\lambda^2$ and $q$ being the success probability of performing the successful standard Bell measurement by the repeater station. The state $\rho_{AB}^I(p(q\lambda^2),2)$ is separable if $\lambda \leq 1/\sqrt{3q}$. All two-qubit states are entanglement distillable if and only if they are entangled \cite{horodecki1998mixed}. All entanglement distillable states have non-zero rates for secret-key distillation \cite{horodecki2009general}. Alice and Bob can use the shared state $\rho_{AB}^I(p(q\lambda^2),2)$ to perform a DI-QKD protocol based on the tilted CHSH inequality \cite{acin2012randomness} or the modified standard CHSH inequality \cite{schwonnek2021device}. It was shown in \cite{farkas2021bell} that for such protocols, the device-independent secret key distillation rate is zero when the visibility $q\lambda^2$ of isotropic state $\rho_{AB}^I(p(q\lambda^2),2)$ is below the critical threshold 
\begin{equation} \label{eq:critThresholdWerner}
    \gamma_{\text{crit}}^{\theta} = \frac{\gamma_L^{\theta} + 1}{3 - \gamma_L^{\theta}},
\end{equation}
where $\gamma_L^{\theta} = 1/(\cos \theta + \sin \theta)$ and $\theta \in (0,\pi/2)$ \textcolor{black}{characterizes the measurement settings of Alice and Bob.} The standard CHSH-based DI-QKD protocols use settings with $\theta = \pi/4$, which gives $\gamma_{\text{crit}}^{\pi/4} \approx 0.7445$. The DI-QKD rate is known to be non-zero for $q\lambda^2 \geq 0.858$ \cite{acin2007device}. 
For $n$ repeater stations in between them (see Fig.~\ref{fig:qkdMultipleRepeater}), Alice and Bob share the two-qubit isotropic state
\begin{equation} \label{sec:pracPrototypeStatedState}
    \rho_{AB}^I(p(q^n\lambda^{n+1}),2) = q^n\lambda^{n+1}~\Psi^+_{AB} + (1 - q^n\lambda^{n+1})\frac{\mathbbm{1}_{AB}}{4} 
\end{equation}
of visibility $q^n\lambda^{n+1}$, where $q$ is the success probability of performing a standard Bell measurement by the individual repeater stations. We call such linear links of repeaters with standard Bell measurement (possibly noisy) performed at relay stations as standard linear DI key repeater chains. In the following proposition, we present limitations on the use of isotropic states for distilling secret keys via DI-QKD protocols.

\begin{proposition} \label{prop:isotropicQKD}
Consider a standard linear DI key repeater chain with $n$ relay (intermediate) stations between two end nodes. The successive nodes $v_i$ and $v_j$ of the repeater chain share a two-qubit isotropic state $\rho_{ij}^I(p(\lambda^2),2)$ and the relay stations perform standard Bell measurement with success probability $q$. The end nodes of such a network cannot perform CHSH-based DI-QKD protocols for
\begin{eqnarray}
    \lambda \in \bigg(0,(\gamma_{\text{crit}}^{\theta}/q^n)^{1/(n+1)}\bigg).
\end{eqnarray}
\end{proposition}
\begin{proof}
The end nodes of the standard linear DI key repeater chain can use the shared state $\rho_{AB}^{I}(p(q^n\lambda^{n+1}),2)$ to perform a DI-QKD protocol based on the tilted CHSH inequality \cite{acin2012randomness} or the modified standard CHSH inequality \cite{schwonnek2021device}. The device-independent secret key rate for such protocols becomes zero for $q^n\lambda^{n+1} \in (0,\gamma_{\text{crit}}^{\theta})$. This limits $\lambda$ to the range
\begin{equation} \label{eq:keyRepeater}
    \lambda \in \bigg(0,(\gamma_{\text{crit}}^{\theta}/q^n)^{1/(n+1)}\bigg)    
\end{equation}
when no secret key can be distilled. 
\end{proof}
The following observation is direct consequence of the above proposition and the fact that family of two-qubit isotropic states $\rho_{ij}^I(p(\lambda),2)$ is known to have zero DI-QKD rate for $\lambda \in (0,\gamma_{\text{crit}}^{\theta})$, where $q=1$ is safely assumed without any ramification for the observation below.
\begin{observation}\label{obv:isotropicQKD}
    There exist quantum states with non-zero DI-QKD rates that are not useful as standard DI key repeaters. For example, for a family isotropic states $\rho_{ij}^I(p(\lambda),2)$ with $\lambda \in \bigg(\gamma_{\text{crit}}^{\theta},(\gamma_{\text{crit}}^{\theta}/q^n)^{1/(n+1)}\bigg)$ the DI key rate is nonzero but the standard DI key repeater rate is zero.
\end{observation}
For non-zero key rates from tilted CHSH inequality-based and the modified standard CHSH inequality-based DI-QKD protocols we require
\begin{equation} \label{eq:boundKey}
    n < \floor[\Bigg]{\frac{\log (\lambda/\gamma^{\theta}_{\text{crit}})}{\log (1/(q\lambda))}}, 
    \end{equation}
where $\floor{.}$ denotes the floor function and $q$ is the probability of success for the perfect, standard Bell measurement at each relay (repeater) station. For values of $n$ below the above threshold, we will have a positive secure key rate. We plot in Fig.~\ref{fig:sampleGraph2}, the dependence of $n$ on $\lambda$ for performing CHSH-based DI-QKD protocol with non-zero key rate for different success probability of standard Bell measurement. In the plot, we have set $\gamma^{\theta}_{\text{crit}}$ to be 0.7445.
\begin{figure}
    \centering
    \includegraphics[scale=0.85]{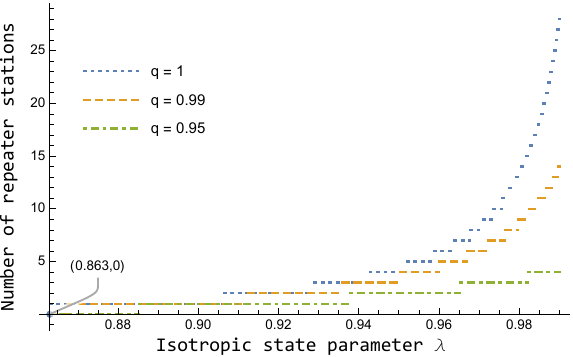}
    \caption{In this figure, we plot the allowed number of relay stations \textcolor{black}{(see Eq.~\eqref{eq:boundKey})} for performing a DI-QKD protocol with non-zero key rates by Alice and Bob as a function of the isotropic state parameter $\lambda$ for different success probability of standard Bell measurement when the critical threshold from Eq.~\eqref{eq:critThresholdWerner} is $\gamma^{\theta}_{\text{crit}} = 0.7445.$ \textcolor{black}{We assume all adjacent nodes in the repeater chain share two-qubit isotropic states with an identical parameter $\lambda$.}  (Color online)}
    \label{fig:sampleGraph2}
\end{figure}
In Fig.~\ref{fig:sampleGraph2}, we observe that the neighbouring nodes of a repeater chain network sharing a two-qubit isotropic state $\rho_{AB}^{I}(p(\lambda^2),2)$ with high values of $\lambda$ allow a large number of repeater stations between the end nodes for performing a DI-QKD protocol with non-zero key rates. In the bipartite scenario with two binary inputs and two binary outputs, there is a region where the device is nonlocal but has zero key~\cite{WDH22}. Results similar to Proposition~\ref{prop:isotropicQKD} and Eq.~\eqref{eq:boundKey} would apply to such a scenario if considered appropriately.

We present in Appendix~\ref{sec:repeaterRelay} bounds on the number of repeater stations such that the state shared by the end nodes of the network is useful for various information processing tasks. The repeater network discussed in this section can be generalised to a network structure with multiple pairs of end nodes. 

Graph theory provides a framework to assess the topology of networks having spatially separated users across the globe, as well as networks with nodes within a small space over a processor or circuit without getting into specific details of implementation. Motivated by the power of abstraction, we analyze the robustness and scalability of the quantum Internet using graph theoretic tools in the following section.  

\section{Graph theoretic analysis of networks} \label{sec:GraphTheoryAnalysis}
In this section, we first present a graph theoretic framework for networks. Using such a framework, we analyze the limitations on the scalability of repeater-based networks for different information processing tasks. Looking into the robustness of network topologies, we present tools for measuring network robustness and identifying the critical components. Using such measures, we compare the robustness of different network topologies. 

\subsection{Graph theoretic framework of networks} \label{sec:GraphTheoryFramework}
Let us consider networks represented as graph $G(\mathbb{V},\mathbb{E})$ classified as weighted and undirected. A graph is a mathematical structure that is used to define pairwise relations between objects called nodes. The set of nodes, also called vertices is denoted by $\mathbb{V}$ and $\mathbb{E}$ denotes the set of edges which are pairs of nodes of the graph that connect the vertices. We denote $|\mathbb{V}|$ as $N_v$ and $|\mathbb{E}|$ as $N_e$, where $|\mathbb{X}|$ denotes the size of the set $\mathbb{X}$. The vertices of the graph $G$ are denoted as $v_i \in \mathbb{V}$ and the edges connecting the nodes $v_i,v_j \in \mathbb{V}$ as $e_{ij} \in \mathbb{E}.$ We denote the path connecting two distant nodes $v_i$ and $v_k$ as ${\cal P}(v_i,v_k)$ having length $len(\cal P)$. The shortest network path between $v_i$ and $v_k$ is denoted as $dist(v_i,v_k)$. The shortest path length between the most distant nodes of a graph is called the diameter of the graph. A value $w_{ij}\in\mathbb{R}$ assigned to an edge $e_{ij}$ of the graph is called the edge weight. The graph $G$ together with $w_{ij}$ is called a weighted graph. A graph where the edges do not have a direction is called an undirected graph. The edges of an undirected graph indicate a bidirectional relationship where each edge can be traversed from both directions. 

In general, the nodes, edges, and edge weights of a network can change with time (see Appendix~\ref{app:timeEvolving} for details). In this work, we will deal with undirected, weighted graphs depicting networks for communication tasks among multiple users. The edges in the graph are representative of links in the corresponding network, where links between nodes are formed due to quantum channels (or gates) over which resources are being transmitted between connected nodes. We denote labelled graphs as $G(\mathbb{V},\mathbb{E},\mathbb{L}),$ where $\mathbb{L}$ denotes the set of labels associated with the vertices and edges of the graph. 

We denote the nodes of the graph that are present between the end nodes when traversing along a path from one end node to another as the virtual nodes associated with the path. While analysing the network for different tasks, it may be possible that all the virtual nodes of the network are secure and cooperate in the execution of the task. We call this the cooperating strategy. It may also be possible that some virtual nodes of the network may be compromised and are not available for the task. We call this the non-cooperating strategy. We next define the weights for the edges of a network performing different tasks.
\begin{definition} \label{def:weight}
Let a network depicted by an undirected, weighted graph $G(\mathbb{V},\mathbb{E})$, where $v_i\in \mathbb{V}$ for $i\in [N_v]$, be given by $\mathscr{N}(G(\mathbb{V},\mathbb{E}))$. The success probability of transmitting a desirable resource $\chi$ between any two different nodes $v_i, v_j$ (i.e., when $i\neq j$) connected with edge $e_{ij}$ is given by $p_{ij}$; we assume $p_{ii}=1$ for an edge $e_{ii}$ connecting a node $v_i$ with itself. The weight $\mathrm{w}(e_{ij})$ for each edge $e_{ij}$ is given by $-\log p_{ij}$. We define an effective weight $\mathrm{w}_{\ast}(e_{ij})$ of an edge $e_{ij}$ over which the resource $\chi$ is transmitted between $v_i$ to $v_j$ for a particular information processing task ($\mathrm{Task}_\ast$) as, for all $i,j\in[N_v]$ 
\begin{align}
    \mathrm{w}_{\ast}({e_{ij}})\coloneqq \begin{cases}
      -\log p_{ij} & \text{if $p_{ij}\geq p_{\ast}$},\\
      \infty & \text{otherwise}, \label{eq:weight}
    \end{cases}       
\end{align}
where $p_{\ast}$ is the critical probability below which the desired information processing task ${\rm Task}_{\ast}$ fails. Since $G(\mathbb{V},\mathbb{E})$ is undirected, we have $p_{ij}=p_{ji}$, $\mathrm{w}(e_{ij})=\mathrm{w}(e_{ji})$, $\mathrm{w}_{\ast}(e_{ij})=\mathrm{w}_{\ast}(e_{ji})$.
\end{definition}
\begin{observation} \label{obs:task}
The weight $\mathrm{w}(e_{ij})$ for an edge is path-dependent and additive across connecting edges. If a resource $\chi$ is being transmitted between nodes $v_i$ and $v_l$ by traversing the virtual nodes $v_j$ and $v_k$ in an order $v_i\to{v_j}\to{v_k}\to{v_l}$, i.e., through a connecting path $e_{ij}\to{e_{jk}}\to {e_{kl}}$, then the weight $\mathrm{w}({e_{i\to{j}\to{k}\to{l}}})= \mathrm{w}(e_{ij})+\mathrm{w}({e_{jk}})+\mathrm{w}({e_{kl}})= -\log p_{ij}p_{jk}p_{kl}$. The effective weight $\mathrm{w}_{\ast}(e_{ij})$ is also path-dependent. The effective weight $\mathrm{w}_{\ast}({e_{i\to{j}\to{k}\to{l}}})= \mathrm{w}_{\ast}(e_{ij})+\mathrm{w}_{\ast}({e_{jk}})+\mathrm{w}_{\ast}({e_{kl}})= -\log p_{ij}p_{jk}p_{kl}$ if $p_{ij}p_{jk}p_{kl}\geq p_{\ast}$, else $\mathrm{w}_{\ast}({e_{i\to{j}\to{k}\to{l}}})= \infty$. If any of the $p_{ij}, p_{jk}, p_{kl}$ is strictly less than $p_{\ast}$ then $p_{ij}p_{jk}p_{kl}<p_{\ast}$. To maximize the success probability of transmitting the resource $\chi$ between any two nodes of the network, it is desirable to select the path between these two nodes that has the minimum weight.
\end{observation}
The critical success probability for performing $\mathrm{Task}_\ast$ over a network limits its diameter. This motivates the following definition of critically large networks for $\mathrm{Task}_\ast$.
\begin{definition}[Critically large network] \label{def:irregularGraph}
We define a network $\mathscr{N}_{\text{crit}}(G(\mathbb{V},\mathbb{E}))$ as {\it critically large} network for $\mathrm{Task}_{\ast}$ if
\begin{eqnarray}
\forall_{e_{ij}\in \mathbb{E}}\,~p_{ij}\leq c,~\text{where}~c\in(0,1),
\end{eqnarray}
and it contains at least two vertices $x_0,y_0\in\mathbb{V}$ which are at distance 
\begin{equation}
    dist(x_0,y_0)\geq \left\lceil \frac{\log p_*}{\log c} \right\rceil +1,
    \label{eq:far_pair}
\end{equation}
where $p_{\ast}$ is the critical probability for successful transmission of resource $\chi$ (Definition~\ref{def:weight}). 
\end{definition}
In the following proposition we show that there are at least two vertices in critically large network that cannot perform $\mathrm{Task}_\ast$ over all paths of length larger or equal to the distance between the vertices.
\begin{proposition} \label{theorem:upperLimitSharing}
    Assume that it is possible to perform ${\rm Task}_{\ast}$ between any two distinct nodes $v_i,v_j$ of the network $\mathscr{N}_{\text{crit}}(G(\mathbb{V},\mathbb{E}))$ if and only if nodes $v_i,v_j$ can share resource $\chi$ over the network path ${\cal P}(v_i,v_j)$ having success probability $s({\cal P}(v_i,v_j)) \geq p_\ast >0$. Then
    \begin{eqnarray} 
    \exists_{n_0\in \mathbb{N}}~\forall_{\{v_i,v_j\} \in \mathbb{V}} \,&&\bigg[dist(v_i,v_j)\geq n_0 \nonumber \\
    && \Rightarrow \forall_{{\cal P}(v_i,v_j)}\,\,s({\cal P}(v_i,v_j))< p_\ast\bigg].
\end{eqnarray}
In other words, there will be at least two vertices in this network, that cannot share resource $\chi$ by $\mathrm{Task}_\ast$.
\end{proposition}
\begin{proof}
    We begin with an observation that the success probability of sharing a resource between vertices $\{v_i,v_j\}$ along the path ${\cal P}(v_i,v_j)$ is given by $s({\cal P}(v_i,v_j)) \coloneqq p_{ik}...p_{rj}$. The vertices $v_i$ and $v_j$ cannot share a resource along the path ${\cal P}(v_i,v_j)$ if $s({\cal P}(v_i,v_j)) < p_\ast$.
    
    Let us choose $n_0$ such that $c^{n_0} < p_\ast$. Consider now any two vertices $v_i$ and $v_j$ be such that $dist(v_i,v_j) > n_0$ (the set of such vertices is non-empty since $dist(x_0,y_0)>n_0$ by assumption in Eq. (\ref{eq:far_pair})). Then any path ${\cal P}(v_i,v_j)$ has length $l \geq dist(v_i,v_j)$. The success probability of sharing resource along the path ${\cal P}(v_i,v_j)$ is given by $s({\cal P}(v_i,v_j)) = p_{ik}p_{km}...p_{rj} \leq (\max\{p_{ik},p_{km},...,p_{rj}\})^l \leq c^l \leq c^{n_0} < p_\ast$. Thus for any $\{v_i,v_j\}$ at distance $\geq n_0$, there does not exist a path ${\cal P}(v_i,v_j)$ with $s({\cal P}(v_i,v_j)) \geq p_\ast$.
\end{proof}

Consider a $d$-dimensional lattice $G_{lat}(\mathbb{V},\mathbb{E})$ having $|\mathbb{V}| \rightarrow \infty$. The vertex set $\mathbb{V}$ is defined as the set of elements of $\mathbb{R}^d$ with integer coordinates. Let us denote $G_{sg}(\mathbb{V}_{sg},\mathbb{E}_{sg})$ as a finite subgraph of $G(\mathbb{V},\mathbb{E})$ from which the entire graph can be constructed by repetition and
\begin{eqnarray}
    \exists_{N_{sg} \ll |V|} : \forall_{v_j \in \mathbb{V}_{sg}}~\mathrm{degree}(v_j) \leq N_{sg}.
\end{eqnarray}
A percolation configuration $\omega_p = (\omega_{ij} : e_{ij} \in \mathbb{E})$ on the graph $G(\mathbb{V},\mathbb{E})$ is an element of $\{0,1\}^{|\mathbb{E}|}$. If $\omega_{ij} = 1$, the edge is called open, else closed. If a node $v_i \in \mathbb{V}$ fails, all the edges $e_{ij} \in \mathbb{E}$ connected to $v_i$ will be disconnected and for such edges $\omega_{ij} = 0$. A configuration $\omega_p$ is a subgraph of $G(\mathbb{V},\mathbb{E})$ with vertex set $\mathbb{V}$ and edge-set $\mathbb{E}_p \coloneqq \{e_{ij} \in \mathbb{E} : \omega_{ij} = 1\}$ (cf.~\cite{duminil18}).  
For performing $\mathrm{Task}_\ast$ over lattice network, we next present a theorem that describes conditions for which there is no percolation in a lattice network.
\begin{theorem}~\label{lem:percolation}
    Let us consider performing $\mathrm{Task}_\ast$ (Definition~\ref{def:weight}) over the lattice $G_{lat}(\mathbb{V},\mathbb{E})$  where each edge is open with probability $p_{ij}$ and $0<p_\ast \leq p_{ij} <1$. Then
    the network arising among nodes from this task does not form a percolation configuration, i.e., a connected component of length $N_c$ such that $N_c/|\mathbb{V}| > 0$.
\end{theorem}
The above theorem follows from the facts that there are periodic repetitions of the finite subgraph $G_{sg}(\mathbb{V}_{sg},\mathbb{E}_{sg})$ in $G_{lat}(\mathbb{V},\mathbb{E})$ and there exists at least two subgraphs whose distance is greater than critical threshold for inter-subgraph nodes to remain connected for $\mathrm{Task}_{\ast}$ with some desirable probability (see Proposition~\ref{theorem:upperLimitSharing}). See also~\cite[Page 20]{SA18} for discussion on 1-dimensional lattice and condition for percolation~\cite{BS05} to exist. \textcolor{black}{Theorem~\ref{lem:percolation} implies limitations on the scalability of quantum communication (see Observation~\ref{obs:graphDiameter}) and DI-QKD (see Example~\ref{example:DIQKD}) over networks.}

Let us next define the adjacency matrix and the effective success matrix of a network in terms of the success probability of transmitting a desirable resource $\chi$ between two nodes of the network.
\begin{definition}
Consider a network $\mathscr{N}(G(\mathbb{V},\mathbb{E}))$ with $|\mathbb{V}|=N_{v}$ and ${\mathrm{w}(e_{ij})}_{i,j}$. The adjacency matrix $\mathsf{A}$ of the network is a $N_v \times N_v$ matrix such that for all $i,j\in [N_v]$ we have
\begin{align}
    [\mathsf{A}]_{ij}= {\rm w}(e_{ij}),
\end{align}
and the effective adjacency matrix $\mathsf{A}_{\ast}$ of the network is given by
\begin{align}
    [\mathsf{A}_{\ast}]_{ij}= {\rm w}_{\ast}(e_{ij}).
\end{align}
The effective success matrix of any network $\mathscr{N}(G(\mathbb{V},\mathbb{E}))$ for the transmission of a desirable resource $\chi$ (associated with ${\rm Task}_{\ast}$) between its nodes is an $N_v \times N_v$ matrix $\mathbb{f}_{\ast}$ such that for all $i,j\in [N_v]$ we have
\begin{align}
    [\mathbb{f}_{\ast}]_{ij}\coloneqq \begin{cases}
      p^{\rm max}_{ij} & \text{if $p^{\rm max}_{ij} \geq p_{\ast}$ for $i\neq j$},\\
      0 & \text{otherwise},
    \end{cases}       
\end{align}
where $p^{\rm max}_{ij}$ is the maximum success probability of transmitting the desirable resource $\chi$ between nodes $v_i$ and $v_j$ over all possible paths between the two nodes.
\end{definition}

The elements $[\mathbb{f}_\ast]_{ij}$ of the success matrix provide the highest probability with which $\chi$ can be shared between the nodes $v_i$ and $v_j$ thus corresponds to the path having minimum cumulative effective weight. Let us consider a graph of diameter $2$ as shown in Fig.~\ref{fig:adjSuccessMatrixCompare}. The adjacency matrix and the success matrix for this graph are different as $[\mathsf{A}]_{14} = [\mathbb{f}_{\ast}]_{14} = 0.79$ while $[\mathsf{A}]_{12} = 0.198$ and $[\mathbb{f}_{\ast}]_{12} = 0.54173.$ We observe that the success probability of transferring $\chi$ between nodes $v_1$ and $v_2$ via a non-cooperating strategy leads to a lower success probability $p_{1\to{2}}$ as compared to that via a cooperative strategy $p_{1\to{3}\to{2}}.$ 
\begin{figure}
    \centering
    \includegraphics[scale=0.5]{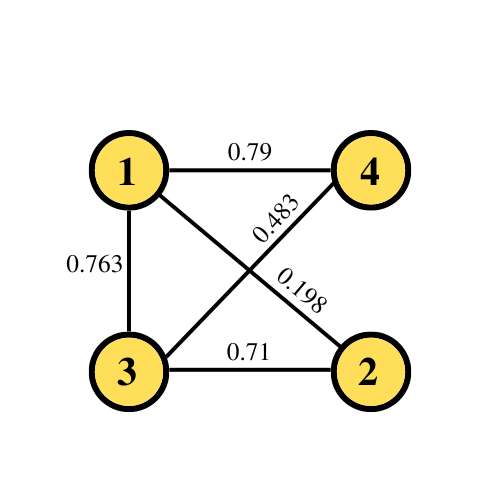}
    \caption{A partially connected mesh network with $4$ nodes. In this graph the adjacency matrix $\mathsf{A}$ and the success matrix $\mathbb{f}$ are different. While performing an information processing task using this network say transferring a resource $\chi$ from $v_1 \to{v_2}$, it is preferable to use a cooperative strategy $v_1\to{v_3}\to{v_2}$ over a non-cooperative strategy $v_1\to{v_2}$ as it has a higher success probability. (Color online)}
    \label{fig:adjSuccessMatrixCompare}
\end{figure}

\begin{note}
Henceforth, we will be dealing with communication over networks where the success probability of transmitting resources from a node $a$ to node $c$ via node $b$ is less than or equal to the multiplication of the success probability of resource transmission from $a$ to $b$ and $b$ to $c$.
\end{note}

\subsection{Limitations on repeater networks} \label{sec:criticalParametersForNetworks}
In this subsection, we present the critical success probability of the elementary links for implementing different information processing tasks. Extending to a linear repeater-based network, we present the critical time and length scales for implementing different information processing tasks.

\subsubsection{Critical success probability for repeater networks}
Consider a network $\mathscr{N}(G(\mathbb{V},\mathbb{E}))$ for performing a particular information processing task denoted by the symbol $\mathrm{Task}_\ast$. As examples, we let the $\mathrm{Task}_\ast$ be sharing of entanglement, or implementing teleportation protocol between the nodes $\{v_i,v_j\} \in \mathbb{V}$. Let $v_i,v_j$ share an isotropic state given by 
\begin{equation}
    \rho_{ij}^{I}(p,d) = p_{ij} \Psi^+_{ij} + (1-p_{ij}) \frac{\mathbbm{1}_{ij} - \Psi^+_{ij}}{d^2 - 1},
\end{equation}
via a qudit depolarising channel (cf.~\cite{das2021universal,kaur2022upper}). Let performing $\mathrm{Task}_{\ast}$ require nodes $v_i$ and $v_j$ to share $\Psi^+_{ij}$ with critical success probability $p_\ast$.
We note that the state $\rho_{ij}^{I}(p,d)$ become separable for $p_{ij} < 1/d$, this implies a critical success probability $p_\ast^{\text{ent}} \geq  1/d$. The singlet fraction of $\rho_{ij}^{I}(p,d)$ is given by $f_{ij} = p_{ij}$. The maximum achievable teleportation fidelity of a bipartite $d \times d$ system in the standard teleportation scheme is given by $F = \frac{f_{ij} d + 1}{d + 1}$~\cite{HHH99}. \textcolor{black}{The maximum fidelity achievable by classical strategy} is $F_{cl} = \frac{2}{d + 1}$~\cite{HHH99}. Thus the shared state between $v_i$ and $v_j$ is useful for quantum teleportation if $f_{ij} > 1/d$. The critical success probability $p_\ast^{\text{tel}}$ for performing teleportation protocol over $\mathscr{N}(G(\mathbb{V},\mathbb{E}))$ is $p_\ast^{\text{tel}} > 1/d$.

\subsubsection{Critical time and length scales for repeater networks} \label{sec:criticalScaleRepeater}
Consider an entanglement swapping-based repeater network. Let there be two sources $S_1$ and $S_2$ producing dual-rail encoded entangled pairs (for details see Appendix~\ref{app:dualRail}) in the state $\Psi^+$ with probability $\eta_s$ and with probability $1 - \eta_s$ produces a vacuum state. The source $S_1$ sends one qubit from its entangled pair to Alice and the other to the repeater station via optical fibers of length $l$. Similarly, $S_2$ sends one qubit from its entangled pair to Bob and the other to the repeater station via optical fibers of length $l$ (see Fig.~\ref{fig:repeatermemory}). Let the qubits be stored in quantum memories at the repeater station and the stations of Alice and Bob for time $t$. 
\begin{figure}
    \centering
    \includegraphics[scale=0.35]{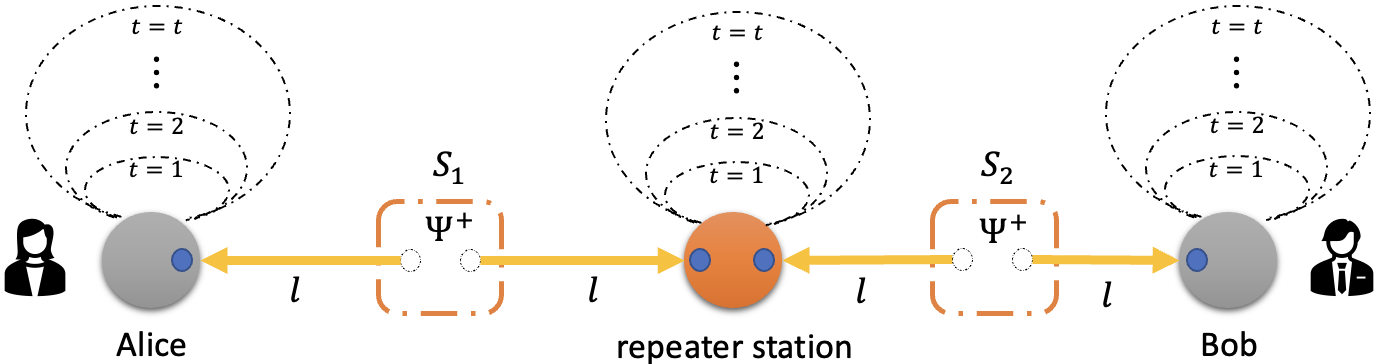}
    \caption{In this figure, we present an entanglement swapping-based repeater network. There are sources $S_1 (S_2)$ producing state $\Psi^+$ and sending it to a repeater station and Alice (Bob) via optical fibers of length $l$ (shown in yellow). The qubits are stored in quantum memories at the repeater station and the stations of Alice and Bob for $t$ time steps (shown as self-loops). The repeater station performs standard Bell measurement on its share of qubits. (Color online)}
    \label{fig:repeatermemory}
\end{figure}
We model the evolution of the qubits through the fiber and at the quantum memory as a qubit erasure channel (see Sec.~\ref{app:ErasureChannel}) with channel parameter $\eta_e = \mathrm{e}^{-(\alpha l + \beta t)}$, where $\alpha$ and $\beta$ are respectively the properties of the fiber and the quantum memory. The repeater station performs a standard Bell measurement on its share of qubits with success probability $q$. After the repeater station has performed the standard Bell measurement, Alice and Bob share the state $\Psi^+$ with probability $q~\eta_s^2~\mathrm{e}^{-2(\alpha l + \beta t)}$ (cf.~\cite[Eq.~(64)]{das2021universal}). Let Alice and Bob require to perform $\mathrm{Task}_{\ast}$ (Definition~\ref{def:weight}) using their shared state. Furthermore, let performing $\mathrm{Task}_{\ast}$ require Alice and Bob to share $\Psi^+$ with critical success probability $p_\ast$. We then require 
\begin{equation} \label{eq:critLengthTimeRepeater}
    q~\eta_s^2~\mathrm{e}^{-2(\alpha l + \beta t)} > p_\ast.
\end{equation}
We observe that Eq.~\eqref{eq:critLengthTimeRepeater} bounds the length of the optical fibers and the time till which the qubits can be stored in the quantum memories. This motivates the definition of the critical length of the fibers $l_c$ and the critical storage time at the nodes $t_c$ above which the shared state becomes useless for information processing tasks. These two critical parameters are related via the expression
\begin{equation} \label{eq:critLengthTimeRepeaterBound}
    \alpha~l_c + \beta~t_c < \frac{1}{2} \ln \bigg(\frac{q~\eta_s^2}{p_\ast}\bigg).
\end{equation}
We plot in Fig.~\ref{fig:critLengthTimeRepeater}, the critical fibre length $l_c$ and the critical storage time $t_c$ for different qubit architectures and set $p_\ast = 0.5$ for some desired $\mathrm{Task}_{\ast}$.
\begin{figure}
    \centering
    \includegraphics[scale=0.6]{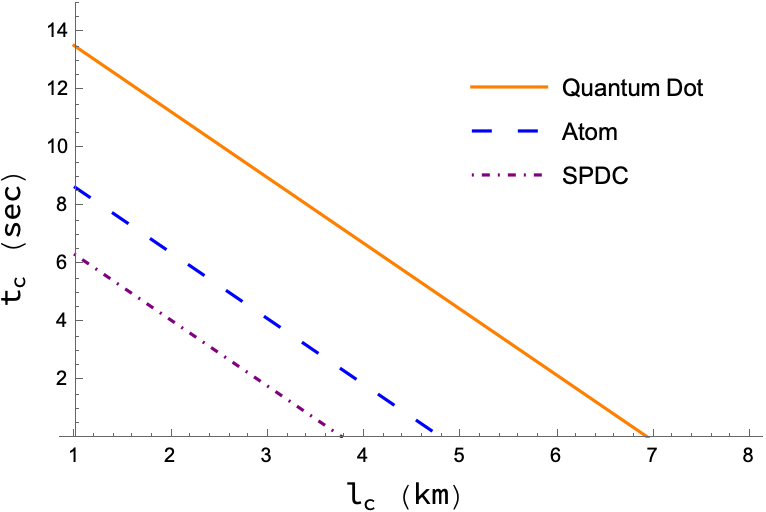}
    \caption{In this figure, we consider a repeater-based network and plot the critical storage time $t_{cr}$ as a function of critical fiber length $l_{cr}$ for different qubit architectures. The single photon architectures have efficiencies (a) $\eta_s$ = 0.97 (Quantum Dots~\cite{press2007photon}) (b) $\eta_s = 0.88$ (Atoms~\cite{barros2009deterministic}) and (c) $\eta_s = 0.84$ (SPDC~\cite{altepeter2005phase}). \textcolor{black}{For this, we set $p_\ast = 0.5$, $q = 1$, $\alpha = 1/22 \text{ km}^{-1}$ and $\beta = 1/50 \text{ sec}^{-1}$.}  (Color online)}
    \label{fig:critLengthTimeRepeater}
\end{figure}

Let us introduce a finite number of repeater stations between the two end nodes, each performing standard Bell measurements on its share of qubits. The network is useful for $\mathrm{Task}_{\ast}$ when
\begin{equation} \label{eq:critLengthTimeRepeaterRelay}
    q^r~\eta_s^{r+1}~\mathrm{e}^{-2r(\alpha l + \beta t)} > p_\ast,
\end{equation}
where $r$ denotes the number of repeater stations between the end nodes. Let $l_{cr}$ and $t_{cr}$ denote the critical length of the fiber and the critical storage time of the quantum memory. These two critical parameters are then related via the expression
\begin{equation}\label{eq:critLengthTimeMultipleRepeaterRelay}
    \alpha~l_{cr} + \beta~t_{cr} < \frac{1}{2r} \ln \bigg( \frac{q^r \eta_s^{r+1}}{p_\ast} \bigg).
\end{equation}
We plot in Fig.~\ref{fig:critLengthTimeMultipleRepeater}, the critical fiber length $l_{cr}$ and the critical storage time $t_{cr}$ such that Eq.~\eqref{eq:critLengthTimeMultipleRepeaterRelay} holds for different values of $r$. The bounds on $l_{cr}$ and $t_{cr}$ for other values of $r$ not shown in Fig.~\ref{fig:critLengthTimeMultipleRepeater} can be obtained from Eq.~\eqref{eq:critLengthTimeMultipleRepeaterRelay}. We observe that for a given information processing task, increasing the number of repeater stations allows shorter fiber lengths and quantum memory storage times. 
\begin{figure}
    \centering
    \includegraphics[scale=0.6]{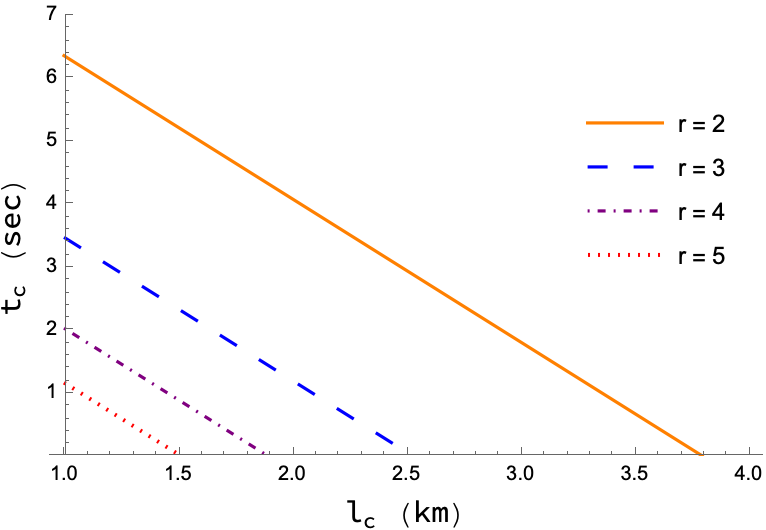}
    \caption{In this figure, we consider a repeater-based network and plot the critical storage time $t_{cr}$ as a function of critical fiber length $l_{cr}$ for the different numbers of repeater stations. \textcolor{black}{For this, we set $p_\ast = 0.5$, $\alpha = 1/22 \text{ km}^{-1}$, $\beta = 1/50 \text{ sec}^{-1}$, $q = 1$ and $\eta_s = 0.999$.} (Color online)}
    \label{fig:critLengthTimeMultipleRepeater}
\end{figure}

We note that the optimal rate of two-way assisted quantum communication or entanglement transmission (i.e., in an informal way, it is the maximum number of ebits per use of the channel in the asymptotic limit of the number of uses of the channel) over an erasure channel, also called LOCC (local operations and classical communication)-assisted quantum capacity of an erasure channel, is given by $\eta_e\log_2d$~\cite{BDS97}, where $1-\eta_e$ is the erasing probability and $d$ is the dimension of the input Hilbert space. Two-way assisted quantum and private capacities for erasure channel coincide~\cite{GEW16,pirandola2017fundamental} (see \cite{WTB16} for strong-converse capacity). Two-way assisted private and quantum capacities for a qubit erasure channel is $\eta_e$~\cite{GEW16,pirandola2017fundamental}.

\textcolor{black}{In the following section, we present limitations on the scalability of networks for quantum communication tasks assuming some hypothetical scheme can improve the transmittance of quantum channels connecting the nodes of the network.}
\subsection{Limitations on quantum network topologies} \label{sec:networkTopology}
Let us consider an equilateral triangle-mesh network having the set of nodes as $\mathbb{V}$ (see Fig.~\ref{fig:transmittanceRepeater1}(a) for $n = 3$). The nodes of the triangle network are connected by qubit erasure channels having transmittance $\eta = \mathrm{e}^{-(\alpha l + \beta t)}$, where $l$ denotes the distance between the nodes and $t$ is the time it takes for some resource $\chi$ to pass through the channel. In this section, we assume perfect measurements for the sake of simplicity. Consideration of imperfect measurements will only increase the critical probability $p_\ast$ for the transmission of quantum  resources over an edge between the nodes. This would imply that quantum network topologies and architecture with imperfect measurement devices will be more limited (constrained) than the network with perfect measurement devices.

Let us introduce a repeater scheme in the form of a star network (see Fig.~\ref{fig:transmittanceRepeater1}(b) for $n = 3$) to effectively mitigate the losses due to transmission. The star network has virtual channels connecting the repeater node $v_R$ to the node $v_i \in \mathbb{V}$. The transmittance of the virtual channels $\eta_R = \eta^{1/\sqrt{3}}$ is greater than $\eta$.

We now consider the transmittance of channels connecting pairs of nodes with a hypothetical scheme in a network that would lead to the following assumption.  
\begin{assumption} \label{assumption:transmittance}
Let us assume there exist repeater-based quantum communication or quantum key distribution schemes that can mitigate the loss due to transmittance of a quantum channel such that the rate of communication (ebits or private bits per channel use~\cite{das2021universal}) is effectively of the order
\begin{equation} \label{eq:conjecture1}
    \eta_R = \eta^{1/f}
\end{equation}
in some regime (distance)\footnote{This need not be true for any length/distance in general and may hold only in certain distance regimes or sections.}, where $\eta$ is the transmittance of the channel and for some $f \geq 1$ (cf.~\cite{lucamarini2018overcoming,wang2018twin,ma2018phase,zeng2020symmetry}).
\end{assumption}

To illustrate Assumption~\ref{assumption:transmittance}, let us consider a regular polygon network with $n$ nodes having vertex set $\mathbb{V}$ and edge set $\mathbb{E}$ as shown in Fig.~\ref{fig:transmittanceRepeater1}(a). For $e_{ij} \in \mathbb{E}$ the nodes $\{v_i,v_j\}$ are connected by erasure channels having transmittance $\eta = \mathrm{e}^{-\alpha l}$, where $l$ is the distance between the nodes. Let us introduce a repeater scheme in the form of a star network as shown in Fig.~\ref{fig:transmittanceRepeater1}(b) having the vertex set $\mathbb{V}_s \coloneqq \mathbb{V} \cup \{v_R\}$ where $v_R$ denotes the repeater node and the edge set is denoted by $\mathbb{E}_s$. For $e_{kR} \in \mathbb{E}_s$ the channel connecting nodes $v_k,v_R \in \mathbb{V}_s$ have transmittance $\eta_R = \mathrm{e}^{-\alpha l/2 \sin (\pi/n)}$, where $n = |\mathbb{V}|$. Comparing with Eq.~\eqref{eq:conjecture1} we observe that the star network provides an advantage of $f = 2 \sin (\pi/n)$ over a repeater-less scheme for $n < 6$. For the triangle network discussed previously, we have $f = \sqrt{3}$. There may be other repeater-based schemes which uses entanglement distillation and error-correction techniques to further mitigate the transmission loss and enhance the rate of communication between end nodes. 
\begin{figure}
    \centering
    \includegraphics[scale=0.32]{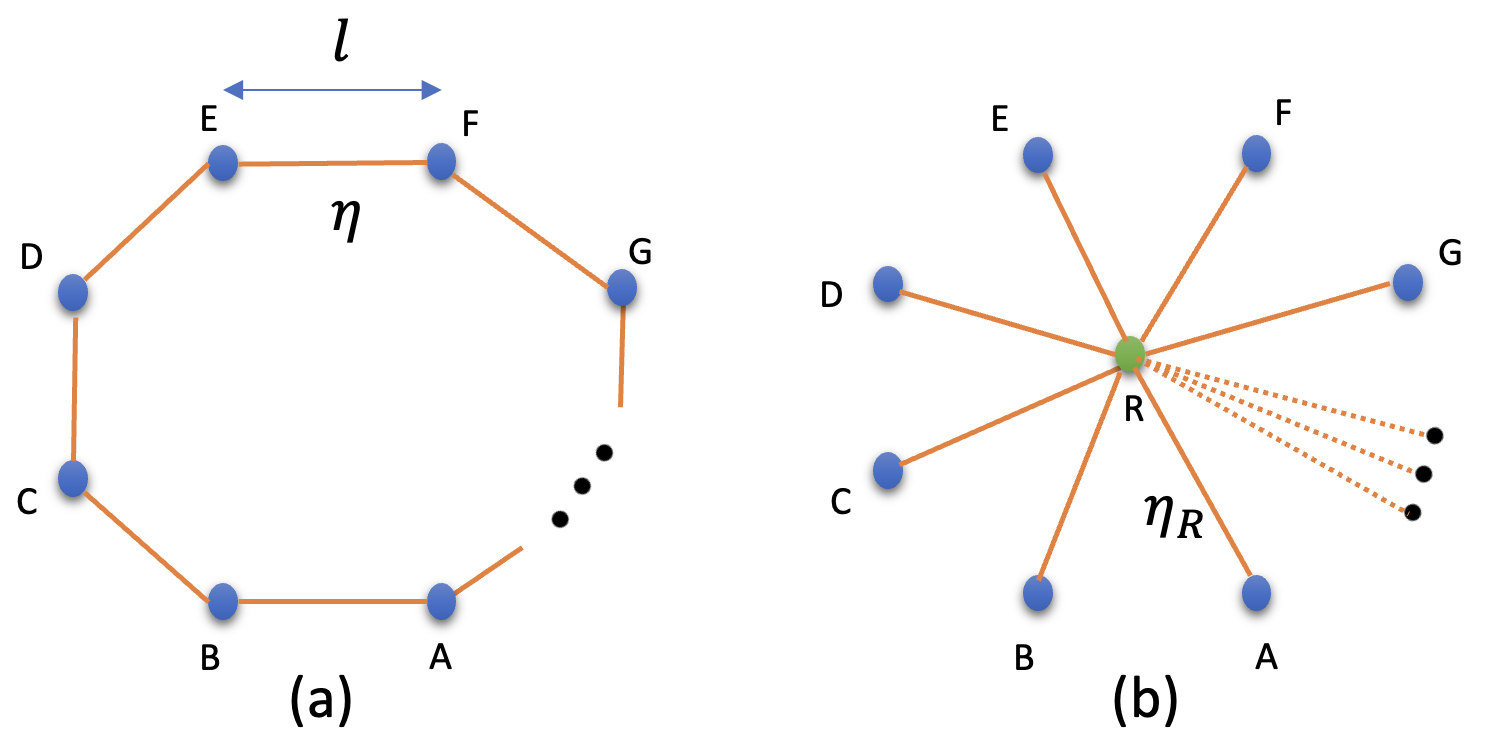}
    \caption{In this figure, we show (a) repeater-less  regular polygon network with $n$ nodes and (b) star-repeater network with $n$ nodes. (Color online)}
    \label{fig:transmittanceRepeater1}
\end{figure}

In Eq.~\eqref{eq:conjecture1}, the case of $f = 1$ has been shown in ~\cite{das2021universal} for measurement-device-independent quantum key distribution (see \cite{pirandola2017fundamental,WTB16} for quantum key distribution over point-to-point channel) and the case of $f = 2$ has been shown in \cite{lucamarini2018overcoming,wang2018twin,ma2018phase,zeng2020symmetry,XWL+23} for twin-field quantum key distribution and asynchronous measurement-device-independent quantum key distribution~\cite{xie2022breaking,ZLX+23}. In Fig.~\ref{fig:transmittanceRepeater}, we set $\beta = 0$ and plot the variation of $\eta_R$ as a function of the channel length $l$ for different values of $f$. 

Consider the task of sending ebits or private bits over the qubit erasure channel at a rate greater than $p_\ast$. To successfully perform the task, the critical length $l_c$ and the critical time $t_c$ are related via the expression
\begin{equation} \label{eq:critParameterNetworkAdvantage}
    \alpha l_c + \beta t_c \leq - f \ln p_\ast.
\end{equation} 
\begin{figure}
    \centering
    \includegraphics[scale=0.6]{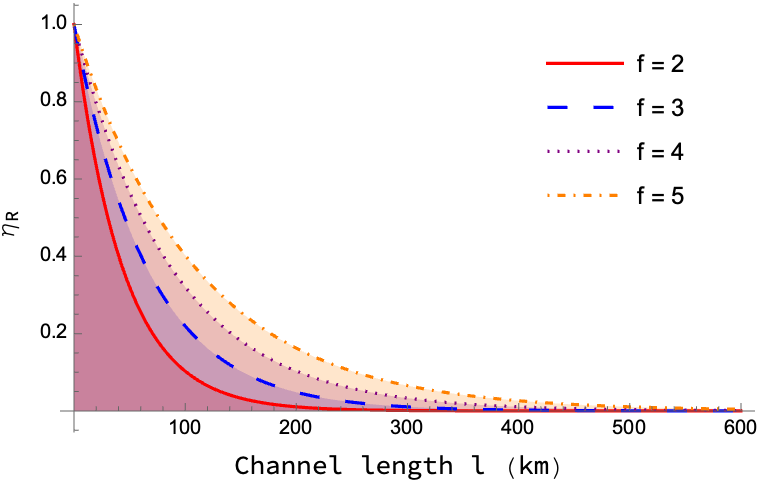}
    \caption{In this figure, we plot the variation of $\eta_R$ as a function of the channel length $l$ (km) (setting $\beta = 0$) for different values of $f$ \textcolor{black}{for a repeater scheme that mitigates the transmission loss by a factor of $1/f$.} (Color online)}
    \label{fig:transmittanceRepeater}
\end{figure}
We see from Eq.~\eqref{eq:critParameterNetworkAdvantage} that using repeaters provides $f-$fold advantage over repeater-less networks. We plot in Fig.~\ref{fig:critParRepeaterAdvantage} the critical length $l_c$ and critical time $t_c$ for sending ebits or private bits over the channel at a rate $p_\ast = 0.5$ for different values of $f$. In the plot, we have set $\alpha = 1/22$ km$^{-1}$, and $\beta = 1/10$ s$^{-1}$.
\begin{figure}
    \centering
    \includegraphics[scale=0.5]{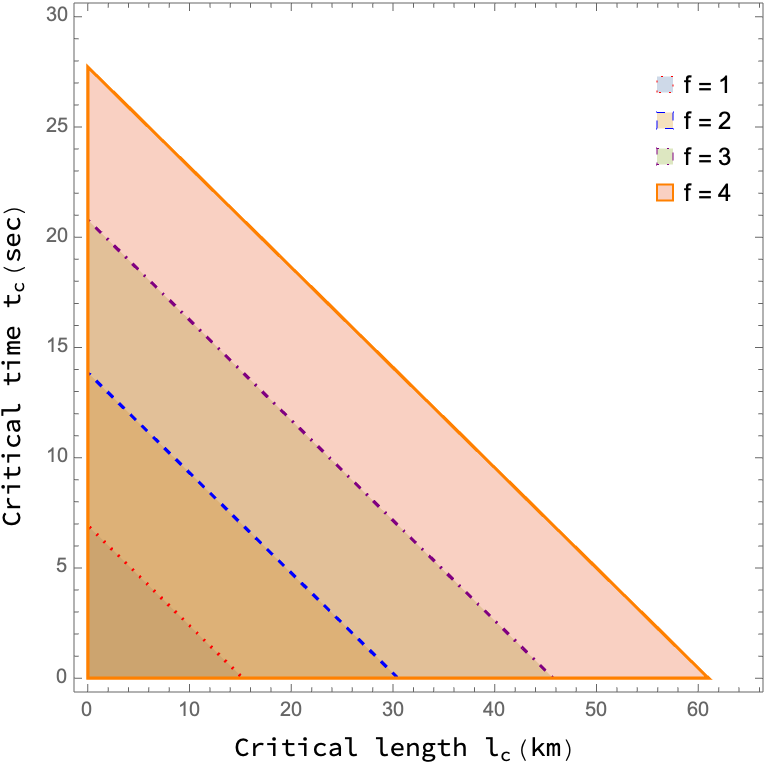}
    \caption{In this figure, we plot the critical length $l_c$ (km) and critical time $t_c$ (sec) for sending ebits or private bits over the channel at a rate of $p_\ast = 0.5$ for different values of $f$. \textcolor{black}{For this, we have considered $\alpha = 1/22 $ km$^{-1}$ and $\beta = 1/10$ sec$^{-1}$.} (Color online)}
    \label{fig:critParRepeaterAdvantage}
\end{figure}

\textcolor{black}{We next consider the task of performing quantum communication over a lattice network with fiber-based elementary links and present limitations on its scalability.}
\begin{observation} \label{obs:graphDiameter}
    Let us require to perform $\mathrm{Task}_\ast$ over the lattice $G_{lat}(\mathbb{V},\mathbb{E})$ having open edges with probability $p_{ij}$ and $0<p_\ast \leq p_{ij} <1$. It follows from Theorem~\ref{lem:percolation} that for any vertex $v \in \mathbb{V}$, the set of vertices connected to it via a network path is finite. This implies $G_{lat}$ has a finite diameter. We may assume the elementary link formed by edge $e_{ij} \in \mathbb{E}$ connecting nodes $\{v_i,v_j\} \in \mathbb{V}$ to be optical fibers of length $L$ having attenuation factor of $\alpha_{dB/km} = 0.22$ dB/km. The fibers having transmittance $\eta = \mathrm{e}^{-\alpha L}$ where $\alpha = 0.051$/km.\footnote{Note that $\alpha = \alpha_{dB/km} \frac{\ln(10)}{10}.$} Let there be some repeater-based scheme that increases the transmittance from $\eta$ to $\eta^{1/f}$ (see Assumption~\ref{assumption:transmittance}). Assuming that performing $\mathrm{Task}_\ast$ by the nodes $\{v_i,v_j\}$ requires transmittance of at least $\epsilon ( = 0.5)$ bounds the length of the fiber to $L \leq (f/\alpha) \ln (1/\epsilon) \approx 27$ km for $f = 2$. If there are $r = 10$ elementary links each of length $L = 27$ km between two nodes separated by a distance $l = rL = 270$ km, then performing $\mathrm{Task}_\ast$ requires $f \geq r L \alpha / \ln(1/\epsilon) \approx 20$.
\end{observation}

\textcolor{black}{In the following example, we present limitations on the scalability of DI-QKD networks assuming there exists some scheme that can mitigate loss due to transmittance over a quantum channel.}
\begin{example} \label{example:DIQKD}
Let us consider the task of connecting end nodes separated by the continental scale of (order of $1000$ km) distance. To enable such a task, let there be a repeater-based network with $t = 10$ elementary links connecting two virtual nodes separated by metropolitan distances (order of $100$ km). Each elementary link has two sources say $S_i$ and $S_j$ producing dual-rail encoded entangled pairs (for details see Appendix~\ref{app:dualRail}) in the state $\Psi^+$. The sources $S_i$ and $S_j$ send one qubit from its entangled pair to the nearest virtual node and the other to the repeater station via optical fibers of length $l = 25$ km having attenuation factor $\alpha = 0.04$ km$^{-1}$. Let the qubits evolve through the fiber as a qubit erasure channel (see Sec.~\ref{app:ErasureChannel}) with channel parameter $\eta_e = \mathrm{e}^{-(\alpha l/f)}$, where we assume that some technique allow us to increase the transmissivity of optical fiber by factor $1/f$ for $f > 1$. After passing through the optical fiber, the qubits are stored in identical quantum memories at the repeater station, the virtual nodes, and the end nodes for $n = 2$ time steps. We model the evolution of the qubits in the quantum memory as depolarising channel (see Sec.~\ref{app:DepolChannel}) having channel parameter $p = 0.01$. Assuming the repeater station performs perfect standard Bell measurement on its share of qubits, the singlet fidelity of the state shared by Alice and Bob is given by $\eta_{0.01}^2~\mathrm{e}^{-2\alpha l t/f}$. The state shared by the end nodes is useful for CHSH-based DI-QKD protocols if $\eta_{0.01}^2~\mathrm{e}^{-2\alpha l t/f} \geq 0.7445$ which requires $f \geq 37$.
\end{example}

Observation~\ref{obs:graphDiameter} and Example~\ref{example:DIQKD} illustrate how far is current technology from designing quantum networks for performing quantum communication and implementing DI-QKD protocols.

In the analysis of networks represented as graphs, it is important to analyze the robustness of the network for different information-processing tasks. In recent works, the robustness has been studied in the context of removal of network nodes~\cite{albert2000error} and has been modelled as a percolation process on networks~\cite{callaway2000network,das2018robust,mohseni2021percolation} represented as graphs. In these studies, the vertices are considered present if the nodes connecting them are functioning normally. In the following subsection, taking motivations from degree centrality~\cite{freeman1979centrality}, betweenness centrality~\cite{freeman1977set} and Gini index~\cite{gini1912variabilita} of network graphs we present figures of merit to compare the robustness of network topologies. We then compare different network topologies based on these measures. 

\subsection{Robustness measure for networks} \label{sec:robustnessMeasure}
Networks that have a large number of edges are more tolerant to non-functioning nodes and edges as compared to those with fewer edges. Taking motivation from the degree centrality of a graph \cite{freeman1979centrality}, we define the link sparsity of a network to assess the performance of a network.
\begin{definition} \label{def:linkSparsity}
Consider a network $\mathscr{N}(G(\mathbb{V},\mathbb{E}))$ having the effective success matrix as $\mathbb{f_{\ast}}.$ Let the number of non-zero entries in $\mathbb{f_{\ast}}$ be $n_{\ast}$. The link sparsity of such a network is given by
\begin{equation} \label{eq:linkSparsity}
    \Upsilon(\mathscr{N}) = 1 - \frac{n_{\ast}}{N_v^2}.
\end{equation}
\end{definition}
Typically it is desirable for the network to have low values of link sparsity. The networks represented as graphs can exist in different topologies and have the same number of nodes, edges and also the same weighted edge connectivity. These networks are said to be isomorphic to one another.  

It follows from Eq.~\eqref{eq:linkSparsity} that isomorphic networks (see Definition~\ref{def:graphIsomorphism}) have the same value of link sparsity. Networks having the same value of link sparsity can differ in the distribution of edge weights. Taking motivation from the betweenness centrality of a graph \cite{freeman1977set}, we define the connection strength of the nodes in the network and the total connection strength of the network.

\begin{definition} \label{def:connectionStrength}
Consider a network $\mathscr{N}(G(\mathbb{V},\mathbb{E}))$ with $|\mathbb{V}| = N_v$ having the adjacency matrix and the effective success matrix as $\mathbf{A_{\ast}}$ and $\mathbb{f_{\ast}}.$ The connection strength of a node $v_i$ in the network $\mathscr{N}$ is given by 
\begin{align} \label{eq:nodeConnectionStrength}
    \zeta_\mathscr{N}(v_i)\coloneqq \begin{cases}
      \bigg(\sum_{\substack{j \\ j \neq i}} 2^{-[\mathsf{A}_{\ast}]_{ij}}\bigg)/N_v & \text{non-cooperative},\\
      \bigg(\sum_{\substack{j \\ j \neq i}} [\mathbb{f}_{\ast}]_{ij} \bigg)/N_v   & \text{cooperative}.
    \end{cases}       
\end{align}
The total connection strength of the network $\mathscr{N}$ in the (non-)cooperative strategy is obtained by adding the connection strengths of the individual nodes and is given by
\begin{equation} \label{eq:totalConnectionStrength}
    \Gamma(\mathscr{N}) = \sum_{j} \zeta_\mathscr{N}(v_j).
\end{equation}
\end{definition}

In the following subsection, we compare the robustness measures for different network topologies.   

\subsection{Comparison of the robustness of network topologies} \label{sec:robustnessComparison}
The topology of a network is defined as the arrangement of nodes and edges in the network. The information processing task that the network is performing decides the topology of the network. Two of the most commonly used network topologies are the star and mesh topologies.

A star network topology \cite{stallings1984local} of $n$ nodes is a level 1 tree with 1 root node and $n - 1$ leaf nodes. A star network with $8$ nodes is shown in Fig.~\ref{fig:starNet}. The root node labelled $1$ is the hub node and acts as a junction connecting the different leaf nodes labelled $2$ to $8$.
    \begin{figure}
        \centering
        \includegraphics[scale=0.4]{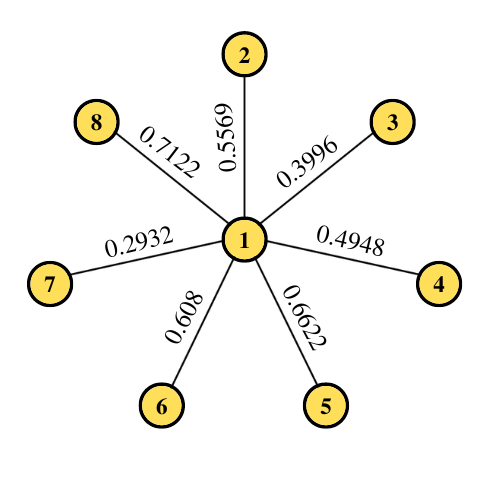}
        \caption{In this figure, we present a star network with $8$ nodes. The node $v_1$ is the hub and nodes $v_2$ to $v_8$ are the outer nodes. The functioning of the hub node is critical to the functioning of the network. (Color online)}
        \label{fig:starNet}
    \end{figure}
Among all network topologies, this network typically requires the minimum number of hops for connecting two nodes that do not share an edge between them. The working of the hub is most critical to the functioning of the star network. The failure of a leaf node or an edge connecting a leaf node to the root does not affect the rest of the network. As an example of its use, this type of network finds application as a router or a switch connecting a ground station to different locations in an entanglement distribution protocol. In such a protocol, an adversary can attack the root node to prevent the proper functioning of the network. If the root node fails to operate, all leaf nodes connected to it become disconnected.

In a mesh network topology \cite{green1980introduction}, each node in the network shares an edge with one or more nodes, as can be seen from Fig.~\ref{fig:meshNet}~and~\ref{fig:meshNetPartial}. There are two types of mesh topologies depending on the number of edges connected to each node. A mesh is called fully connected if each node shares an edge with every other node of the network, as shown in Fig.~\ref{fig:meshNet}. A mesh is called partially connected if it is not fully connected, as shown in Fig.~\ref{fig:meshNetPartial}.  A mesh where each node shares an edge with only one other node of the network is called a linear network.
    \begin{figure}
        \centering
        \begin{subfigure}{0.95\linewidth}
            \centering
            \includegraphics[scale=0.3]{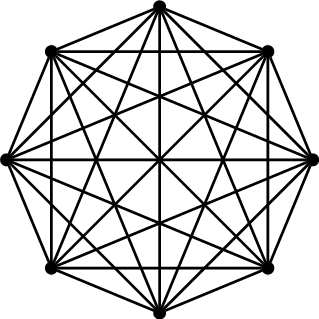}
            \caption{Fully connected mesh with 8 nodes}
            \label{fig:meshNet}
        \end{subfigure}
        \hfill
        \begin{subfigure}{0.95\linewidth}
            \centering
            \includegraphics[scale=0.3
            ]{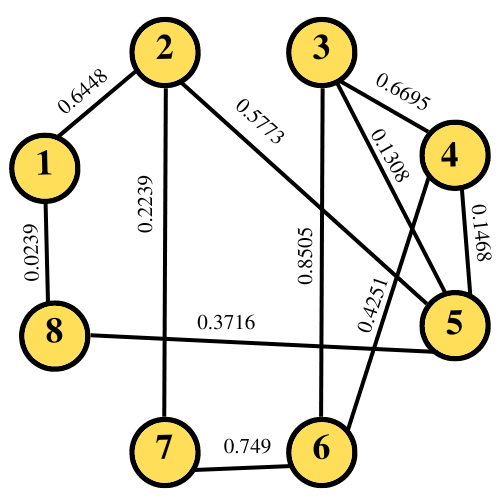}
            \caption{A partially connected mesh with 8 nodes.}
            \label{fig:meshNetPartial}
        \end{subfigure}
    \caption{In this figure, we present (a) fully connected and (b) partially connected mesh networks with $8$ nodes. In a fully connected mesh network, there are edges between every pair of nodes. The presence of multiple paths between two nodes of the mesh network makes it more robust as compared to a star network with the same number of nodes. (Color online)}
    \label{fig:meshNetworkTypes}
    \end{figure}
In a partially or fully connected mesh, the presence of multiple paths between two nodes of the network makes it robust. As an example of its use, a mesh network can be used for a satellite-based entanglement distribution, which we discuss in detail in a later section.

Let us consider a mesh network $\mathscr{N}_m(G_m(\mathbb{V},\mathbb{E}))$ of diameter $d$ with $|\mathbb{V}| = N_v$ and for $(v_i,v_j) \in \mathbb{E},$ $p_{ij} = p.$ Let there be a non-cooperating strategy for sharing resources between the nodes of the network. For such a strategy, the rows of the adjacency matrix of the network have $N_z$ number of zero entries where $0 \leq N_z \leq N_v -1.$ Next, let there also exist a cooperating strategy for sharing resources between the nodes of this network. For such a strategy, the node $v_i$ of the network does not share an edge with $N'_z$ number of nodes where $0 \leq N'_z \leq N_z \leq N_v -1.$ From the remaining nodes, there exists edges between $v_i$ and $(N_v - N'_z - 1)/d$ number of nodes with $p_{ij} = p^{j}$ where $1 \leq j \leq d.$
For such a network, we have the link sparsity as
\begin{align}
    \Upsilon(\mathscr{N}_m) \coloneqq \begin{cases}
    (N_z+1)/N_v & \text{non-cooperating,}\\
    (N'_z+1)/N_v & \text{cooperating.}
    \end{cases}
\end{align}
The connection strength of the node $v_i$ is given by
\begin{align}
    \zeta_{\mathscr{N}_m} (v_i) \coloneqq \begin{cases}
    \big[1 + p(N_v - N_z)\big] / N_v & \text{non-cooperating,}\\
    \big[1 + \frac{p(p^d - 1)(N_v - N'_z - 1)}{d (p-1)}\big] / N_v & \text{cooperating.}
    \end{cases}
\end{align}
The total connection strength of the network is given by $\Gamma(\mathscr{N}_m) = N_v \times \zeta_{\mathscr{N}_m}(v_i)$. It can be seen that for this network the sparsity index $\Xi(\mathscr{N}_m)$ (see Eq.~\eqref{eq:sparsityIndex}) is the same as the connection strength of the node. This follows from the equal distribution of the weights in the network. 

Next consider a star network $\mathscr{N}_s(G_s(\mathbb{V},\mathbb{E}))$ with $|\mathbb{V}| = N_v$ and for $(v_i,v_j) \in \mathbb{E}$ we have
\begin{align} \label{eq:starNet}
    p_{ij} \coloneqq \begin{cases}
      p & \text{if $i = 1$ and $p \geq p_{\ast}$},\\
      0 & \text{otherwise}.
    \end{cases}       
\end{align}
The link sparsity of such a network is $\Upsilon(\mathscr{N}_s) = 1 - (1/N_v)$ for the cooperative strategy and $\Upsilon(\mathscr{N}_s) = 1 - [(3 N_v - 2)/N_v^2]$ for the non-cooperative strategy. The connection strength of the root node is $\zeta_{\mathscr{N}_s}(v_i) = (1 + p(N_v-1))/N_v,$ while for the leaf nodes is  
\begin{align}
    \zeta_{\mathscr{N}_s}(v_i) \coloneqq \begin{cases}
      (1+p)/N_v & \text{non-cooperative},\\
      [1 + p + p^2(N_v - 2)]/N_v & \text{cooperative.}
    \end{cases}           
\end{align}
The sparsity index (see Eq.~\eqref{eq:sparsityIndex}) of the network is
\begin{equation}
    \Xi(\mathscr{N}_s) \coloneqq \begin{cases}
      \big[N_v^2 + p(N_v^2 \\ \qquad \qquad \quad + N_V - 2)\big] / N_v^3  & \text{non-cooperative},\\
      \big[N_v^2 +  p(N_v-1)(N_v+2)  \\
      \qquad+ p^2(N_v-2)(N_v-1)\\ \qquad \qquad(N_v+1)\big] / N_v^3  & \text{cooperative.}
    \end{cases}            
\end{equation}
As another example, consider a network $\mathscr{N}(G(\mathbb{V},\mathbb{E}))$  having $|V| = N_v$ number of nodes, and each node shares an edge with $d$ other nodes. The effective adjacency matrix $\mathsf{A}_{\ast}$ of such a network is a circulant matrix. Let the success probability of transferring a resource between the node $1$ and the node $n$ be given by,
\begin{align}
    p_{1n}\coloneqq \begin{cases}
      p^n & \text{if $n \leq d/2$ and d = even},\\
      p^{d - n + 1} & \text{if $n > d/2$ and d = even},\\
      p^n & \text{if $n \leq (d+1)/2$ and d = odd},\\
      p^{d - n + 1} & \text{if $n > (d+1)/2$ and d = odd}.
    \end{cases}       
\end{align}
The first row of the adjacency matrix is formed by calculating the weights $w_{\ast}(e_{1n}) = - \log p_{1n}.$ The $k^{\text{th}}$ row is formed by taking the cyclic permutation of the first row with an offset equal to $k.$ In the non-cooperative strategy, the link sparsity of the network is then given by 
\begin{equation}
    \Upsilon(\mathscr{N}) = 1 - \bigg(\frac{d+1}{N_v}\bigg),
\end{equation}
and the connection strength of the node $v_i$ is given by
\begin{align}
    \zeta_\mathscr{N}(v_i)\coloneqq \begin{cases}
     \frac{(1+p)(p^{(d+1)/2}-1)}{N_v (p-1)}  & \text{if $d$ is odd},\\
      \frac{1}{N_v} \bigg[1 + \frac{2p(p^{d/2}-1)}{(p-1)}\bigg] & \text{if $d$ is even}.
    \end{cases}       
\end{align}
In the following subsection, we introduce measures for identifying the critical nodes of a given network. 
\subsection{Critical nodes in a network} \label{sec:critNodesDef}
The critical nodes of a network are the nodes that are vital for the proper functioning of the network. Removing any of these nodes can lead to some of the other nodes in the network being disconnected. Given a network $\mathscr{N}(G)$, we proceed to define a measure for the criticality of the nodes in $G$. For this, at first, taking motivation from \cite{watts1998collective}, we define the clustering coefficient for the nodes of a given network. 
\begin{definition}
For a network $\mathscr{N}(G(\mathbb{V},\mathbb{E}))$, let $G_i(\mathbb{V}_i,\mathbb{E}_i)$ be a sub-graph of $G$ formed by the neighbours of node $v_i \in \mathbb{V}.$ Let $n_i = |\mathbb{V}_i|$ be the number of nodes present in $G_i$ and $e_i = |\mathbb{E}_i|$ be the number of edges present in $G_i$ with $p_{ij}\geq p_{\ast}$. The clustering coefficient of the node $v_i$ is defined as
\begin{align}\label{eq:clusterCoeff}
    C_i= \frac{2~e_i}{n_i (n_i - 1)}.
\end{align}
\end{definition}
The average clustering coefficient of a network is calculated by taking the average of $C_i$ for all the nodes of the network. We next proceed to define the average effective weight of a network using Eq.~\eqref{eq:weight}. 
\begin{definition}
For a network $\mathscr{N}(G(\mathbb{V},\mathbb{E}))$, the average effective weight of a network denoted by $\widetilde{\mathrm{w}}_{\ast}$ is defined as the mean of the effective weight between all the node pairs in the network and is expressed as
\begin{equation} \label{eq:globalEff}
    \widetilde{\mathrm{w}}_{\ast} (G) = \frac{1}{\textcolor{black}{N_v(N_v-1)}} \sum_{\substack{v_i,v_j\in \mathbb{V} \\ i \neq j }} \mathrm{w}_{\ast}({e_{i\to{...}\to{j}}}),
\end{equation}
where $\mathrm{w}_{\ast}({e_{i\to{...}\to{j}}})$ is the effective weight associated with the path connecting the nodes $v_i$ and $v_j$ \textcolor{black}{The term $N_v(N_v-1)$ represents the total number of possible ordered pairs of distinct nodes in the network $\mathscr{N}(G(\mathbb{V},\mathbb{E}))$ having $N_v$ nodes}. 
\end{definition}
When two nodes $(v_i,v_j)$ are disconnected, the effective weight $\mathrm{w}_{\ast}({e_{i\to{...}\to{j}}})$ becomes infinite. Small values of $\widetilde{\mathrm{w}}_{\ast}(G)$ indicate that the network performs the task with high efficiency. 

Taking motivation from \cite{freeman1977set}, we define centrality for the nodes of a given network. Then, using definitions of centrality, average effective weight and clustering coefficient, we define the critical parameter for the nodes of a given network. 
\begin{definition} \label{def:criticalParameter}
For a network $\mathscr{N}(G(\mathbb{V},\mathbb{E}))$, let us denote the shortest path connecting the node pairs $(v_i,v_j) \in G$ as $d_{ij} \in \mathbb{D}$. We define the centrality $\tau_{i}$ of the node $v_i$ as the number of paths belonging to the set $\mathbb{D}$ in which the node $v_i$ appears as a virtual node. The critical parameter associated with the node $v_i$ is defined as 
\begin{equation} \label{eq:Pcrit}
    \nu_i = \frac{\tau_i}{C_i \widetilde{w}_{\ast}(G_i)}.
\end{equation}
\end{definition}
The critical nodes of a graph have high values of $\nu_i$. These nodes of the network are essential for the proper functioning of the network. If one of these nodes is removed, it will reduce the success probability of sharing $\chi$ between pairs of network nodes for performing $\mathrm{Task}_\ast$. We present a heuristic algorithm in Sec.~\ref{app:critNetworkNodes} to obtain the critical parameter $\nu_i$ for the node $v_i \in \mathbb{V}$ of the network $\mathscr{N}(G(\mathbb{V},\mathbb{E}))$.

\begin{example}
Let us obtain the critical nodes of the graph with $8$ nodes shown in Fig. \ref{fig:sampleGraph1}. 
\begin{figure}
    \centering
    \includegraphics[scale=0.5]{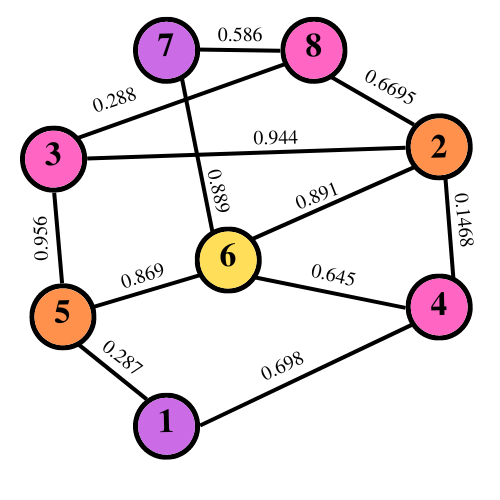}
    \caption{A network represented by a weighted graph with $8$ nodes and $12$ edges. Node $v_6$ (shown in yellow) is the most critical node followed by nodes $v_2$ and $v_5$ (shown in orange). These nodes are the most critical for the proper functioning of the network. (Color online)}
    \label{fig:sampleGraph1}
\end{figure}
We present the critical parameter for the different nodes of the network in the table below.\\
\begin{table}[H] \label{table:critParNodes}
\centering
 \begin{tabular}{ c c || c c}  
 \hline
 node number $(i)$  & $\nu_i$ & node number $(i)$  & $\nu_i$ \\ [0.5ex] 
 \hline\hline
 $v_1$ & 0.1714 \hspace{5pt} & $v_5$ & 1.5238 \hspace{5pt}\\
 $v_2$ & 1.6667 \hspace{5pt} & $v_6$ & 2.5714 \hspace{5pt} \\
 $v_3$ & 0.7619 \hspace{5pt} & $v_7$ & 0.1714 \hspace{5pt}\\
 $v_4$ & 0.9523 \hspace{5pt} & $v_8$ & 0.5714 \hspace{5pt}\\ 
 \hline
 \end{tabular}
\end{table}
We observe that node $v_6$ is the most critical followed by node $v_2$ and node $v_5$. We label these nodes as the critical nodes of the network.
\end{example}

Quantum network architectures have potential applications in diverse fields ranging from communication to computing. In the following section, we present the bottlenecks of some potential use of quantum networks for real-world applications.

\section{Real world instances} \label{sec:realWorldInstances}
\textcolor{black}{In this section, we first present practical bottlenecks in the distribution of entangled pairs between two far-off cities based on current experimental primitives.} We then consider different currently available as well as near-future quantum processor architectures as networks and observe their robustness parameters. Then we propose networks for (a) the major international airports to communicate via a global quantum network and (b) the Department of Energy (DoE) at Washington D.C. to communicate with the major labs involved in the National Quantum Initiative (NQI) by sharing entanglement. We present the practical bottlenecks for such models. 

\subsection{Entanglement distribution between cities} \label{sec:entanglementDistributionCities}
Consider a satellite-based network for sharing entangled pairs between two far-off cities. Let such a network be a two-layered model consisting of a global scale and a local scale. On a global scale, there are multiple ground stations located across different cities. Such ground stations are interconnected via a satellite network. Two ground stations share an entangled state using the network via the shortest network path between them (see Sec.~\ref{sec:shortestPath} for an illustration of shortest network-path finding algorithm). 
\begin{figure}
    \centering
    \includegraphics[scale=0.45]{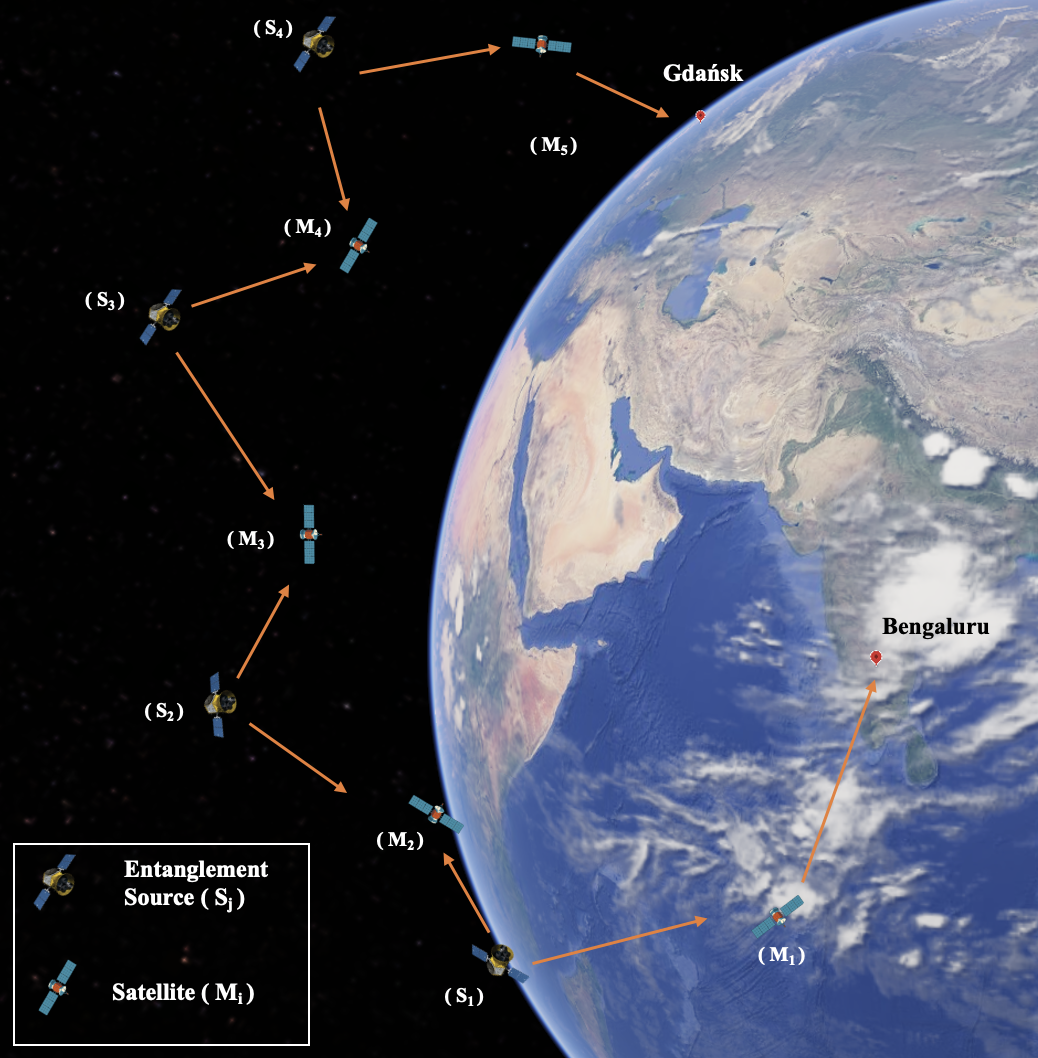}
        \caption{In this figure, we present the shortest network path between the ground stations at Bengaluru and Gda\'nsk via the global satellite network. The entangled sources are marked as $S_i$ and the satellite stations are marked as $M_i$. The shortest path has $6$ entangled sources and $5$ satellite stations. The image was created using the Google Earth software \cite{Google_2022}. (Color online)} 
        \label{fig:networkSystemRoute}
\label{fig:globalArchitecture}
\end{figure}

\begin{figure}[H]
\centering
\begin{subfigure}{0.49\linewidth}
    \includegraphics[width=0.99\textwidth]{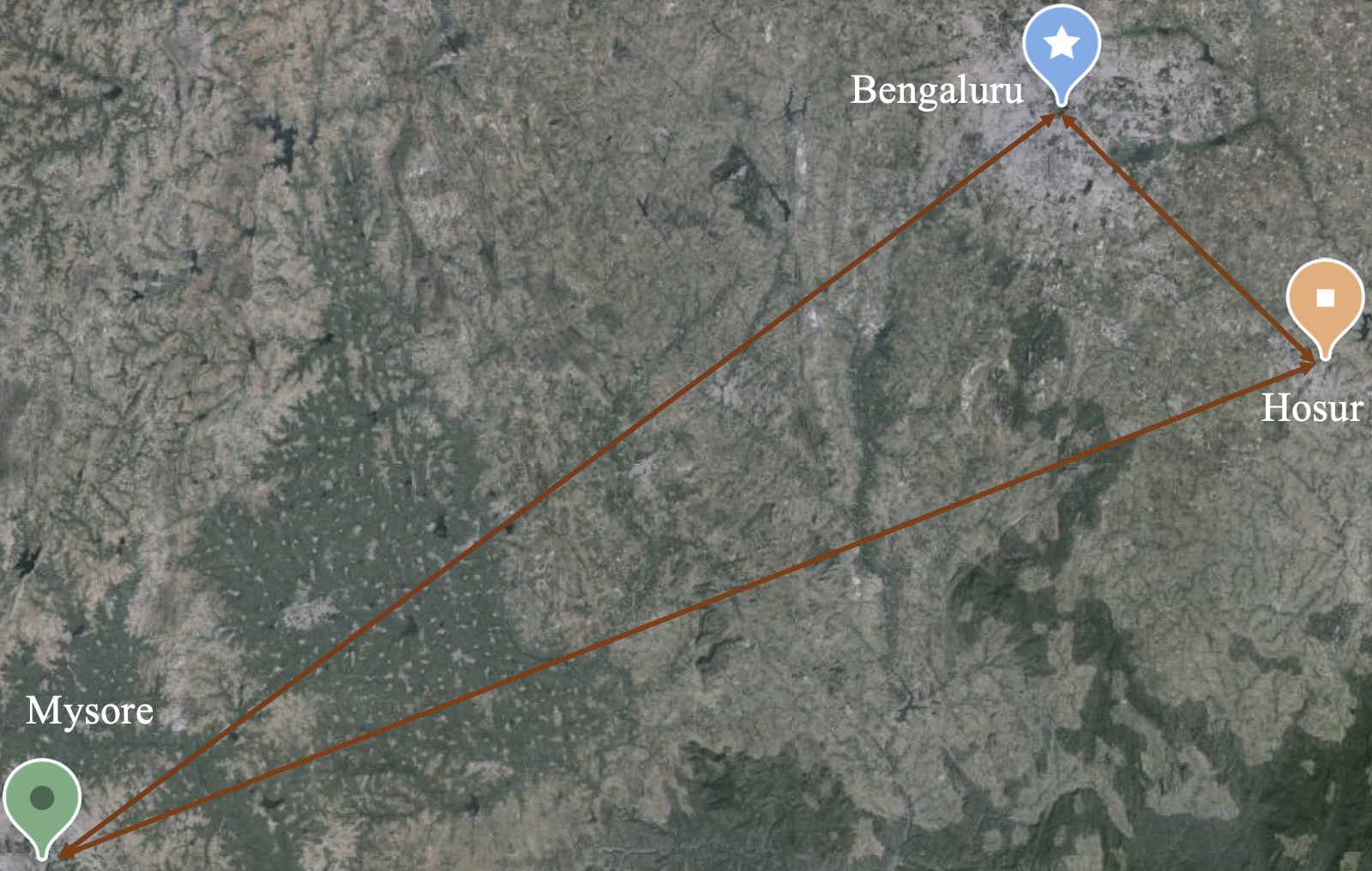}
    \caption{Bengaluru ground station}
    \label{fig:localArchitectureBlr}
\end{subfigure}
\hfill
\begin{subfigure}{0.49\linewidth}
    \includegraphics[width=0.95\textwidth]{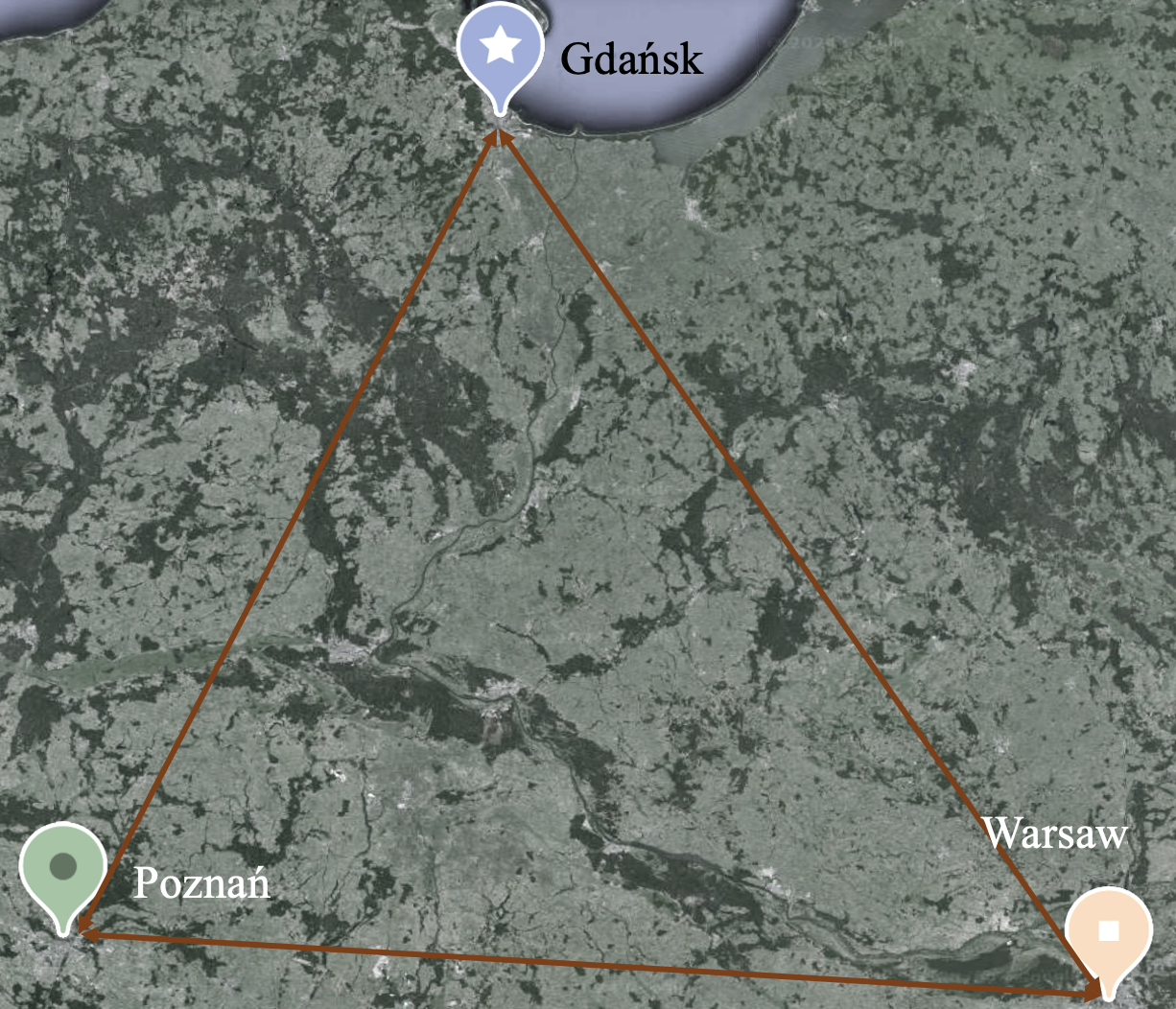}
    \caption{Gda\'nsk ground station}
    \label{fig:localArchitectureGdansk}
\end{subfigure}
\caption{In this figure, we present the local scale network architectures at (a) Bengaluru and (b) Gda\'nsk for sharing entangled pairs across nearby localities. The ground stations are connected to multiple localities via optical fibers (shown in black and orange lines). The images were created using Google Earth software \cite{Google_2022}. (Color online)}
\label{fig:networkArchitectureBetweenCities}
\end{figure}
As an example, we show in Fig.~\ref{fig:networkSystemRoute} the shortest path connecting the ISTRAC ground station located at Bengaluru to the ground station at Gda\'nsk via the global satellite network. On the local scale, different localities (end nodes) are connected to their nearest ground station via optical fibers. We show in Fig.~\ref{fig:networkArchitectureBetweenCities} the local scale network for the ground station located at Bengaluru and Gda\'nsk. The ground station at Bengaluru is connected to the localities of Hosur and Mysore. The ground station at Gda\'nsk is connected to the localities of Pozna\'n and Warsaw. 

In the satellite network, there are sources $(S_j)$ producing entangled photonic qubit pairs in the state $\Psi^+$. The source then sends the photons belonging to a pair to different neighbouring satellites via a quantum channel as shown in Fig.~\ref{fig:networkSystemRoute}. We model the quantum channel between the ground station and the satellite at the limits of the atmosphere as a qubit thermal channel (see Sec.~\ref{app:ThermalChannel} and Sec.~\ref{app:atmosphericChannel}) and that between two satellites as an erasure channel (see Sec.~\ref{app:ErasureChannel}) having erasure parameter $\eta_e$. The erasure channel parameter is assumed to be identical throughout the network. The satellite stations $M_j$ outside the limits of the atmosphere perform standard Bell measurement with success probability $q$ on their share of the qubits that they received from their neighbouring source stations. The satellites at the boundary of the atmosphere transmit their share of qubit via the atmospheric channel to the ground stations (which we call local servers). The ground stations on receiving the state store it in a quantum memory. In the quantum memory, the state evolves via a depolarising channel (see Sec.~\ref{app:DepolChannel}). The local servers distribute the quantum states to different localities (which we call clients) on request using optical fibers as can be seen in Fig.~\ref{fig:networkArchitectureBetweenCities}.   

In future, it may be that India and Poland establish communication links that share each halves of the entangled pairs between a designated hub in each country so that they can perform desired quantum tasks in collaboration. Let us assume that the Indian Space Research Organization (ISRO) headquarter\footnote{As we were finalizing the paper, we learnt that ISRO was successful in soft landing of its spacecraft \textbf{Chandrayaan-3} (Vikram lander and Pragyan rover) on the Moon's south polar region on 23-08-2023 at 18:03 IST.} at Bengaluru would like to perform delegated quantum computing~\cite{childs01,GLM08,BFK09,fitzsimons2017private,BKB+12} by securely accessing the IBM Quantum Hub at Pozna\'n~\cite{IBMPol}. For this, the ISRO headquarter can share entangled system with the IBM Quantum Hub via the shortest route in the satellite-based network. For an illustration, let the shortest network path between the ground stations at Bengaluru and Gda\'nsk have $5$ satellites and $4$ entangled photon sources as shown in Fig.~\ref{fig:networkSystemRoute}. The entanglement yield of the network is given by
\begin{equation}
    \xi_{\text{avg}} = q^3~\eta_t^{G} (\eta_e^2)^3,
\end{equation}
where $\eta_t^{G}$ is obtained from Eq.~\eqref{etat} and takes into account the local weather conditions at Bengaluru and Gda\'nsk. In the general network with $n$ satellite-to-satellite links between the two ground stations, the average entanglement yield is given by 
\begin{equation}
    \xi_{\text{avg}} = \eta_t^G (\eta_e^2)^{n-1}q^{n-1}.
\end{equation}
The ground stations store the incoming qubits in different quantum memory slots and serve the receiving traffic\footnote{Traffic is the flow of photons between the nodes of the network for enabling the network to perform a specific information processing task.} requests from different local clients following queuing discipline (see Algorithm~\ref{algo:resourceAllocation} for details of the incoming and outgoing traffic threads\footnote{Thread is a sequential execution of tasks in a process.}). The total number of memory slots available in the quantum memory is fixed. The evolution of the stored qubits in the quantum memory is modelled via a depolarising channel with channel parameter $p$. We model the quantum memory as a max-heap data structure (see Definition~\ref{def:maxHeap}) with the key as the fidelity of the stored quantum state. If the fidelity of any quantum state stored in the memory drops below a pre-defined critical value, $\eta_{\text{crit}}$, that state is deleted from the memory. The value of $\eta_{\text{crit}}$ is determined by the task or the protocol that the end parties may be interested in performing using their shared entangled state. On receiving a connection request from a single local client, the ground station transmits the latest qubit that it has received as outward traffic. Now when the ground station receives traffic requests from multiple local clients, there is the problem of optimizing the traffic flow\footnote{Traffic flow is a sequence of quantum states that is sent from the ground station to the local station.}. For such a flow problem, we introduce the following modified fair queuing algorithm.

Let us define $t_p^{(i)}$ as the time to process the $i^{\text{th}}$ quantum state in the memory, $t_i^{(i)}$ as the starting time for the transmission from the memory and $t_f^{(i)}$ as the time when the state has been transmitted from the memory. We then have,
\begin{equation}
    t_f^{(i)} = t_i^{(i)} + t_p^{(i)}.
\end{equation}
Now, there is a possibility that the state has arrived at the memory before or after the processing of $i-1$ states in this heap. In the latter case, the state arrives at an empty heap memory and is transmitted immediately if there is a traffic request. In the other case, it swims through the heap depending on its fidelity and is stored in the memory. Let us denote $t_r^{(i)}$ as the time required for the node to swim up to the root node from its current position in the heap. Then we have, 
\begin{equation}
    t_f^{(i)} = \max\bigg(t_f^{(i-1)},t_r^{(i)}\bigg) + t_p^{(i)},
\end{equation}
where $t_f^{(i-1)}$ is the time required for processing $(i-1)^\text{th}$ quantum state. If there are multiple flows, the clock advances by one tick when all the active flows receive one state following the qubit-by-qubit round-robin basis. If the quantum state has spent $s$ time steps in the memory then the average entanglement yield is given by
\begin{equation}
    \xi_{\text{avg}} \coloneqq \begin{cases}
        \eta_d^s \hspace{3pt} \eta_t^G (\eta_e^2)^{n-1} q^{n-1} & \text{if }\eta_d^s > \eta_{\text{crit}},\\
        0 & \text{otherwise,}    
    \end{cases}
\end{equation}
where $\eta_d^s$ is obtained from Eq.~\eqref{app:depolYield} and takes into account the loss in yield per time step in the quantum memory. Let us assume the ground station at Bengaluru and Gda\'nsk transmits the state via identical fibers to the ISRO headquarter and Pozna\'n, respectively. Considering the fiber losses at the two ground stations given by $\mathrm{e}^{-\alpha l_B}$ and $\mathrm{e}^{-\alpha l_M}$, the sources producing the state $\Psi^+$ with probability $\eta_s$, and assuming that the quantum state has spent $s$ time steps in the memory, the average entanglement yield is given by
\begin{equation}
    \xi_{\text{avg}} = \eta_d^s \hspace{3pt} \eta_t^G (\eta_e^2)^{n-1} (\eta_s)^{n - 1} \mathrm{e}^{-\alpha (l_B + l_M)} q^{n-1},
\end{equation}
where $l_B$ and $l_M$ are the lengths of the fibers from ISRO headquarter and Pozna\'n to the Bengaluru and Gda\'nsk ground stations respectively. Inserting $\eta_d^s$ from Eq.~\eqref{app:depolYield} and $\eta_t^G$ from Eq.~\eqref{etat}, we have the yield given by
\begin{eqnarray} \label{eq:avgYieldSatelliteNetwork}
\xi_{\text{avg}} &&=  \mathrm{e}^{-\alpha(l_B + l_M)} \left(\eta_e^2\right)^{n-1} \eta_s^{n - 1} q^{n-1} \nonumber \\
&&\bigg[(1-p)^{2 s}-\frac{1}{4} (p-2) p \left((s-1) (1-p)^{2 (s-1)}+1\right)\bigg] \nonumber \\
&&\qquad \qquad \bigg[\kappa_g(\kappa_g-1)(\eta_g-1)^2+\frac{1}{2}(1+\eta_g^2)\bigg] 
\end{eqnarray}
where $\eta_g$, $\kappa_g$ take into account the local weather condition at Bengaluru and Gda\'nsk and $s$ is the number of applications of depolarising channel in the quantum memory.
\begin{figure}
    \centering
    \includegraphics[scale=0.85]{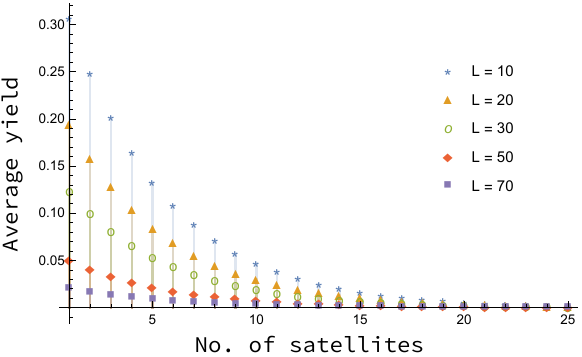}
    \caption{In this figure, we plot the average yield $\xi_{\text{avg}}$ (see Eq.~\eqref{eq:avgYieldSatelliteNetwork}) as a function of the number of satellites in the network for different values of the total optical fiber length $L = l_B + l_M$ (shown figure inset). \textcolor{black}{We set the parameters $\eta_s = 0.9, s = 1, p = 0.1, \eta_e = 0.95, \eta_g = 0.5, \kappa_g = 0.5, \alpha = 1/22$ km$^{-1}$, $q = 1$.} (Color online)}
    \label{fig:sampleGraph99}
\end{figure}
We plot in Fig.~\ref{fig:sampleGraph99}, the average entanglement yield $\xi_{\text{avg}}$ of the two end nodes connected by the network as a function of the number of satellite-to-satellite links $n$ between their nearest ground stations for different values of the total optical fiber length $L = (l_B + l_M)$. We observe that for a fixed value of $L$, $\xi_{\text{avg}}$ decreases with an increase in $n$. Also, for a fixed value of $n$, $\xi_{\text{avg}}$ decreases with increase in $L$. 
\begin{figure}
\centering
\includegraphics[scale=0.8]{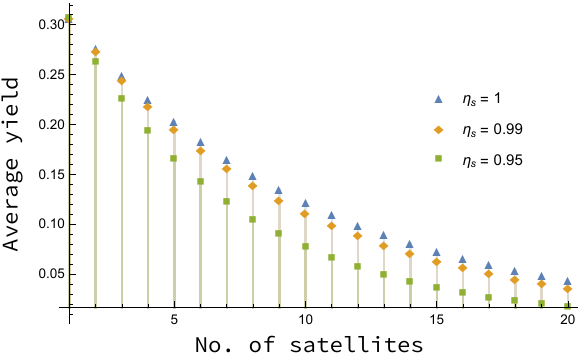}
    \caption{In this figure, we plot the average yield $\xi_{\text{avg}}$ (see Eq.~\eqref{eq:avgYieldSatelliteNetwork}) as a function of the number of satellites in the network for single photon source architectures. The single photon source architectures have the source efficiencies (a) $\eta_s = 0.95$ (b) $\eta_s = 0.99$ and (b) $\eta_s = 1.$ \textcolor{black}{For this, we set the parameters $L = l_B + l_M = 10$ km$, s = 1, p = 0.95, \eta_e = 0.95, \eta_g = 0.5, \kappa_g = 0.5, \alpha = 1/22$ km$^{-1}$, $q = 1$.} (Color online)}
\label{fig:sampleGraph15}
\end{figure}
Furthermore, we plot in Fig.~\ref{fig:sampleGraph15}~and~\ref{fig:sampleGraph19} the variation in $\xi_{\text{avg}}$ as a function of $n$ for different values of $\eta_s$. We observe that for a given $\eta_s$, $\xi_{\text{avg}}$ decreases with increase in $n$. Also, we observe that Quantum dot-based, atom-based, and SPDC-based entangled photon sources are best suited for the entanglement distribution network. Finally, we plot in Fig.~\ref{fig:sampleGraph999}, the variation in $\xi_{\text{avg}}$ as a function of $n$ for different values of $q$. We observe that for a given $q$, $\xi_{\text{avg}}$ decreases with increase in $n$. Also, $\xi_{\text{avg}}$ decreases with a decrease in $q$.

\begin{figure}
    \centering
    \includegraphics[scale=0.85]{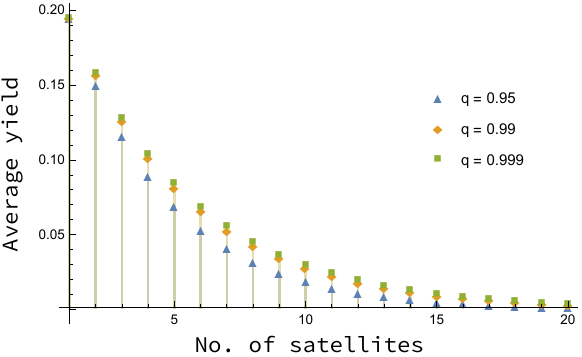}
    \caption{In this figure, we plot the average yield $\xi_{\text{avg}}$ (see Eq.~\eqref{eq:avgYieldSatelliteNetwork}) as a function of the number of satellites in the network for different values of the success probability of the standard Bell measurement denoted by $q$ (shown figure inset). \textcolor{black}{We set the parameters $\eta_s = 0.9, s = 1, p = 0.1, \eta_e = 0.95, \eta_g = 0.5, \kappa_g = 0.5, \alpha = 1/22$ km$^{-1}$, $L = l_B + l_M = 20$ km.} (Color online)}
    \label{fig:sampleGraph999}
\end{figure}

Our methods in general apply to sharing multipartite entangled states among different ground stations distributed at different geographical locations across the globe. To observe this, note that if certain users of the network share bipartite entangled states, then such states can be used to distill multipartite entangled states with the use of ancilla and entanglement swapping protocols~\cite{ZHW+97,de2020protocols,das2021universal}.

\subsection{Quantum processors} \label{sec:processorNetworks}
The currently available quantum processors (QPUs) use different technologies to implement the physical processor. The processors of IonQ and Honeywell utilize a trapped ion-based architecture, while IBM, Rigetti, and Google have a superconducting architecture. The superconducting architecture requires physical links between qubits that are to be entangled, while the trapped ion-based architecture does not have any topological constraint. In this section, we model the different quantum processor architectures as graphical networks (see Table~\ref{Table:Processor}) and present the robustness measures (defined in Sec.~\ref{sec:robustnessMeasure}) for such networks.   
\begin{table}
\centering
 \begin{tabular}{ l | l | l | l | l }
 \hline
Name & Layout & Fidelity & Qubits & $T_1$ \\ [0.5ex]
 \hline\hline
\makecell{Sycamore \\ (Google) \\ \cite{Google}} & \makecell{square \\ lattice} & \makecell{$96.9\%$ (RO) \\ $99.85\%$ (1Q) \\ $99.64\%$ (2Q)} & 54 & 15  $\mu s$\\ [0.5ex]
 \hline
\makecell{Eagle \\ (IBM) \\ \cite{IBM,BGC23}} & \makecell{heavy \\ hexagonal} & \makecell{$99.96\%$ (RO) \\ $99.99\%$ (1Q) \\ $99.94\%$ (2Q)} & 127 & 95.57  $\mu s$\\ [0.5ex]
\hline
\makecell{Aspen-M-2 \\ (Rigetti) \\ \cite{Rigetti}} & \makecell{octagonal} & \makecell{$97.7\%$ (RO) \\ $99.8\%$ (1Q) \\ $90\%$ (2Q)} & 80 & 30.9  $\mu s$\\ [0.5ex]
\hline
\end{tabular}
\caption{The performance details of different quantum processors. In the above table, the second column provides the arrangement of the qubits in the processor. The third column provides the 1 qubit (1Q), 2 qubits (2Q), and the readout (RO) fidelity of the processors. The fourth column provides the total number of qubits in the processor, and the fifth column provides the thermal relaxation time $(T_1)$ of the qubits of the processor.}
\label{Table:Processor}
\end{table}

We consider different qubit quantum processor network architectures in square, heavy-hexagonal and octagonal layouts. The link sparsities of each unit cell\footnote{Unit cell is the smallest group of processor qubits which has the overall symmetry of the processor, and from which the entire processor can be constructed by repetition.} for these different network layouts are given by
\begin{align} \label{eq:sparsityTopology}
    \Upsilon\mathscr({N}) \coloneqq \begin{cases}
    1 - \big(\frac{16}{64}\big) \approx 0.75 & \text{octagonal,}\\
    1 - \big(\frac{8}{16}\big) \approx 0.5 & \text{square,} \\
    1 - \big(\frac{24}{144}\big) \approx 0.833 & \text{heavy hexagonal.}
    \end{cases}
\end{align}
We observe that the square structure has the lowest link sparsity, followed by octagonal and heavy hexagonal structures. Next, let the edges present in these network layouts have success probability $p$. 

The connection strength of the $i^{th}$ node in the unit cell of octagonal, heavy hexagonal and square network for a non-cooperative strategy is given by
\begin{align}
    \Gamma(\mathscr{N}) \coloneqq \begin{cases}
    p/4 & \text{octagonal,}\\
    p/2 & \text{square,} \\
    p/6 & \text{heavy hexagonal.}
    \end{cases}
\end{align}
We plot in Fig.~\ref{fig:totalConnectionStrengthProcessor} the connection strength of the $i^{th}$ node for different values of the success probability of edge. We observe that the connection strength for a given success probability of edge is highest for square networks, followed by octagonal and heavy-hexagonal networks. 
\begin{figure}
    \centering
    \includegraphics[scale=0.5]{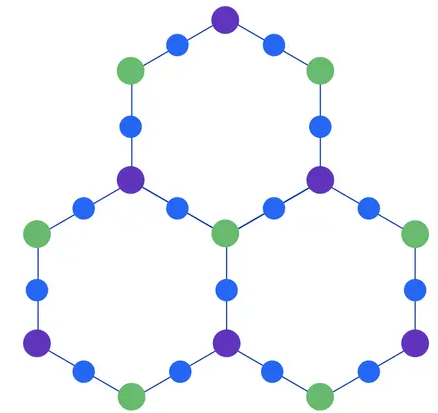}
    \caption{A slice of a quantum processor model based on heavy-hexagonal structure discussed in~\cite{NPC21}. The link sparsity of the unit cells in the network is $0.833$, and the critical nodes of the network slice are shown in green and violet. (Color online)}
    \label{fig:IBM1008}
\end{figure}

We propose a 1024-node square lattice-based quantum processor network architecture represented as a $32 \times 32$ lattice. 
\begin{figure}
    \centering
    \includegraphics[scale=0.4]{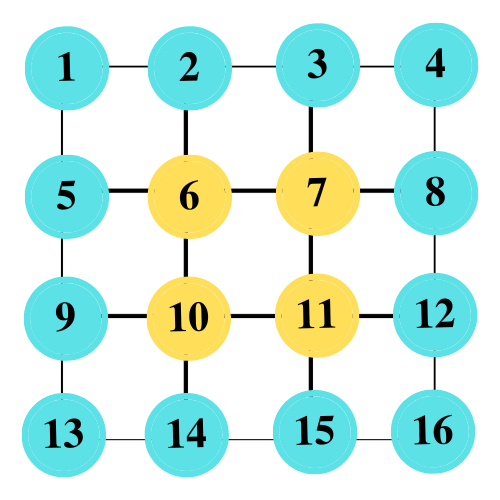}
    \caption{A $4 \times 4$ slice of a $1024$ node quantum processor architecture based on square structure. The total network layout is represented as a $32 \times 32$ lattice. The link sparsity of the network is $0.9962$, and the critical nodes of the network slice are shown in yellow. (Color online)}
    \label{SquareLattice}
\end{figure}
We show in Fig.~\ref{SquareLattice} a $4 \times 4$ slice of the lattice as a representation of the entire quantum processor. The link sparsity of the $1024$ node square network is $0.9962$. The nodes shown in yellow in Fig.~\ref{SquareLattice} are identified as critical nodes. We observe in Fig.~\ref{SquareLattice} that there are three types of nodes in the network based on the number of edges that are connected to the node. We call a node a corner node, edge node, and inner node if it shares an edge with two, three, and four other nodes, respectively. The connection strength of these three types of nodes is given by
\begin{align}
    \zeta_\mathscr{N}(v_i) \coloneqq \begin{cases}
    p/256 & \text{inner node,}\\
    3p/1024 & \text{boundary node,} \\
    p/512 & \text{corner node.} 
    \end{cases}
\end{align}

\subsection{Network propositions} \label{sec:networkPropositions}
Let us first consider a network connecting all the major airports in the world. We say that two airports are connected if there exists at least one commercial airline currently operating between them. We consider a network consisting of 3463 airports all over the globe forming the nodes of the network and 25482 edges or airline routes between these airports~\cite {AirportNetwork}. For such a network, we observe that the longest route is between Singapore, Changi International Airport, and New York John F. Kennedy International Airport, in the United States, with a distance of approximately 15331 km. The average distance between the airports in the network is approximately 1952 km. We propose a quantum network with airports as nodes and the connections between the airports as edges. We define the edge weight of the edge connecting the nodes $\{v_i,v_j\}$ of the network as   
\begin{align}
    \mathrm{w}(e_{ij}) \coloneqq \begin{cases}
    \mathrm{e}^{-L_{ij}/22} & \text{if $L_{ij} < 50$ km}\\
    0.8 & \text{if $L_{ij} \geq 50$ km}
    \end{cases}
\end{align}
where $L_{ij}$ is the distance between two airports denoted by nodes $v_i$ and $v_j$. The link sparsity of such a network is 0.99575, and the total connection strength is given by 0.99787. We observe that the most critical airports present in this network are Istanbul International Airport, Dubai International Airport, Anchorage Ted Stevens in Alaska, Beijing Capital International Airport, Chicago O'Hare International Airport, and Los Angeles International Airport. 

Let the airports of the network require to securely communicate with each other. The sharing of entangled states among the airports is a primitive for secure communication among them. Let us assume all these airports are located at the same altitude. Let the ground stations located at the airports share an entangled state using a global satellite-based mesh quantum network as described in Sec.~\ref{sec:entanglementDistributionCities}. The ground stations connect to the satellite network via the atmospheric channel modelled as a qubit thermal channel. The satellites of the network are interconnected via a qubit erasure channel. For two airports $a_1$ and $a_2$ requiring to connect, the average entanglement yield is given by
\begin{eqnarray} 
    \xi_{\text{avg}} &&= \eta_t^{a} (\eta_e^2)^{n-1} q^{n-1} \nonumber \\
    && = q^{n-1} \left(\eta_e^2\right)^{\left\lfloor \frac{L}{L_0}\right\rfloor -1} \bigg[\kappa_g(\kappa_g-1)(\eta_g-1)^2+\frac{1}{2}(1+\eta_g^2)\bigg], \nonumber \\
    \label{eq:yieldAirportNetwork} 
\end{eqnarray}
where $L$ denotes the distance between $a_1$ and $a_2$ and $q$ is the success probability of the standard Bell measurement at the satellite stations. $L_0$ denotes the distance between the nodes of the satellite network. We assume identical atmospheric conditions at $a_1$ and $a_2$ and set $\eta_e = 0.95, \eta_g = 0.5$, $\kappa_g = 0.5$ and $q = 1$. With these choices of parameters, we present in Fig.~\ref{fig:sampleGraph115} the average yield as a function of the distance between the nodes of the satellite network for different values of $L$. Furthermore, we plot in Fig.~\ref{fig:sampleGraph11125}, the variation in the average yield $\xi_{\text{avg}}$ as a function of the distance between the virtual nodes for different values of $q$. For this we set $\eta_e = 0.95$, $\eta_g = 0.5$, $\kappa_g = 0.5$ and $L = 4000$ km. The entanglement yield between the airports for different channel parameters not considered in this section can be obtained from Eq.~\eqref{eq:yieldAirportNetwork}.
\begin{figure}
    \centering
    \includegraphics[scale=0.85]{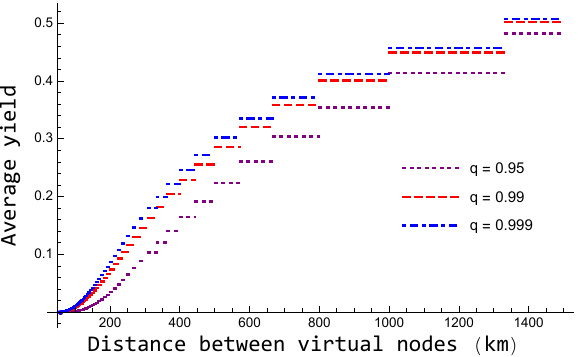}
    \caption{In this figure, we plot the average yield $\xi_{\text{avg}}$ (see Eq.~\eqref{eq:yieldAirportNetwork}) as a function of the distance between the virtual nodes $(L_0)$ for different values of $q$. \textcolor{black}{For this, we set the parameters $\eta_e = 0.95, \eta_g = 0.5, \kappa_g = 0.5$ and $L = 4000$ km.} (Color online)}
    \label{fig:sampleGraph11125}
\end{figure}

\begin{figure}
    \centering
    \includegraphics[scale=0.85]{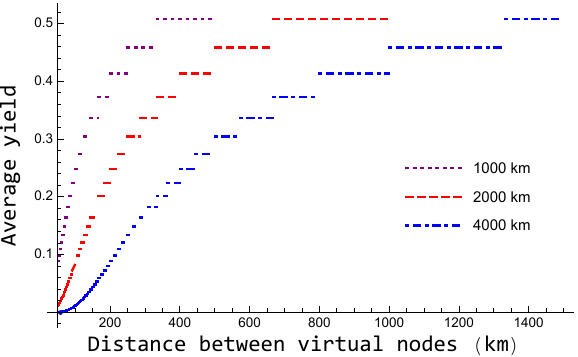}
    \caption{In this figure, we plot the average yield $\xi_{\text{avg}}$ (see Eq.~\eqref{eq:yieldAirportNetwork}) as a function of the distance between the virtual nodes for different lengths between the airports. \textcolor{black}{For this, we set the parameters $\eta_e = 0.95, \eta_g = 0.5, \kappa_g = 0.5$, $q = 1$.} (Color online)}
    \label{fig:sampleGraph115}
\end{figure}

\begin{figure}
    \centering
    \includegraphics[scale=0.6]{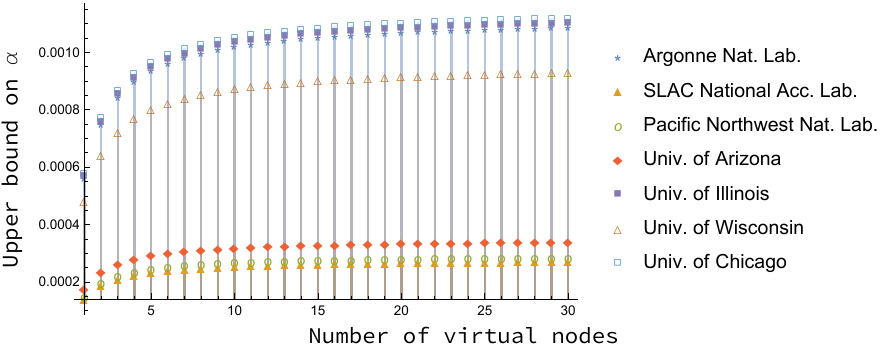}
    \caption{The upper bound on the channel parameter $\alpha$ for sharing entanglement between the U.S. Department of Energy (DoE) at Washington D.C. and some other labs~\cite{NQI} involved in The National Quantum Initiative (NQI). In this network, the DoE is the hub node while the labs are at the outer nodes. The edge connecting the hub to an outer node represents a repeater relay network. In this plot we set $q = 1$. (Color online)}
    \label{fig:upperBoundAlphaRepeaterRelay}
\end{figure}

\begin{figure}
    \centering
    \includegraphics[scale=0.7]{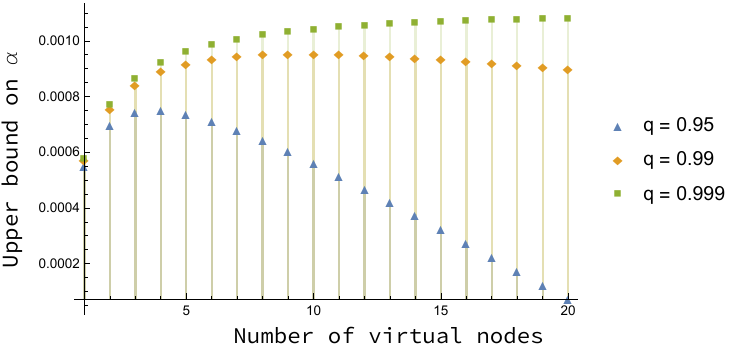}
    \caption{The upper bound on the channel parameter $\alpha$ for sharing entanglement between the U.S. Department of Energy (DoE) at Washington D.C. and the University of Chicago ($L = 952$ km) involved in The National Quantum Initiative (NQI) for different values of success probability of standard Bell measurement. In this network, the DoE is the hub node while the labs are at the outer nodes. The edge connecting the hub to an outer node represents a repeater relay network. (Color online)}
    \label{fig:upperBoundAlphaRepeaterRelayq}
\end{figure}

The quantum Internet can be used for secure communication
between a central agency and end parties. It may be
desirable here for the central agency to prevent direct
communication between the end parties. As an example, consider the U.S. Department of Energy (DoE) in Washington D.C. require to securely communicate by sharing entangled states with the major labs~\cite{NQI} that are involved in The National Quantum Initiative (NQI) using the Star network. The DoE is at the hub node of such a network and the different labs are the leaf nodes. For each edge of the network, let there be independent fiber-based repeater chain networks (described in Sec.~\ref{sec:limitationNetwork}.) Let each edge present in the network have a success probability $p = \mathrm{e}^{-\alpha L / n}$ where $\alpha$ is the channel loss parameter, $L$ is the total distance      between the DoE and the lab, and $n$ is the number of virtual nodes. Using the fact that isotropic state of visibility $q^n \lambda^{n+1}$ is entangled for $q^n \lambda^{n+1} > 1/3$, we obtain the upper bound on $\alpha$ as 
\begin{equation}
    \alpha < \frac{\log(3q^n)}{L (1 + \frac{1}{n})}.
\end{equation}
For different labs, we plot in Fig.~\ref{fig:upperBoundAlphaRepeaterRelay} the upper bound on $\alpha$ for different values of $n$. In the figure, we have considered some of the major labs, which can be extended to all other labs involved in the NQI. Furthermore, in Fig.~\ref{fig:upperBoundAlphaRepeaterRelayq}, we plot the upper bound on $\alpha$ for different values of $n$ and $q$ (see inset). For this figure we consider sharing entanglement between the U.S. Department of Energy (DoE) at Washington D.C. and the University of Chicago ($L = 952$ km).

\textcolor{black}{The practical implementation of the quantum Internet requires routing and allocation logic for the resources that are shared across the network. We next, present algorithmic primitives that are necessary to execute such tasks.}

\section{Algorithms} \label{sec:algorithms}
In Sec.~\ref{sec:shortestPath}, we provide an algorithm to find the shortest path between a pair of end nodes. We then provide an algorithm in Sec.~\ref{sec:networkCosntruction} to constrict a network architecture for sharing resources between two parties, each having multiple nodes. Then in Sec.~\ref{app:critNetworkNodes}, we provide an algorithm to obtain the critical parameter for the nodes of a given network. In Sec.~\ref{sec:resourceAllocationGroundStation}, we provide an algorithm to optimize the flow of resources at a node having multiple input and output channels. 

\subsection{Shortest path between a pair of nodes} \label{sec:shortestPath}
For performing $\mathrm{Task}_\ast$ with maximum success probability, it is desirable to transmit resource $\chi$ between any two nodes via the shortest network path connecting them. Recent works have considered different network topologies~\cite{SMI+16,CRD+19,LLL+20} and limitations on current and near-term hardware~\cite{skhatri22,inesta2023optimal} for routing resources over quantum networks. For a given network $\mathscr{N}$ represented as a graph $G$, we consider here the task of finding the path between two given nodes in $G$ that have the lowest effective weight for routing resources between them~\cite{SMI+16}. We call a path connecting two nodes in the network and having the lowest effective weight as the shortest path between them. Finding the shortest path between two nodes of a network is important as longer paths are more vulnerable to node and edge failures. To find the shortest path between two nodes, we use Dijkstra's algorithm \cite{dijkstra1959note,vMSL+13} suited to our network framework. In Algorithm~\ref{algo:shortestPath}, the shortest spanning tree is generated with the source node as the root node. Then the nodes in the tree are stored in one set and the other set stores the nodes that are not yet included in the tree. In every step of the algorithm, a node is obtained that is not included in the second set defined above and has a minimum distance from the source.
\begin{algorithm}[H]
\caption{Obtaining the shortest spanning tree with source node as root} \label{algo:shortestPath}
\begin{algorithmic}[1]

\Function{SpanTree}{$G,S,\text{target}$}
    \Initialize{\strut$\text{pq} \gets \text{empty min priority queue}$\\
    $\text{dist} \gets \emptyset$ \\
    $\text{pred} \gets \emptyset$}
    \For{\text{every node in G}}
        \If{\text{node }$= S$}
            \State pq[node] $\leftarrow$ 0
        \Else
            \State pq[node] $\leftarrow$ infinite
        \EndIf
    \EndFor
    
    \For{\text{every node and minDist in pq}}
        \State dist[node] $\gets$ minDist
        \If{node = target}
            \State break
        \EndIf
        \For{every neigh of node}
            \If{neigh $\in$ pq}
                \State score $\gets$ dist[node] + $G[\text{node, neigh}][\text{weight}]$
                \If{score $<$ pq[neigh]}
                    \State pq[neigh] $\gets$ score
                    \State pred[neigh] $\gets$ node
                \EndIf
            \EndIf
        \EndFor
    \EndFor
    
    \Return dist, pred
\EndFunction  \label{dijkstra}
\end{algorithmic}
\end{algorithm}
To obtain the shortest path between any two nodes of a given graph, we apply Algorithm~\ref{algo:shortestPathExample}. Algorithm~\ref{algo:shortestPathExample} returns a path only \textit{iff} the weight associated with the network path between the source and target nodes is at most equal to the critical weight $\mathrm{w}_{\text{crit}} (= -\log p_\ast)$, $p_\ast$ being the critical success probability for $Task_\ast$.
\begin{algorithm}[H]
\caption{Obtaining the shortest network path between the source and target nodes in a graph for end nodes to perform $\mathrm{Task}_\ast$ by sharing $\chi$.} \label{algo:shortestPathExample}
\begin{algorithmic}[1]
\Initialize{\strut$source \gets \text{starting node}$\\
    $\text{target} \gets \text{target node}$ \\
    $\text{flag} \gets 0$ \\
    $p_\ast \gets \text{success probability for } Task_\ast$}

\State G $\gets$ the given graph 

\State [dist, pred] $\gets$ DIJKSTRA(G, source, target)
\State end $\gets$ target
\State path $\gets$ [end]
\While{(end $\neq$ source) AND (flag $\leq -\log p_\ast$)}
    \State end = pred[end]
    \State now = path[end]
    \State path.append(end)
    \State next = path[end]
    \State flag $\gets$ flag + weight(G.edge(next,now))
\EndWhile
\If{path[end] = source}
    \State disp(path)
\Else
    \State disp(disconnected nodes)
\EndIf \label{algo:shortestPath1}
\end{algorithmic}
\end{algorithm}
\begin{example}
Let us consider a weighted graph with $8$ nodes as shown in Fig. \ref{fig:multiplePaths}. 
\begin{figure}
    \includegraphics[scale=0.6]{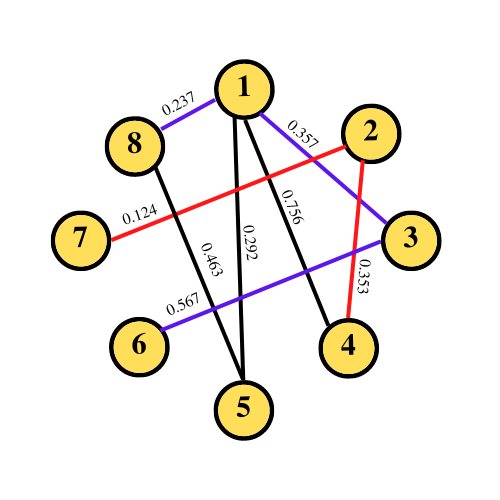}
    \caption{A network represented by a weighted graph with $8$ nodes. Multiple pairs of nodes can share resources using this network. As an example, nodes $(8,6)$ can share resource via the path $8 \leftrightarrow 1 \leftrightarrow  3 \leftrightarrow 6$ (shown in blue), then nodes $(7,4)$ can share resource via the path $7 \leftrightarrow 2 \leftrightarrow 4$ (shown in red). (Color online)}
    \label{fig:multiplePaths}
\end{figure}
A physical interpretation can be to consider the transfer of quantum states from node $v_i$ to node $v_j$ via quantum channels denoted by the edges. The edge weight between the nodes $v_i$ and $v_j$ is given by $-\log p_{ij}$ where $p_{ij}$ is the success probability of sharing the resource between these two nodes. The shortest path connecting the nodes would then provide the highest success probability for the task. If we consider the source as node $v_8$ and the target as node $v_6$, then the algorithm returns the shortest path as $v_8 \leftrightarrow v_1 \leftrightarrow  v_3 \leftrightarrow v_6.$ It may be desirable for another node say $v_7$ to share a resource with node $v_4$ using the same network. The node $v_7$ and $v_4$ can share resources via the path $v_7 \leftrightarrow v_2 \leftrightarrow v_4$ without involving the virtual nodes in the shortest path between $(v_8,v_6)$. We observe that multiple pairs of nodes can share resources using this network.
\end{example}
In the following subsection, we present an algorithm to construct a network for sharing resources between two parties each having multiple nodes. 

\subsection{Network construction} \label{sec:networkCosntruction}
Let two parties Alice (denoted by $A$) and Bob (denoted by $B$) require to share a resource using a mesh network. We assume that $A$ and $B$ have $n_A$ and $n_B$ number of nodes respectively. We introduce Algorithm~\ref{network} to obtain the structure of the mesh that ensures there exist distinct paths between nodes of $A$ and $B$. In Algorithm~\ref{network}, we impose the constraints that (a) at a time all nodes of $A$ shall be connected to distinct nodes of $B$ via the shortest available path with unique virtual nodes and (b) there exists a path between every nodes of $A$ and $B$.
\begin{algorithm}[H]
\caption{Network construction algorithm} \label{network}
\begin{algorithmic}[1]
\Function{network}{$n_A,n_B$}
    \Initialize{count $\gets n_A + n_B$ \\
    g $\gets$ complete graph (no. of nodes: count) \\
    g[weight] $\gets$ mesh edge weights \\
    wt $\gets$ local edge weights
    }
    \For{every node $v_i$ of $A$}
        \State add g.node($v_i$)
        \State add g.edge($v_i$, g[node=count], weight = wt[count])
        \State count $\gets$ count - 1
    \EndFor
    \For{every node $v_j$ of $B$}
        \State add g.node($v_j$)
        \State add g.edge($v_i$, g[node=count], weight = wt[count]) 
        \State count $\gets$ count - 1
    \EndFor
    \Return g
\EndFunction 
\end{algorithmic}
\end{algorithm}
\begin{example}
Consider two geographically separated companies $A$ and $B$ requiring to connect to each other via a mesh network. We call the headquarters of the companies as hubs. In Fig.~\ref{fig:networkDesign} we show the hubs of $A$ and $B$ in blue and orange respectively. 
\begin{figure}
    \centering
    \includegraphics[scale=0.5]{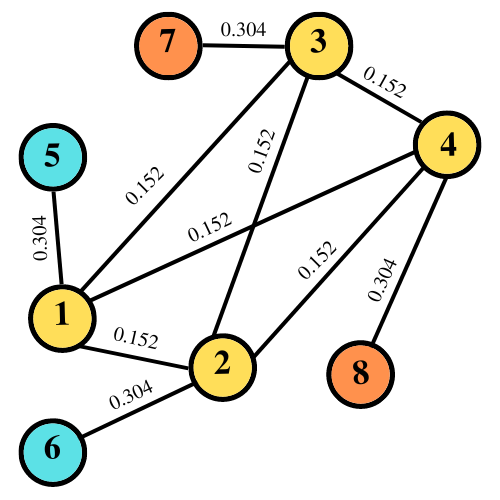}
    \caption{A $4$ node network (shown in yellow) constructed using Algorithm~\ref{network} for connecting the hubs $A$ and $B$. The hubs $A$ and $B$ each have two nodes and are shown in orange and blue respectively. (Color online)}
    \label{fig:networkDesign}
\end{figure}
Using Algorithm~\ref{network} we obtain the network topology for which there exists distinct paths for possible pairs of $(a_i,b_j)$ where $a_i \in A$ and $b_j \in B$.  
\end{example}

\subsection{Critical nodes in a network} \label{app:critNetworkNodes}
The critical nodes of the network are essential for the proper functioning of the network. If one of these nodes is removed, it will lead to a decrease in the overall performance efficiency of the network. We present heuristic Algorithm~\ref{algo:criticalNodes} to obtain the critical parameter $\nu_i$ for the network node $v_i \in \mathbb{V}$ of the network $\mathscr{N}(G(\mathbb{V},\mathbb{E}))$ using Eq. (\ref{eq:Pcrit}). 
\begin{algorithm}[H]
\caption{Finding the critical parameter $\nu_i$ for node $v_i \in \mathbb{V}$ of the network $\mathscr{N}(G(\mathbb{V},\mathbb{E}))$} \label{algo:criticalNodes}
\begin{algorithmic}[1]

\State G $\gets$ the given graph  

\For{every node in G}
    \State $C_i \gets $ clustering coeff using Eq. (\ref{eq:clusterCoeff})
    \State $\Bar{\mathrm{w}}_{\ast}(G) \gets $ avg cost using Eq. (\ref{eq:globalEff})
    \State $\tau_i \gets $ centrality of the node
    \State $\nu_i \gets \tau_i/C_i~\Bar{\mathrm{w}}_{\ast}(G)$
    \State critPar[node] $\gets \nu_i$
\EndFor
\end{algorithmic}
\end{algorithm}

In the following subsection, we present an algorithm for optimizing resource flow at a node with multiple input and output channels.

\subsection{Resource allocation at a node} \label{sec:resourceAllocationGroundStation}

The ground stations in the satellite-based network presented in Sec.~\ref{sec:entanglementDistributionCities} share an entanglement buffer~\cite{LGN23} to store the incoming quantum states from the satellite network. The stored states are later distributed via different output channels to neighbouring nodes based on traffic requests. We present Algorithm~\ref{algo:resourceAllocation} for the optimal flow of states at a node with multiple input (producer(thread)) and output channels (consumer(thread)).
\begin{algorithm}[H] 
\caption{Resource allocation at a node} \label{algo:resourceAllocation}
\begin{algorithmic}
\Initialize{\strut$\text{buffSize} \gets \text{size of quantum memory}$\\
    $\text{buffer} \gets \emptyset$}
\Procedure{producer}{thread}
    \While{state incoming AND empty memory slot}
        \State gain access to memory
        \State insert state into empty memory slot
        \State update other memory slots as per task
        \State release memory access
    \EndWhile
\EndProcedure
\Procedure{consumer}{thread}
    \While{memory is not empty}
        \State gain access to memory
        \State acquire state from the memory
        \State update other memory slots as per task
        \State release memory access
    \EndWhile
\EndProcedure
\State create and start all producer threads
\State create and start all consumer threads
\end{algorithmic}
\end{algorithm}
Algorithm~\ref{algo:resourceAllocation} is the standard producer-consumer model in networking where the procedures CONSUMER threads\footnote{Thread is a sequential execution of tasks in a process.} are the instances of the output channels that extract quantum states from the buffer, while the PRODUCER are the instances for the input channels that inputs quantum states to the buffer.

\section{Discussion} \label{sec:discussion}
We envision that the implementation of the quantum Internet will follow a task-oriented approach. The underlying network structure at any stage of implementation is expected to provide loose coupling, meaning end users can perform information processing tasks without requiring to know the details of implementation, thereby reducing dependencies between different tasks. This requires assessing the practical limitations for implementing different tasks. 

Our work provides a step in that direction by taking a graph theory (and information theory-)based approach to analyze the scalability and robustness of the quantum Internet. Focusing on the latter part, we have provided measures for comparing the robustness and identifying the critical nodes of different network topologies. Identifying quantum processors as real-world mesh networks, we compared the robustness measures for the quantum processor architectures by Google, IBM, and Rigetti. With the vision of having a 1024-qubit quantum processor in the future, we extend the 54-qubit layout by Google to include 1024 qubits and observe the robustness of such a network. 

The network structure of the quantum Internet is determined by the information processing tasks that are implemented using it. Looking at the elementary link level, we have obtained bounds on the critical success probability for performing different tasks. Extending to a more general repeater-based network, we have obtained a trade-off between the channel length and the time interval for which the states can be stored at the nodes such that the shared state is useful for different tasks. 

Considering performing some desirable information processing tasks over lattice networks, we present a theorem specifying conditions that lead to the absence of percolation. As implications of the theorem, we highlight the constraints on network scalability and limitations of current technology for performing quantum communication and implementing DI-QKD protocols. \textcolor{black}{Our analysis also provides a theoretical baseline as to how hypothetical repeater improvements can potentially mitigate these bottlenecks.}

Looking at the specific details of implementation, considering repeater-based networks, we have provided the range of isotropic state visibility and an upper bound on the number of repeater nodes for distilling secret keys at non-zero rates via DI-QKD protocols. We have considered practical parameters like atmospheric conditions and imperfect devices in obtaining bottlenecks for implementing a satellite-based model distributing resources between far-off places. For such a network we have presented algorithms for implementing certain underlying network-related tasks such as obtaining the network layout, obtaining the network routing path and allocating resources at the network nodes. Overall, the assessment presented in this paper can be used in benchmarking the critical parameters involved in realizing the quantum Internet. 

\begin{acknowledgements}
AS thanks Keval Jain for the useful discussions. The authors thank Aryaman Manish Kolhe for comments on the manuscript. MAS acknowledges the Summer Research InternSHip on Technological Innovations (SRISHTI) programme held during the summer of 2022 at IIIT Hyderabad; this program was supported by iHub-Data, a Technology Innovation Hub (TIH) established by IIIT Hyderabad as part of the National Mission on Interdisciplinary Cyber-Physical Systems (NM-ICPS) scheme of Department of Science and Technology, Government of India. KH acknowledges partial support by the Foundation for Polish Science (IRAP project, ICTQT, contract no.\ MAB/2018/5, co-financed by EU within Smart Growth Operational Programme). The `International Centre for Theory of Quantum Technologies' project (contract no.\ MAB/2018/5) is carried out within the International Research Agendas Programme of the Foundation for Polish Science co-financed by the European Union from the funds of the Smart Growth Operational Programme, axis IV: Increasing the research potential (Measure 4.3). KH and SD acknowledge partial support from the National Science Centre, Poland, grant Opus 25, 2023/49/B/ST2/02468. SD acknowledges support from the Science and Engineering Research Board, Department of Science and Technology (SERB-DST), Government of India under Grant No. SRG/2023/000217 and partial support from the Ministry of Electronics and Information Technology (MeitY), Government of India, under Grant No. 4(3)/2024-ITEA. AS and SD thank IIIT Hyderabad for the Faculty Seed Grant.  
\end{acknowledgements}

\appendix

\section{Dual rail encoding of photons} \label{app:dualRail}
In the dual rail encoding scheme \cite{knill2001scheme}, the qubits are encoded using the optical modes\footnote{An optical mode of the photon is defined by the state space consisting of a superposition of number states.} of photons. The computational basis of the photonic qubit system $A$ encoded in the polarization modes $l_1$ and $l_2$ is given by
\begin{equation}
    \ket{H}_A \rightarrow \ket{1,0}_{l_1,l_2},~\ket{V}_A \rightarrow \ket{0,1}_{l_1,l_2}.
\end{equation}
A pure state $\ket{\psi}_{A}$ of the qubit system can be expressed in such a computational basis as
\begin{eqnarray}
\ket{\psi}_A && = \alpha \ket{1,0}_{l_1,l_2} + \beta \ket{0,1}_{l_1,l_2} \\
&& = \alpha \ket{H}_A + \beta \ket{V}_A,
\end{eqnarray}
where $\alpha, \beta \in \mathbbm{C}$ and $|\alpha|^2 +|\beta|^2 = 1.$ The bipartite entangled state $\Psi_{AB}^+$ is expressed as $\Psi_{AB}^+ = \op{\Psi^+}_{AB}$ where $\ket{\Psi^+}_{AB} = \frac{1}{\sqrt{2}}(\ket{HH}_{AB} + \ket{VV}_{AB})$.

\section{Two-qubit states} \label{app:twoQubitStates}
The maximally entangled states for the two-qubit systems can be expressed as
\begin{eqnarray}
    &&\Phi^\pm_{AB} \coloneqq  \op{\Phi^\pm}_{AB} \nonumber \\
    &&\Psi^\pm_{AB} \coloneqq \op{\Psi^\pm}_{AB} 
\end{eqnarray}
where $\ket{\Psi}^\pm_{AB}$ and $\ket{\Phi}^\pm_{AB}$ are defined as
\begin{eqnarray}
    &&\ket{\Psi}^\pm_{AB} := \frac{1}{\sqrt{2}}\bigg(\ket{00}_{AB} \pm \ket{11}_{AB}\bigg), \\
    &&\ket{\Phi}^\pm_{AB} := \frac{1}{\sqrt{2}}\bigg(\ket{01}_{AB} \pm \ket{10}_{AB}\bigg).
\end{eqnarray}
Following from Eq.~\eqref{eq:iso}, a two-qubit isotropic state can be expressed as
\begin{equation}
    \rho_{AB}^I (p,2) = \frac{1}{3}(4p - 1) \Psi^+_{AB} + \frac{4}{3} (1 - p) \frac{\mathbbm{1}_{AB}}{4}.
\end{equation}
Similarly, following from Eq.~\eqref{def:wernerState}, a two-qubit Werner state can be expressed as
\begin{equation}
    \rho_{AB}^W (p(\lambda),2) = \lambda \Psi^-_{AB} + (1 - \lambda) \frac{\mathbbm{1}_{AB}}{4},
\end{equation}
where $\lambda \in [-1/3,1]$.

\section{Two qubit Bell measurement on isotropic states} \label{app:repeaterNode}
Let there be a measurement station that performs standard Bell measurement on the halves of two-qubit isotropic states $\rho_{A_1 A_1'}^I (p'(\lambda),2)$ and $\rho_{B_1 B_1'}^I (p'(\lambda),2)$, where
\begin{equation} \label{eq:werner}
    \rho_{AB}^I (p'(\lambda),2) := \lambda \Psi^+_{AB} + (1 - \lambda) \frac{\mathbbm{1}_{AB}}{4}.
\end{equation}
We call $\lambda$ the visibility of the isotropic state $\rho_{AB}^I (p'(\lambda),2)$. Let the success probability of a standard Bell measurement be $q$. We denote the action of the noisy standard Bell measurement channel as
\begin{eqnarray}
&&\mathcal{E}_{A_1'B_1' \to I_AI_A'I_BI_B'} \bigg( \rho_{A_1A_1'}^I (p'(\lambda),2) \otimes \rho_{B_1B_1'}^I (p'(\lambda),2)\bigg) \nonumber \\
= ~&& \frac{\lambda^2 q}{4}\bigg[\Psi_{A_1B_1}^- \otimes \op{00}_{I_AI_B} \nonumber \\
&& + \Psi_{A_1B_1}^+ \otimes \op{11}_{I_AI_B} +\Phi_{A_1B_1}^- \otimes \op{22}_{I_AI_B} \nonumber \\
&& +\Phi_{A_1B_1}^+ \otimes \op{33}_{I_AI_B}\bigg] \otimes \op{00}_{I_A'I_B'} \nonumber \\
&& + (1 - \lambda^2)q\frac{\mathbbm{1}_{A_1B_1}}{4} \otimes \frac{1}{4} \sum_{i=0}^{3} \op{ii}_{I_AI_B} \otimes \op{00}_{I_A'I_B'} \nonumber \\
&& + (1 - q)\frac{\mathbbm{1}_{A_1B_1}}{4} \otimes \op{\perp}_{I_AI_B} \otimes \op{11}_{I_A'I_B'}. \label{eq:WernerBellMeasurement}
\end{eqnarray}
The flag state $\op{11}_{I_A'I_B'}$ indicates error in the standard Bell measurement with a probability $(1-q)$ and the state $\frac{\mathbbm{1}_{A_1B_1}}{4}$ is left on $\mathscr{H}_{A_1B_1}$. The flag state $\op{00}_{I_A'I_B'}$ indicates a successful standard Bell measurement with probability $q$. If error corrections are possible post-Bell measurement, then from a single use of repeater we have the state
\begin{equation}
    \rho_{AB}^{I}(p(q\lambda^2),2) = q\lambda^2~\Psi^+_{AB} + (1 - q\lambda^2) \frac{\mathbbm{1}_{AB}}{4}. 
\end{equation}

\section{Sparsity index}
Taking motivation from the Gini index of network graphs~\cite{gini1912variabilita} and using the definition of the connection strength of the nodes in a network, we define the sparsity index of the network.
\begin{definition}
Consider the plot with the cumulative sum of the number of nodes in the network $\mathscr{N}(G(\mathbb{V},\mathbb{E}))$ along the horizontal axis and the cumulative sum of $\zeta_\mathscr{N}(v_i)$s along the vertical axis. The sparsity index of the network $\mathscr{N}$ is given by
\begin{equation} \label{eq:sparsityIndex}
    \Xi(\mathscr{N}) = \frac{\text{area enclosed by the curve and x-axis}}{\text{area enclosed by the $45^{\circ}$ line}}. 
\end{equation}
\end{definition}
As an example, the star network $\mathscr{N}_s(G_s(\mathbb{V},\mathbb{E}))$ (see Eq.~\eqref{eq:starNet}) with $N_v = 8$ and $p = 0.5$ have sparsity index $0.1934$ for non-cooperative strategy and $0.3779$ for cooperative strategy. The sparsity index $\Xi(\mathscr{N})$ of the network $\mathscr{N}$ measures the extent of inequality in the distribution of connection strength among the nodes of the network. High values of $\Xi$ indicate a high cumulative percentage of connection strength for the cumulative fractile of the nodes. Typically it is desirable for the network to have high values of sparsity index.

\section{Actions of some quantum channels} \label{app:channel}
\subsection{The qubit depolarizing channel} \label{app:DepolChannel}
The action of a qubit depolarizing channel \cite{bennett1996mixed} on the qubit density operator $\rho_A$ is given by 
\begin{equation} \label{depolChannelDef}
    \mathcal{D}_{A \rightarrow B} (\rho_A) = (1 - p) \hspace{5pt} \rho_B + p \hspace{5pt} \frac{\mathbbm{1}_B}{2},
\end{equation}
where $p\in[0,\frac{4}{3}]$ is the channel parameter and $\mathbbm{1}_B$ is the identity operator. In the formalism given by Kraus \cite{kraus1971general} and Choi \cite{choi1975completely}, the effect of the channel can be defined by the following operators,
\begin{eqnarray}
&&\mathcal{K}_0 = \sqrt{1 - \frac{3 p}{4}} \hspace{5pt} \mathbbm{1} \\
\text{and } && \mathcal{K}_i = \frac{\sqrt{p}}{2} \hspace{5pt} \mathbbm{\sigma}_i \text{ with } i \in \{1,2,3\},
\end{eqnarray}
where $\sigma_i$ are the Pauli matrices. The action of the depolarizing channel on each of the systems $A$ and $\Bar{A}$ is given by,
\begin{eqnarray}
&& \mathcal{D}_{A\rightarrow B} \hspace{2pt} \otimes \hspace{2pt} \mathcal{D}_{\Bar{A}\rightarrow \Bar{B}} \left( \Psi_{A\Bar{A}}^+ \right) \nonumber\\
&& = \sum_{i = 0}^{4} \sum_{j = 0}^{4} \left(\mathcal{K}_i \otimes \mathcal{K}_j\right)  \Psi_{A\Bar{A}}^+  \left(\mathcal{K}_i \otimes \mathcal{K}_j\right)^\dagger \\
&& = (-1 + p)^2~\Psi_{B\Bar{B}}^+ + \hspace{5pt} p \hspace{5pt} (2 - p) \hspace{5pt} \frac{\mathbbm{1}_{B\Bar{B}}}{4}  \label{depolState}
\end{eqnarray}

Noting that $\mathbbm{1}_{B\Bar{B}}$ can be expressed as the sum of four maximally entangled states, we have the fidelity of the final state to the starting state as $\eta_d = 1 - \frac{3}{4} \hspace{5pt} p \hspace{5pt} (2-p).$ Following the approach of \cite{inesta2023optimal} and applying the depolarising channel $n$ times, the final state after the evolution through the channel is  
\begin{equation}
    \rho_{B\Bar{B}}^n = \left(\mathcal{D}_{A\rightarrow B} \hspace{2pt} \otimes \hspace{2pt} \mathcal{D}_{A\rightarrow B}\right)^{\otimes n} \left( \Psi_{A\Bar{A}}^+ \right). \nonumber
\end{equation}
This state $\rho_{B\Bar{B}}^n$ has a fidelity with $\Psi_{B\Bar{B}}^+$ given by
\begin{equation} \label{app:depolYield}
    \eta_d^n = (1-p)^{2 n}-\frac{1}{4} (p-2) p \left((n-1) (1-p)^{2 (n-1)}+1\right),
\end{equation}
where it can be seen that $\eta_d^1 = \eta_d$. We plot in Fig.~\ref{Fig:delopFidelity}, the variation of $\eta_d^n$ (using Eq.~\eqref{app:depolYield}) as a function of the number of applications $n$ of the depolarising channel and the depolarising channel parameter $p$.
\begin{figure}
    \centering
    \includegraphics[scale=0.5]{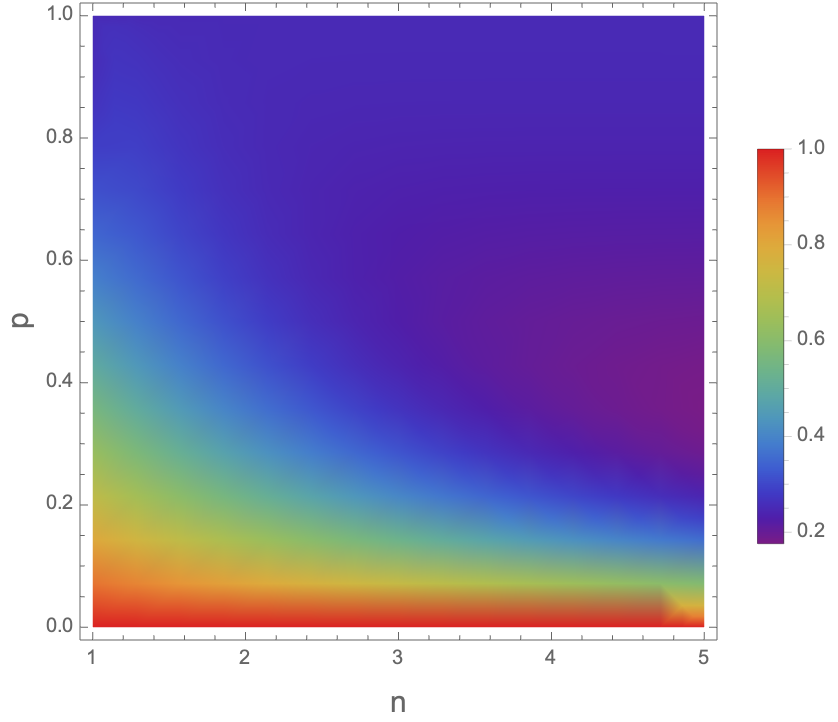}
    \caption{In this figure, we plot the variation of $\eta_d^n$ as a function of $n$ and $p$. (Color online)}
    \label{Fig:delopFidelity}
\end{figure}
\subsection{The qubit erasure channel} \label{app:ErasureChannel}
The action of a qubit erasure channel~\cite{grassl1997codes} on an input density operator $\rho_A$ is given by
\begin{equation}
    \mathcal{E}_{A\rightarrow B} (\rho_A) = \eta_e \rho_B + (1 - \eta_e) \Tr[\rho_B] \op{e}_B.
\end{equation}
The action of the erasure channel on the qubit system is such that it outputs the exact input state with probability $\eta$ or with probability $1-\eta$ replaces it with erasure state $\ket{e}$, where $\ket{e}$ is the vacuum state. The action of the channel on the maximally entangled state $\Psi_{A\Bar{A}}^+$ in the two-qubit space $A\Bar{A}$ is given by 
\begin{eqnarray}
\mathcal{E}_{A\rightarrow B} \hspace{5pt} \otimes \hspace{5pt} &&\mathcal{E}_{\Bar{A}\rightarrow \Bar{B}} \hspace{5pt} (\Psi_{A\Bar{A}}^+) \nonumber \\
&&= \eta_e^2 \Psi_{B\Bar{B}}^+ + (1 - \eta_e^2) \Psi_{B\Bar{B}}^\perp, \label{erasureState}
\end{eqnarray}
where 
\begin{eqnarray}
    &&\Psi_{B\Bar{B}}^\perp \coloneqq \frac{\eta_e}{1+\eta_e} \bigg(\frac{1}{2}\mathbbm{1}_B \otimes \op{e}_{\Bar{B}} + \op{e}_B \otimes \frac{1}{2} \mathbbm{1}_{\Bar{B}}\bigg) \nonumber \\
    && \qquad \qquad \qquad \qquad + \frac{1 - \eta_e}{1 + \eta_e} \op{e}_B \otimes \op{e}_{\Bar{B}}
\end{eqnarray}
is a state orthogonal to $\Psi_{B\Bar{B}}^+$.

Consider two sources creating pairs of entangled state $\Psi^+$. The first and second source distributes the entangled pairs to node pairs $(v_1,v_2)$ and $(v_2,v_3)$ respectively via erasure channels. The node $v_2$ then performs a standard Bell measurement on its share of states. Assuming error correction is possible post-measurement, the node pair $(v_1,v_3)$ share the state $\Psi^+$ with probability $\eta_e^2.$

\subsection{The qubit thermal channel} \label{app:ThermalChannel}
The action of a qubit thermal channel on the density operator $\rho_A$ is given by
\begin{equation}
    \mathcal{L}^{\eta_g,\kappa_g}_{A\rightarrow B} (\rho_A) = \Tr_E [ U_{\eta_g} \hspace{5pt} (\rho_A \otimes \rho_E) \hspace{5pt} U_{\eta_g} ^\dagger ],
\end{equation}
where $\rho_E$ is the density operator of the environment given by
\begin{equation}
    \rho_E = (1 - \kappa_g) \ket{0}\bra{0}_E + \kappa_g \ket{1} \bra{1}_E.
\end{equation}
The qubit thermal channel is modelled by the interaction of a qubit system $\rho_A$ with the environment $\rho_E$ at a lossy beamsplitter having transmittance $\eta_g$ (see Eq.~\eqref{eq:transmittanceAtmosphere}). The evolution through the beamsplitter is via the unitary $U_{\eta_g}$ expressed as
\begin{equation}
    U_{\eta_g} =  \left(
\begin{array}{cccc}
 1 & 0 & 0 & 0 \\
 0 & \sqrt{\eta_g} & \sqrt{1 - \eta_g} & 0 \\
 0 & -\sqrt{1 - \eta_g} & \sqrt{\eta_g} & 0 \\
 0 & 0 & 0 & 1 \\
\end{array}
\right).
\end{equation}
The Generalised Amplitude Damping Channel (GADC) is equivalent to the qubit thermal channel up to the reparameterization $\mathcal{A}^{\eta_g,\kappa_g}_{A\rightarrow B} (\rho_A) \equiv \mathcal{L}^{1-\eta_g,\kappa_g}_{A\rightarrow B}(\rho_A)$ with $\eta_g \in [0,1]$ and $\kappa_g \in [0,1].$ The effect of the thermal channel can be defined by the following Kraus operators in the standard basis
\begin{eqnarray}
&&\widetilde{\mathcal{A}_1} = \sqrt{1-\kappa_g} (\ket{0}\bra{0} + \sqrt{\eta_g} \ket{1}\bra{1}) \\
&&\widetilde{\mathcal{A}_2} = \sqrt{(1-\eta_g)(1-\kappa_g)} \ket{0}\bra{1} \\
&&\widetilde{\mathcal{A}_3} = \sqrt{\kappa_g} (\sqrt{\eta_g} \ket{0}\bra{0} + \ket{1}\bra{1}) \\
&&\widetilde{\mathcal{A}_4} = \sqrt{\kappa_g (1-\eta_g)} \ket{1}\bra{0}.
\end{eqnarray}
The action of the thermal channel on each of the systems $A$ and $\Bar{A}$ is given by,
\begin{eqnarray}
\mathcal{\tau}_{B\Bar{B}}^{\eta_g,\kappa_g} \hspace{5pt} &&:= \mathcal{L}^{\eta_g,\kappa_g}_{A\rightarrow B} \hspace{2pt} \otimes \hspace{2pt} \mathcal{L}^{\eta_g,\kappa_g}_{\Bar{A}\rightarrow \Bar{B}} \left( \Psi^+_{A\Bar{A}} \right) \nonumber\\
&& = \sum_{i = 0}^{4} \sum_{j = 0}^{4} \left(\widetilde{\mathcal{A}_i} \otimes \widetilde{\mathcal{A}_j}\right) \left( \Psi^+_{A\Bar{A}} \right) \left(\widetilde{\mathcal{A}_i} \otimes \widetilde{\mathcal{A}_j}\right)^\dagger \nonumber \\
&&= \eta_g~\Psi^+_{B\Bar{B}} \nonumber \\
&& \hspace{10pt} + \hspace{5pt} (1-\eta_g)(\kappa_g-1)\bigg(\kappa_g(1-\eta_g)-1\bigg) \ket{00} \bra{00}_{B\Bar{B}} \nonumber \\
&& \hspace{10pt} + \hspace{5pt} \kappa_g (-1+\eta_g)\bigg(\kappa_g(-1+\eta_g)-\eta_g\bigg) \ket{11} \bra{11}_{B\Bar{B}} \nonumber \\
&& \hspace{10pt} + \hspace{5pt} \kappa_g (1 - \kappa_g) (1-\eta_g) ^2 \left( \ket{01} \bra{01}_{B\Bar{B}} + \ket{10} \bra{10}_{B\Bar{B}} \right). \nonumber \\ \label{GADCstate}
\end{eqnarray}
The state $\mathcal{\tau}_{B\Bar{B}}^{\eta_g,\kappa_g}$ has a fidelity with $\Psi^+_{B\Bar{B}}$ given by
\begin{eqnarray}
    \eta_t = \frac{1}{2}(1+\eta_g^2) + \kappa_g(\kappa_g-1)(1-\eta_g)^2. \label{etat}
\end{eqnarray}
\begin{figure}
    \centering
    \includegraphics[scale=0.5]{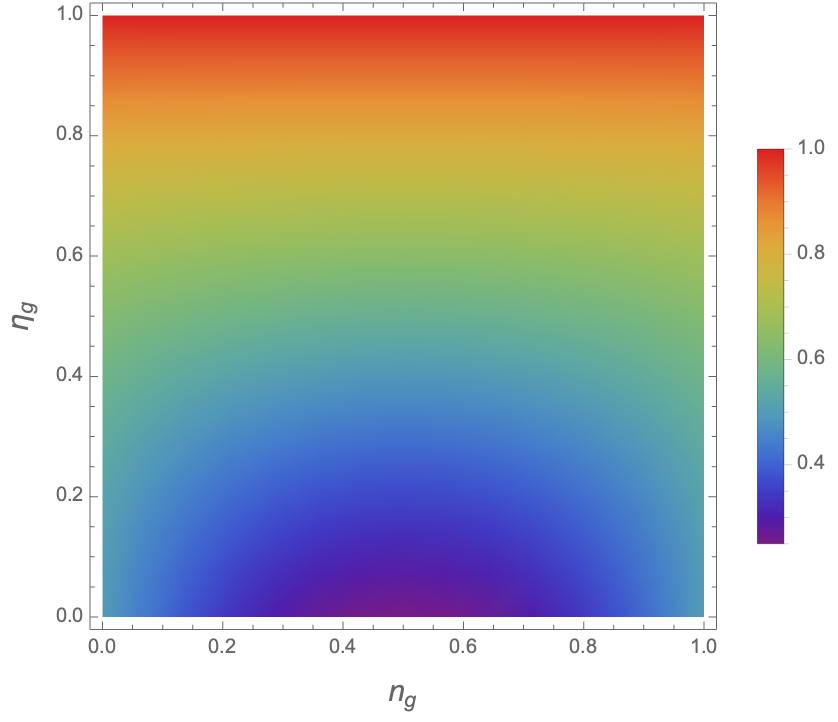}
    \caption{In this figure, we plot the variation of $\eta_t$ as a function of $\kappa_g$ and $\eta_g$. (Color online)}
    \label{Fig:GADCFidelity}
\end{figure}
We plot in Fig~\ref{Fig:GADCFidelity}, the variation of $\eta_t$ as a function of the channel parameters $\eta_g$ and $\kappa_g$. 

\section{The atmospheric channel} \label{app:atmosphericChannel}
The losses in the transmission of optical signals via an optical fiber are greater than that for free space transmission. In the vacuum space above the earth's atmosphere, the losses are nearly negligible. The non-birefringent nature of the atmosphere causes negligible change to the polarization state of the photons passing through it. These observations motivate the use of space and satellite technologies in establishing an entanglement distribution network using such channels~\cite{lu2022micius}. We observe in Fig.~\ref{fig:transmissionFiberSpace}, that the losses in the satellite-based free-space channel is much less compared to fiber-based channels for distances greater than $~ 70 \text{ km.}$  
\newline 
\begin{figure}[H]
    \centering
    \includegraphics[scale=0.4]{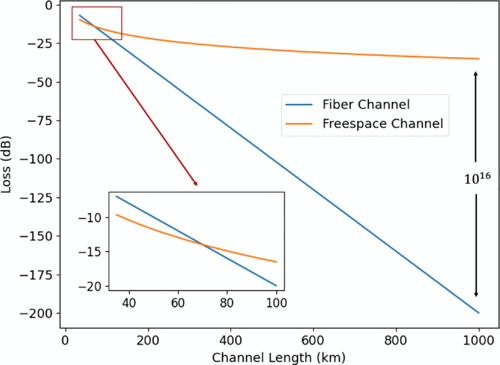}
    \caption{The comparison of losses in fiber and free-space channels as a function of channel length as discussed and plotted in \cite{lu2022micius}. It is observed that the free-space channel is advantageous for distances over 70 km. (Color online)}
    \label{fig:transmissionFiberSpace}
\end{figure}
The network architecture should be a hybrid satellite-optical fiber based model with ground-based global nodes at different geographical locations that are connected to different local nodes at small distances via optical fibers. These global nodes will be connected to the inter-satellite network. This architecture can in principle be extended to deep space allowing the possibility of sharing entanglement between nodes on Earth and the moon.

The factors affecting the transmission of optical signals between a satellite and a ground station are analyzed next. We obtain the efficiency in transmission $\xi_{\text{eff}}$ considering losses due to (a) inefficiencies in transmitting $(\xi_t)$ and receiving systems $(\xi_r)$, (b) beam diffraction $(\xi_d)$, (c) air turbulence $(\xi_{at})$, (d) mispointing $(\xi_p)$ and (e) atmospheric absorption $(\xi_{as})$. The total transmittance is given by 
\begin{equation}
    \xi_{\text{eff}} = \xi_t \hspace{3pt} \xi_r \hspace{3pt} \xi_d \hspace{3pt} \xi_{at} \hspace{3pt} \xi_p \hspace{3pt} \xi_{as} \hspace{3pt}. \label{app:xiEff}
\end{equation}
The diffraction of an optical beam depends on the beam's spatial mode, wavelength and aperture of the telescope. Assuming a Gaussian beam from the source with a waist radius of $\omega_0$, the radius at a distance $z$ is given by $\omega_d(z) = \omega_0 \sqrt{1 + (z/z_R)^2}$ with $z_R$ being the Rayleigh range. If the aperture radius of the telescope is $r$, then the receiving efficiency is given by \cite{lu2022micius}
\begin{equation}
    \xi_d = 1 - \exp{-\frac{2 r^2}{\omega_d^2}}. \label{app:xiD}
\end{equation}
The turbulence in the atmosphere induces inhomogeneity in the refractive index which changes the direction of the propagating beam. It was shown in  \cite{vasylyev2016atmospheric} that large-scale turbulence causes beam deflection while small-scale turbulence induces beam broadening. At the receiver end, it was shown in \cite{dios2004scintillation} that the average long-term accumulation of the moving spots shows a Gaussian distribution with an equivalent spot radius of $\omega_{\text{at}}(z) = \omega_d(z) \sqrt{1 + 1.33 \sigma_R^2 \Lambda^{5/6}},$ where $\sigma_R^2$ is the Rytov variance for plane wave and $\Lambda$ is the Fresnel ratio of the beam at the receiver. The receiving efficiency is given by
\begin{equation}
    \xi_{at} = 1 - \exp{-\frac{2 r^2}{\omega_{\text{at}}^2}}. \label{app:xiAt}
\end{equation}
Next, for the satellite moving at a high speed a high-precision and high-bandwidth acquisition, pointing and tracking (APT) system generally consisting of coarse and fine tracking systems is required. A combination of closed-loop coarse tracking having a large field of view along with fine tracking having a small field of view is generally used. The pointing error induces a spot jitter with the instantaneous spot following a Rice intensity distribution. It was shown in \cite{toyoshima1998far} that the pointing efficiency is given by 
\begin{equation}
    \xi_p = 1 - \frac{\omega_{\text{at}}^2}{\omega_{\text{at}}^2 + 4 \hspace{3pt} \sigma_p^2}, \label{app:xiP}
\end{equation}
where $\sigma_p$ is the variance of the Gaussian pointing probability distribution.

\begin{figure}
    \includegraphics[scale=0.8]{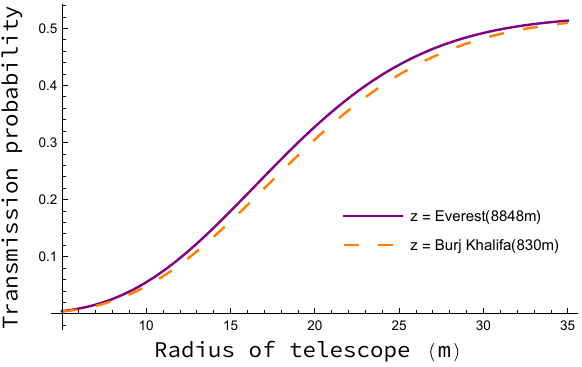}
    \caption{In this figure, we plot the variation of the transmission probability as a function of the radius of the telescope for different receiving station altitudes. For this, we have considered a $780$ nm source with $\omega_0 = 0.0021$ and quality factor $1$. \textcolor{black}{For this, We have set the parameters $z_R = 17.8, \sigma_R = 0.1, \Lambda = 0.1, \xi_r = 0.99, \xi_t = 0.99, \xi_{as} = 0.5$ and $\eta = 0.95$.} (Color online)}
    \label{fig:transmissionAtmosphereChanel1}
\end{figure}
Inserting Eq.~\eqref{app:xiD},~Eq.~\eqref{app:xiAt},~and~Eq.~\eqref{app:xiP} in Eq.~\eqref{app:xiEff} and imposing the condition that $\sigma_p = \textcolor{black}{\kappa}~ \omega_{at}$ we obtain 
\begin{eqnarray}
&&\xi_{\text{eff}} = \frac{\textcolor{black}{\kappa} ^2 \xi_{as}\xi_r\xi_t}{\textcolor{black}{\kappa} ^2+0.25} \bigg(1-\exp{-\frac{2 r^2 z_R^2}{w^2 (z^2+z_R^2)}}\bigg) \nonumber \\
&&\left(1- \exp{-\frac{2 r^2}{w^2 (1.33 \Lambda ^{5/6} \sigma_R^2+1) (\frac{z^2}{z_R^2}+1)}}\right) \label{eq:transmittanceAtmosphere}
\end{eqnarray}

Consider a 780 nm source with a beam waist radius $\omega_0$ of 0.0021 m and quality factor 1. The source has a Rayleigh length $(z_R) = 17.8$ m. Let the channel have Rytov variance $(\sigma_R) = 0.1,$ and the Fresnel ratio of the beam at the receiver end is $(\Lambda) = 0.1.$ Let the efficiency of the receiving unit be $(\xi_r) = 0.99,$ that of the transmitting source be $(\xi_t) = 0.99.$ Also let the probability of successful transmission after atmospheric absorption be $(\xi_{as}) = 0.5$. Assuming $\textcolor{black}{\kappa} = 0.95,$ we plot in Fig.~\ref{fig:transmissionAtmosphereChanel1} the transmission probability through the atmosphere as a function of the radius of the receiving telescope for different altitudes.

\section{Entanglement distribution across cities}

We plot in Fig.~\ref{fig:sampleGraph19} the variation of $\xi_{\text{avg}}$ as a function of the number of satellites in the network for different currently available single photon sources.

\begin{figure}
\centering
\includegraphics[scale=0.8]{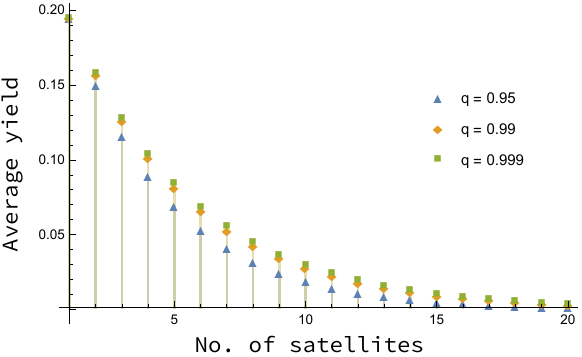}
    \label{fig:currentTechNetworkYield}
    \caption{In this figure, we plot the average yield $\xi_{\text{avg}}$ (see Eq.~\eqref{eq:avgYieldSatelliteNetwork}) as a function of the number of satellites in the network for different single photon source architectures. The single photon source architectures have the source efficiencies (a) $\eta_s = 0.84$ (SPDC \cite{altepeter2005phase}) (b) $\eta_s = 0.88$ (Atoms \cite{barros2009deterministic}) (c) $\eta_s = 0.97$ (Quantum Dots \cite{press2007photon}) (d) $\eta_s = 0.35$ (NV Center \cite{andersen2017ultrabright}) and (e) $\eta_s = 0.26$ (4 wave mixing \cite{smith2009photon}). \textcolor{black}{For this, we set the parameters $L = l_B + l_M = 10$ km$, s = 1, p = 0.1, \eta_e = 0.95, \eta_g = 0.5, \kappa_g = 0.5, \alpha = 1/22$ km$^{-1}$, $q = 1$.} (Color online)}
\label{fig:sampleGraph19}
\end{figure}

\section{Limitations of linear repeater networks} \label{sec:repeaterRelay}
Let us consider two distant parties, Alice and Bob requiring to share an entangled state using the standard linear DI key repeater chain as described in Sec.~\ref{sec:limitationNetwork}. Assuming that there are $n$ relay stations between Alice and Bob, they share an isotropic state
\begin{equation}
    \rho_{AB}^I(p(q^n\lambda^{n+1}),2) = q^n\lambda^{n+1}~\Psi^+_{AB} + \frac{1}{4} (1 - q^n\lambda^{n+1}) \mathbbm{1}_{AB}
\end{equation}
of visibility $q^n\lambda^{n+1}$ where $q$ is the success probability of the standard Bell measurement performed by the relay stations. In the following propositions, we present the bound on the number of virtual nodes in the network such that the state $\rho_{AB}^I(p(q^n\lambda^{n+1}),2)$ (a) is useful for teleportation, (b) can violate the Bell-CHSH inequality and (c) is entangled. 
\begin{proposition}
The state $\rho_{AB}^I(p(q^n\lambda^{n+1}),2)$ can be used to perform teleportation protocol when 
\begin{equation}
    n < \floor[\Bigg]{\frac{\log(3\lambda)}{\log(1/(q\lambda))}}.
\end{equation}
\end{proposition}
\begin{proof}
Let us have $u_k$ as the eigenvalues of the matrix $T^{\dagger}T$ where the $T$ matrix is formed by the elements $t_{nm} = \Tr[\rho_{AB}^I(p(q^n\lambda^{n+1}),2) \hspace{3pt} \sigma_n \otimes \sigma_m]$ where $\sigma_j$ denotes the Pauli matrices. We then define the quantity $N(\rho_{AB}) = \sum_{k=1}^{3} \sqrt{u_k}$. The state $\rho_{AB}^I(p(q^n\lambda^{n+1}),2)$ is useful for teleportation for values of $N(\rho_{AB})$ greater than $1$ \cite{horodecki1996teleportation}. This then implies $n < \floor[\Bigg]{\frac{\log(3\lambda)}{\log(1/(q\lambda))}}.$
\end{proof}
We plot in Fig.~\ref{fig:allowedNodesTeleportation} the maximum number of repeater stations that can be allowed for a given value of $\lambda$ and setting $q \in \{0.625~\text{\cite{bayerbach2022bell}}, 0.95, 0.99\}$ to implement a teleportation protocol successfully. 
\begin{figure}
    \centering
    \includegraphics[scale=0.89]{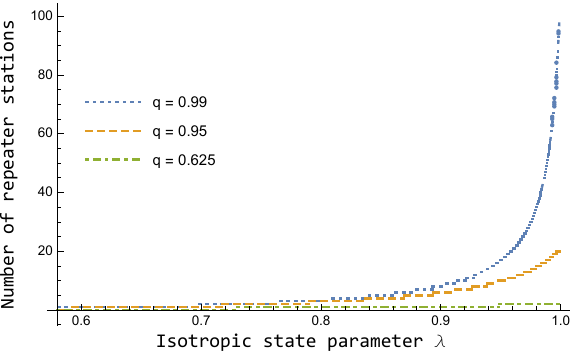}
    \caption{In this figure, we plot the maximum allowed number of relay stations between the end nodes as a function of $\lambda$, considering values of $q \in \{0.625, 0.95, 0.99\}$ such that the end nodes can implement a teleportation protocol. \textcolor{black}{We assume all adjacent nodes in the repeater chain share two-qubit isotropic states with an identical parameter $\lambda$.} (Color online)}
    \label{fig:allowedNodesTeleportation}
\end{figure}

\begin{proposition}
The state $\rho_{AB}^I(p(q^n\lambda^{n+1}),2)$ can violate the Bell-CHSH inequality when
\begin{equation}
    n < \floor{\frac{\log (\sqrt{2} \lambda)}{\log (1/(q\lambda))}}.
\end{equation}
\end{proposition}
\begin{proof}
Let us have $u_i,u_j$ be the two largest eigenvalues of the matrix $T^{\dagger}T$ where the $T$ matrix is formed by the elements $t_{nm} = \Tr[\rho_{AB}^I(p(q^n\lambda^{n+1}),2) \hspace{3pt} \sigma_n \otimes \sigma_m]$ where $\sigma_j$ denotes the Pauli matrices. We then define the quantity $M(\rho_{AB}) = u_i + u_j$. The state $\rho_{AB}^I(p(q^n\lambda^{n+1}),2)$ is Bell-CHSH nonlocal for values of $M(\rho_{AB})$ greater than $1$ \cite{horodecki1995violating}.
This then implies $n < \floor{\frac{\log (\sqrt{2} \lambda)}{\log (1/(q\lambda))}}$.
\end{proof}
We plot in Fig.~\ref{fig:sampleGraph3} the maximum number of repeater stations that can be allowed for a given value of $\lambda$ and setting $q \in \{0.625~\text{\cite{bayerbach2022bell}}, 0.95, 0.99\}$ such that the state shared by the end nodes can violate the Bell-CHSH inequality. 
\begin{figure}
    \centering
    \includegraphics[scale=0.89]{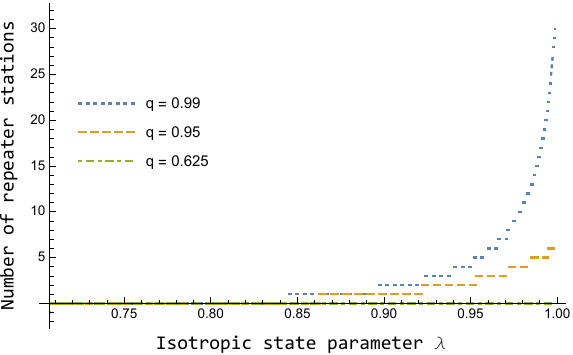}
    \caption{In this figure, we plot the maximum allowed number of relay stations between the end nodes as a function of $\lambda$, considering values of $q \in \{0.625, 0.95, 0.99\}$ such that the end nodes can perform Bell-CHSH violation experiment. \textcolor{black}{We assume all adjacent nodes in the repeater chain share two-qubit isotropic states with an identical parameter $\lambda$.} (Color online)}
    \label{fig:sampleGraph3}
\end{figure}

\begin{proposition}
The state $\rho_{AB}^I(p(q^n\lambda^{n+1}),2)$ remains entangled when
\begin{equation}
    n < \floor[\Bigg]{\frac{\log (\lambda/(\frac{2}{\sqrt{3}}-1))}{\log(1/(q\lambda))}}. \label{eq:repeaterRelayEntanglement}
\end{equation} 
\end{proposition}
\begin{proof}
The concurrence of the state $\rho_{AB}^I(p(q^n\lambda^{n+1}),2)$ is given by $\max\{0,\lambda^e_1 - \lambda^e_2 - \lambda^e_3 - \lambda^e_3\}$ \cite{wootters1998entanglement} where the $\lambda^e_s$ are the square root of the eigenvalues of  $\rho_{AB}^I(p(q^n\lambda^{n+1}),2) (\rho_{AB}^I(p(q^n\lambda^{n+1}),2))_f$ in descending order. The spin flipped density matrix is given by $(\rho_{AB}^I(p(q^n\lambda^{n+1}),2))_f = (\sigma_y \otimes \sigma_y) \rho_{AB}^{I \ast}(p(q^n\lambda^{n+1}),2) (\sigma_y \otimes \sigma_y).$ The state is entangled when the concurrence is greater than zero. This requires $n < \floor[\Big]{\frac{\log (\lambda/(\frac{2}{\sqrt{3}}-1))}{\log(1/(q\lambda))}}.$ 
\end{proof} 
We plot in Fig.~\ref{fig:sampleGraph4} the maximum number of repeater stations that can be allowed for a given value of $\lambda$ and setting $q \in \{0.625, 0.95, 0.99\}$ such that the end nodes can share an entangled state.
\begin{figure}
    \centering
    \includegraphics[scale=0.89]{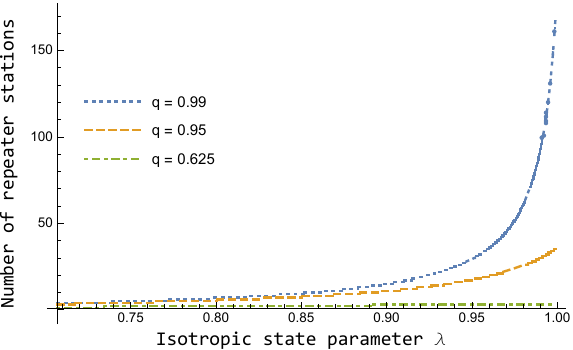}
    \caption{In this figure, we plot the maximum allowed number of relay stations between the end nodes as a function of $\lambda$, considering values of $q \in \{0.625, 0.95, 0.99\}$ such that the end nodes share an entangled state. \textcolor{black}{We assume all adjacent nodes in the repeater chain share two-qubit isotropic states with an identical parameter $\lambda$.} (Color online)}
    \label{fig:sampleGraph4}
\end{figure}
It is observed that the number of allowed repeater stations in the network depends on the information processing task that the network is executing. The number of allowed repeater stations increases with the increase in $\lambda$.

\section{Analysis of time-varying quantum networks} \label{app:timeEvolving}
In general, the network parameters i.e., the number of nodes, edges, and edge weights may change with time. Let us denote such a time-varying network as $\mathscr{N}(G(\mathbb{V}(t),\mathbb{E}(t)))$. We say that two nodes $v_i$ and $v_j$ of a time-varying network are connected in the time interval $[t_1,t_2]$ if $\exists~e_{ij} \in \mathbb{E}(t)$ with $p_{ij} \geq p_\ast~\forall t\in [t_1,t_2]$. In the following example, we present the variation of link sparsity with time for the time-varying network shown in Fig.~\ref{fig:temporalnetwork}. 
\begin{example}
Consider a network $\mathscr{N}(G(\mathbb{V}(t),\mathbb{E}(t)))$ whose vertices and edges are evolving in time as shown in Fig.~\ref{fig:temporalnetwork}. We assume that for $e_{ij} \in \mathbb{E}(t)$,
\begin{align} \label{eq:timeVaryingNetwork}
    p_{ij}(t+1)\coloneqq \begin{cases}
      w\mathrm{e}^{-kt}~p_{ij}(t)  & \text{for $w\mathrm{e}^{-kt}~p_{ij}(t) > p_\ast$},\\
      0 & \text{otherwise},
    \end{cases}       
\end{align}
where $w = 0.9, p_\ast = 0.05, k = 0.3$ sec$^{-1}$ and $t \in [1,4]$. 
\begin{figure}[H]
    \includegraphics[scale=0.35]{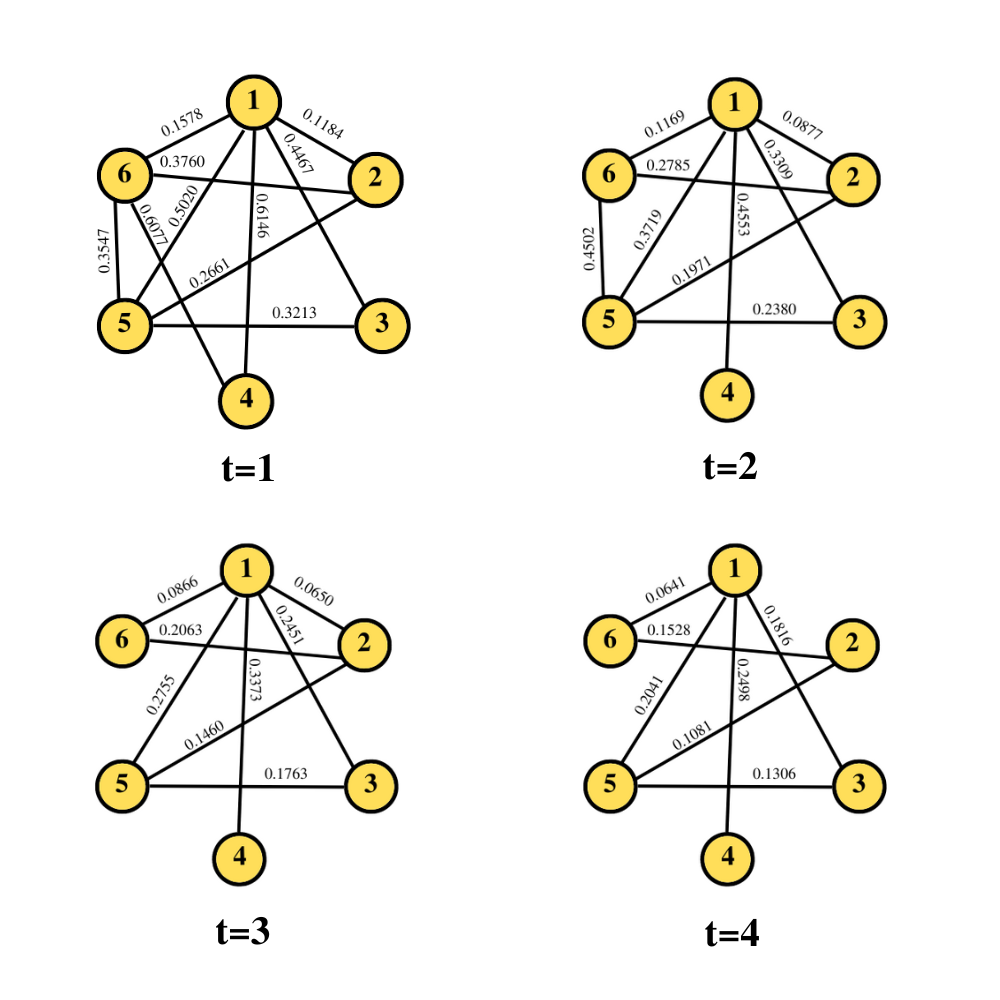}
    \caption{We present a time varying mesh network with $6$ nodes as a series of $4$ static graphs. We consider the network topology at times $t \in \{1,2,3,4\}.$ In this network, the edges $e_{ij}$ connecting nodes $v_i$ and $v_j$ denote the success probabilities $p_{ij}$ of transferring some resource between the $v_i$ and $v_j$. The success probability $p_{ij}$ evolves in time following Eq.~\eqref{eq:timeVaryingNetwork}. (Color online)}
    \label{fig:temporalnetwork}
\end{figure}
For such a network, we have the variation of the link sparsity with time as 
\begin{table}[H]
\centering
\begin{tabular}{ l || l l l l} 
 \hline
 $t$ (sec) & 1 & 2 & 3 & 4 \\ [1 ex]
 \hline \hline
 $\Upsilon(\mathscr{N})$ & 0.4445 & 0.5 & 0.5556 & 0.6112 \\ [1 ex]
 \hline
 \end{tabular}
 \caption{Link sparsity for time-evolving graph} \label{table:timeVaryingNetwork}
\end{table}
We observe from Table~\ref{table:timeVaryingNetwork} that for the network shown in Fig~\ref{fig:temporalnetwork}, the link sparsity increases with time and the network becomes less robust.  
\end{example}

\bibliography{paper}{}

\end{document}